\documentclass[fleqn,10pt]{wlscirep}

\usepackage[sort&compress,square,comma,numbers]{natbib}

\usepackage[utf8]{inputenc}
\usepackage{bbm, amsfonts, url}

\usepackage{subcaption}

\usepackage[export]{adjustbox}
\usepackage{graphicx}
\usepackage{amsmath,amssymb}
\usepackage{mathtools}
\usepackage{amsfonts}
\usepackage{bm}
\usepackage{hyperref}
\usepackage{kbordermatrix}
\usepackage{amsthm}
\usepackage{sgame, tikz}
\usepackage{color}
\usepackage{multirow,array}
\usepackage{pifont}
\usepackage[outline]{contour} 
\contourlength{1pt} 
\usepackage{thmtools,thm-restate} 
\usepackage{natbib}

\usepackage{bbold}
\usepackage[capitalise]{cleveref}
\crefname{equation}{}{} 

\newif\ifKeepSCII

\newcommand{\cmark}{\ding{51}}%
\newcommand{\xmark}{\ding{55}}%
\newcommand{\focalAgent}{\tau}
\newcommand{\residentAgent}{\sigma}
\newcommand*\circled[1]{\tikz[baseline=(char.base)]{
            \node[shape=circle,draw,inner sep=0.7pt] (char) {#1};}}

\newcommand{\ProofInSM}[1]{Refer to the Supplementary Material \cref{#1} for the proof.}

\theoremstyle{remark}
\newsavebox\mybox
\newcommand{\real}{\mathbb{R}}
\def\E{\ensuremath{\mathbf{E}}}
\colorlet{MyBlue}{red!10!green!20!blue!100}

 \title{$\alpha$-Rank: Multi-Agent Evaluation by Evolution}

\author[1]{Shayegan Omidshafiei$^{*}$}
\author[3]{Christos Papadimitriou$^{*}$}
\author[2]{Georgios Piliouras$^{*}$}
\author[1]{Karl Tuyls$^{*}$}
\author[1]{Mark Rowland}
\author[1]{Jean-Baptiste Lespiau}
\author[1]{Wojciech M. Czarnecki}
\author[1]{Marc Lanctot}
\author[1]{Julien Perolat}
\author[1]{Remi Munos}
\affil[1]{DeepMind, 6 Pancras Square, London, UK}
\affil[2]{Singapore University of Technology and Design, Singapore}
\affil[3]{Columbia University, New York, USA}
\affil[*]{\small{Equal contributors, ordered alphabetically. Corresponding author: Karl Tuyls $<$karltuyls@google.com$>$.}}

\newtheorem{theorem}{Theorem}[subsection]

\newtheorem{lemma}[theorem]{Lemma}
\newtheorem{definition}[theorem]{Definition}
\newtheorem{property}[theorem]{Property}

\crefname{property}{Property}{Properties}
\Crefname{property}{Property}{Properties}
\crefname{definition}{Definition}{Definitions}
\Crefname{definition}{Definition}{Definitions}

\newcommand{\mcc}{MCC}

\begin{abstract}

We introduce \emph{$\alpha$-Rank}, a principled evolutionary dynamics methodology, for the \emph{evaluation} and \emph{ranking} of agents in large-scale multi-agent interactions, grounded in a novel dynamical game-theoretic solution concept called \emph{Markov-Conley chains} (MCCs). 
The approach leverages continuous-time and discrete-time evolutionary dynamical systems applied to empirical games, and scales tractably in the number of agents, in the type of interactions (beyond dyadic), and the type of empirical games (symmetric and asymmetric). 
Current models are fundamentally limited in one or more of these dimensions, and are not guaranteed to converge to the desired game-theoretic solution concept (typically the Nash equilibrium).
$\alpha$-Rank automatically provides a ranking over the set of agents under evaluation and provides insights into their strengths, weaknesses, and long-term dynamics in terms of basins of attraction and sink components. 
This is a direct consequence of the correspondence we establish to the dynamical MCC solution concept when the underlying evolutionary model's ranking-intensity parameter, $\alpha$, is chosen to be large, which exactly forms the basis of $\alpha$-Rank.
In contrast to the Nash equilibrium, which is a static solution concept based solely on fixed points, MCCs are a dynamical solution concept based on the Markov chain formalism, Conley's Fundamental Theorem of Dynamical Systems, and the core ingredients of dynamical systems: fixed points, recurrent sets, periodic orbits, and limit cycles.
Our $\alpha$-Rank method runs in polynomial time with respect to the total number of pure strategy profiles, whereas computing a Nash equilibrium for a general-sum game is known to be intractable.
We introduce mathematical proofs that not only provide an overarching and unifying perspective of existing continuous- and discrete-time evolutionary evaluation models, but also reveal the formal underpinnings of the $\alpha$-Rank methodology. 
We illustrate the method in canonical games and empirically validate it in several domains, including AlphaGo, AlphaZero, MuJoCo Soccer, and Poker.
\end{abstract}

\begin{document}
\flushbottom
\maketitle
%
%
\thispagestyle{empty}


\section{Introduction}\label{sec:intro}

This paper introduces a principled, practical, and descriptive methodology, which we call \emph{$\alpha$-Rank}. 
$\alpha$-Rank enables evaluation and ranking of agents in large-scale multi-agent settings, and
is grounded in a new game-theoretic solution concept, called Markov-Conley chains (MCCs), which captures the dynamics of multi-agent interactions. 
While much progress has been made in learning for games such as Go \cite{DSilverHMGSDSAPL16,silver:17} and Chess \cite{silver2018general}, computational gains are now enabling algorithmic innovations in domains of significantly higher complexity, such as Poker \cite{Morav} and MuJoCo soccer \cite{liu2018emergent} where ranking of agents is much more intricate than in classical simple matrix games. 
With multi-agent learning domains of interest becoming increasingly more complex, we need methods for evaluation and ranking that are both comprehensive and theoretically well-grounded.

Evaluation of agents in a multi-agent context is a hard problem due to several complexity factors: strategy and action spaces of players quickly explode (e.g., multi-robot systems), models need to be able to deal with intransitive behaviors (e.g., cyclical best-responses in Rock-Paper-Scissors, but at a much higher scale), the number of agents can be large in the most interesting applications (e.g., Poker), types of interactions between agents may be complex (e.g., MuJoCo soccer), and payoffs for agents may be asymmetric (e.g., a board-game such as Scotland Yard).

This evaluation problem has been studied in Empirical Game Theory using the concept of empirical games or meta-games, and the convergence of their dynamics to Nash equilibria \cite{Walsh02,Wellman06,Tuyls18,TuylsSym}. 
A meta-game is an abstraction of the underlying game, which considers meta-strategies rather than primitive actions \cite{Walsh02,Tuyls18}. In the Go domain, for example, meta-strategies may correspond to different AlphaGo agents (e.g., each meta-strategy is an agent using a set of specific training hyperparameters, policy representations, and so on). 
The players of the meta-game now have a choice between these different agents (henceforth synonymous with meta-strategies), and payoffs in the meta-game are calculated corresponding to the win/loss ratio of these agents against each other over many rounds of the full game of Go. Meta-games, therefore, enable us to investigate the strengths and weaknesses of these agents using game-theoretic evaluation techniques.

\begin{figure}[t!]
    \centering
    \includegraphics[width=\linewidth,page=6]{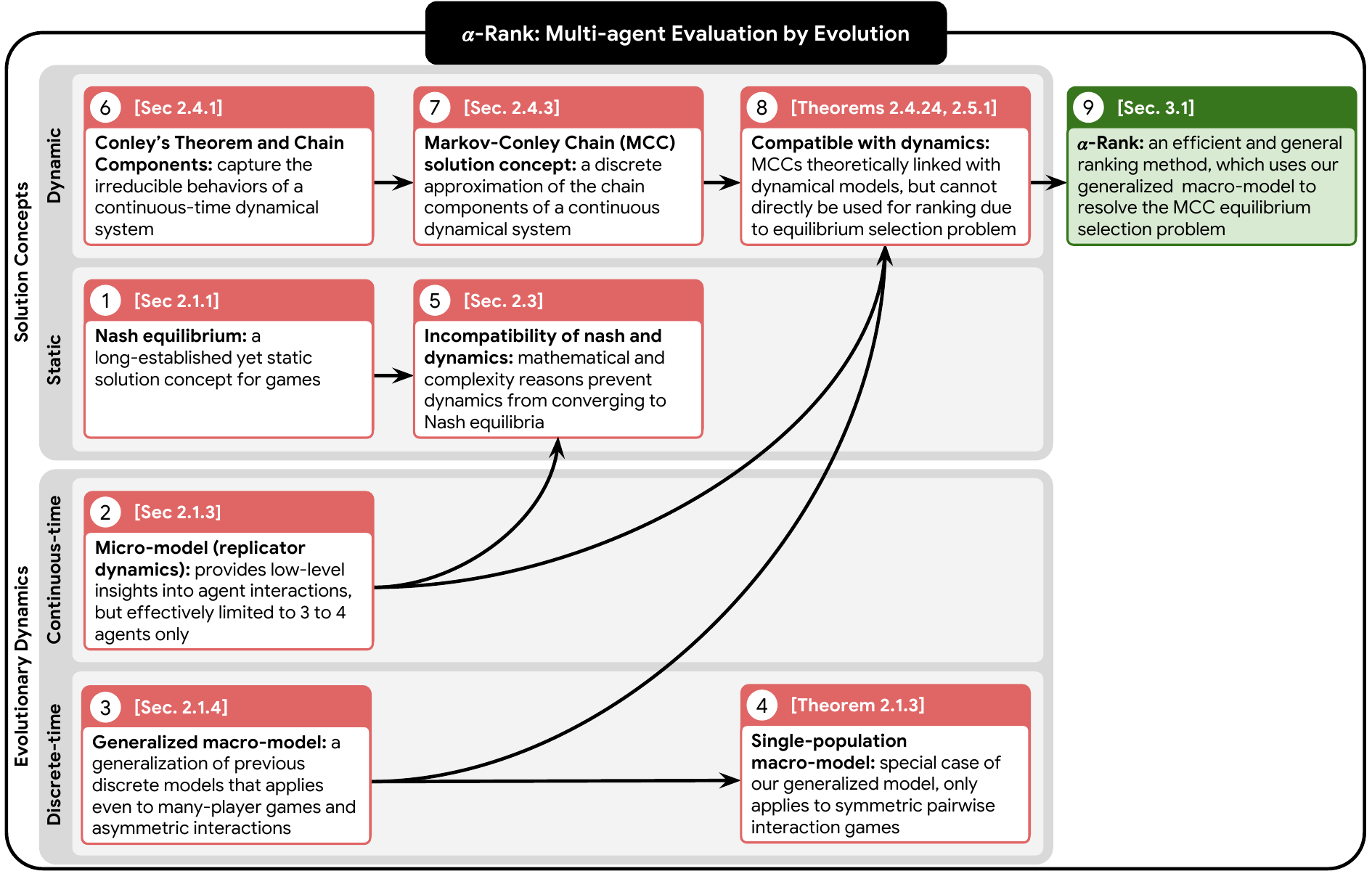}
    \caption{Paper at a glance. Numerical ordering of the concept boxes corresponds to the paper flow, with sections and/or theorems indicated where applicable. The methods and ideas used herein may be classified broadly as either game-theoretic \emph{solution concepts} (namely, static or dynamic) and \emph{evolutionary dynamics} concepts (namely, continuous- or discrete-time). 
    The insights gained by analyzing existing concepts and developing new theoretical results carves a pathway to the novel combination of our general multi-agent evaluation method, $\alpha$-Rank, and our game-theoretic solution concept, Markov-Conley Chains.
    }
    \label{fig:overview}
\end{figure}

Existing meta-game analysis techniques, however, are still limited in a number of ways: 
either a low number of players or a low number of agents (i.e., meta-strategies) may be analyzed \cite{Walsh02,TuylsP07,Tuyls18,TuylsSym}.
Specifically, on the one hand continuous-time meta-game evaluation models, using replicator dynamics from Evolutionary Game Theory \cite{Zeeman80,Zeeman81,Weibull97,Hofbauer96,Gintis09}, are deployed to capture the micro-dynamics of interacting agents. 
These approaches study and visualize basins of attraction and equilibria of interacting agents, but are limited as they can only be feasibly applied to games involving few agents, exploding in complexity in the case of large and asymmetric games.
On the other hand, existing discrete-time meta-game evaluation models (e.g., \cite{traulsen2005coevolutionary,Traulsen06a,Santos11,Segbroek12,veller2016finite}) capture the macro-dynamics of interacting agents, but involve a large number of evolutionary parameters and are not yet grounded in a game-theoretic solution concept. 

To further compound these issues, using the Nash equilibrium as a solution concept for meta-game evaluation in these dynamical models is in many ways problematic: 
first, computing a Nash equilibrium is computationally difficult \cite{VonStengel20021723,Daskalakis06thecomplexity}; 
second, there are intractable equilibrium selection issues even if Nash equilibria can be computed \cite{Harsan88,Avis10,goldberg2013complexity}; 
finally, there is an inherent incompatibility in the sense that it is not guaranteed that dynamical systems will converge to a Nash equilibrium \cite{Papadimitriou:2016:NEC:2840728.2840757}, or, in fact, to any fixed point.
However, instead of taking this as a disappointing flaw of dynamical systems models, we see it as an opportunity to look for a novel solution concept that does not have the same limitations as Nash in relation to these dynamical systems. 
Specifically, exactly as J.~Nash used one of the most advanced topological results of his time, i.e., Kakutani's fixed point theorem \cite{Kakutani41}, as the basis for the Nash solution concept, in the present work, we employ Conley's Fundamental Theorem of Dynamical Systems \cite{conley1978isolated} and propose the solution concept of Markov-Conley chains (MCCs).
Intuitively, Nash is a static solution concept solely based on fixed points.
MCCs, by contrast, are a dynamic solution concept based not only on fixed points, but also on recurrent sets, periodic orbits, and limit cycles, which are fundamental ingredients of dynamical systems.
The key advantages are that MCCs comprehensively capture the long-term behaviors of our (inherently dynamical) evolutionary systems, and our associated $\alpha$-Rank method runs in polynomial time with respect to the total number of pure strategy profiles (whereas computing a Nash equilibrium for a general-sum game is PPAD-complete \cite{Daskalakis06thecomplexity}).



\paragraph{Main contributions: $\alpha$-Rank and MCCs}
While MCCs do not immediately address the equilibrium selection problem, we show that by introducing a perturbed variant that corresponds to a generalized multi-population discrete-time dynamical model, the underlying Markov chain containing them becomes irreducible and yields a unique stationary distribution.
The ordering of the strategies of agents in this distribution gives rise to our $\alpha$-Rank methodology.
$\alpha$-Rank provides a summary of the asymptotic evolutionary rankings of agents in the sense of the time spent by interacting populations playing them, yielding insights into their evolutionary strengths.
It both automatically produces a ranking over agents favored by the evolutionary dynamics and filters out transient agents (i.e., agents that go extinct in the long-term evolutionary interactions).

\paragraph{Paper Overview}
Due to the interconnected nature of the concepts discussed herein, we provide in \cref{fig:overview} an overview of the paper that highlights the relationships between them.
Specifically, the paper is structured as follows:
we first provide a review of preliminary game-theoretic concepts, including the Nash equilibrium (box \circled{1} in \cref{fig:overview}), which is a long-standing yet static solution concept. We then overview the replicator dynamics micro-model (\circled{2}), which provides low-level insights into agent interactions but is limited in the sense that it can only feasibly be used for evaluating three to four agents.
We then introduce a generalized evolutionary macro-model (\circled{3}) that extends previous single-population discrete-time models (\circled{4}) and (as later shown) plays an integral role in our $\alpha$-Rank method.
We then narrow our focus on a particular evolutionary macro-model (\circled{3}) that generalizes single-population discrete-time models (\circled{4}) and (as later shown) plays an integral role in our $\alpha$-Rank method.
Next, we highlight a fundamental incompatibility of the dynamical systems and the Nash solution concept (\circled{5}), establishing fundamental reasons that prevent dynamics from converging to Nash.
This limitation motivates us to investigate a novel solution concept, using Conley's Fundamental Theorem of Dynamical Systems as a foundation (\circled{6}).

Conley's Theorem leads us to the topological concept of \emph{chain components}, which do capture the irreducible long-term behaviors of a \emph{continuous} dynamical system, but are unfortunately difficult to analyze due to the lack of an exact characterization of their geometry and the behavior of the dynamics inside them.
We, therefore, introduce a discrete approximation of these limiting dynamics that is more feasible to analyze: our so-called Markov-Conley chains solution concept (\circled{7}).
While we show that Markov-Conley chains share a close theoretical relationship with both discrete-time and continuous-time dynamical models (\circled{8}), they unfortunately suffer from an equilibrium selection problem and thus cannot directly be used for computing multi-agent rankings.
To address this, we introduced a perturbed version of Markov-Conley chains that resolves the equilibrium selection issues and yields our $\alpha$-Rank evaluation method (\circled{9}). 
 $\alpha$-Rank computes both a ranking and assigns scores to agents using this perturbed model.
We show that this perturbed model corresponds directly to the generalized macro-model under a particular setting of the latter's so-called \emph{ranking-intensity parameter} $\alpha$.
$\alpha$-Rank not only captures the dynamic behaviors of interacting agents, but is also more tractable to compute than Nash for general games.
We validate our methodology empirically by providing ranking analysis on datasets involving interactions of state-of-the-art agents including AlphaGo \cite{DSilverHMGSDSAPL16}, AlphaZero \cite{silver2018general}, MuJoCo Soccer \cite{liu2018emergent}, and Poker \cite{Lanctot17}, and also provide scalability properties and theoretical guarantees for the overall ranking methodology.

\section{Preliminaries and Methods}\label{sec:prelim_and_methods}
In this section, we concisely outline the game-theoretic concepts and methods necessary to understand the remainder of the paper. 
For a detailed discussion of the concepts we refer the reader to \cite{Weibull97, Hofbauer98, Cressman03,Tuyls18}.
We also introduce a novel game-theoretic concept, Markov-Conley chains, which we use to theoretically ground our results in.


\subsection{Game Theoretic Concepts}
\subsubsection{Normal Form Games}\label{sec:nfg}
A $K$-wise interaction Normal Form Game (NFG) $G$ is defined as $(K,\prod_{k=1}^K S^k, \prod_{k=1}^K M^k)$, where each player $k \in \{1,\ldots,K\}$ chooses a strategy $s^k$ from its strategy set $S^k$ and receives a payoff $M^k: \prod_{i=1}^K S^i \to \real$. 
We henceforth denote the joint strategy space and payoffs, respectively, as  $\prod_k S^k$ and $\prod_k M^k$. 
We denote the strategy profile of all players by $s=(s^1,\ldots,s^K) \in \prod_k S^k$, the strategy profile of all players except $k$ by $s^{-k}$, and the payoff profile by $(M^1(s^1,s^{-1}),\ldots, M^K(s^K,s^{-K}))$.
An NFG is symmetric if the following two conditions hold: 
first, all players have the same strategy sets (i.e., $ \forall k,l  \, S^k=S^l$); 
second, if a permutation is applied to the strategy profile, the payoff profile is permuted accordingly.
The game is asymmetric if one or both of these conditions do not hold.
Note that in a $2$-player ($K=2$) NFG the payoffs for both players ($M$ above) are typically represented by a bi-matrix $(A,B)$, which gives the payoff for the row player in $A$, and the payoff for the column player in $B$. If $S^1=S^2$ and $A=B^T$, then this $2$-player game is symmetric.

Naturally the definitions of strategy and payoff can be extended in the usual multilinear 
fashion to allow for randomized (mixed) strategies. In that case,
we usually overload notation in the following manner: if ${{x}}^k$ is a mixed strategy for each player $k$ and $x^{-k}$ the mixed profile excluding that player, then we denote by $M^k(x^k,x^{-k})$ the expected payoff of player $k$, $\E_{s^k\sim x^k, s^{-k}\sim x^{-k} }[M^k({s^k,s^{-k}})]$. Given these preliminaries, we are now ready to define the Nash equilibrium concept:
\begin{definition}[Nash equilibrium]
A mixed strategy profile $x=(x^1,\ldots,x^K)$ is a Nash equilibrium if for all players $k$: $\max_{x'^k}M^k(x'^k,x^{-k})=M^k(x^k,x^{-k})$.
\end{definition}
Intuitively, a strategy profile $x$ is a Nash equilibrium of the NFG if no player has an incentive to unilaterally deviate from its current strategy.


\subsubsection{Meta-games}\label{sec:metagames}
A meta-game (or an empirical game) is a simplified model of an underlying multi-agent system (e.g., an auction, a real-time strategy game, or a robot football match), which considers meta-strategies or `styles of play' of agents, rather than the full set of primitive strategies available in the underlying game \cite{Walsh02,Wellman06,Tuyls18}.
In this paper, the meta-strategies considered are learning agents (e.g., different variants of AlphaGo agents, as exemplified in \cref{sec:intro}). Thus, we henceforth refer to meta-games and meta-strategies, respectively, as `games' and `agents' when the context is clear.
For example, in AlphaGo, styles of play may be characterized by a set of agents $\{AG(r), AG(v), AG(p)\}$, where \emph{AG} stands for the algorithm and indexes $r$, $v$, and $p$ stand for \emph{rollouts}, \emph{value networks}, and \emph{policy networks}, respectively, that lead to different play styles. 
The corresponding meta-payoffs quantify the outcomes when players play profiles over the set of agents (e.g., the empirical win rates of the agents when played against one another). These payoffs can be calculated from available data of the agents' interactions in the real multi-agent systems (e.g., wins/losses in the game of Go), or they can be computed from simulations. 
The question of how many such interactions are necessary to have a good approximation of the true underlying meta-game is discussed in \cite{Tuyls18}. 
A meta-game itself is an NFG and can, thus, leverage the game-theoretic toolkit to evaluate agent interactions at a high level of abstraction.

\subsubsection{Micro-model: Replicator Dynamics}\label{sec:rd}
\emph{Dynamical systems} is a powerful mathematical framework for specifying the time dependence of the players' behavior (see the Supplementary Material for a brief introduction).

For instance, in a two-player asymmetric meta-game represented as an NFG $(2, S^1\times S^2, M=(A,B))$, the evolution of players' strategy profiles under the replicator dynamics \cite{Taylor78,Schuster1983533} is given by,
\begin{equation}
    \dot{x}_{i} = x_{i}((Ay)_{i} - x^TAy) \qquad     \dot{y}_{j} = y_{j}((x^TB)_{j} - x^TBy) \qquad \forall (i, j)\in S^1 \times S^2,
\end{equation}
where $x_i$ and $y_{j}$ are, respectively, the proportions of strategies $i \in S^1$ and $j \in S^2$ in two infinitely-sized populations, each corresponding to a player.
This system of coupled differential equations models the temporal dynamics of the populations' strategy profiles when they interact, and can be extended readily to the general $K$-wise interaction case (see Supplementary Material \cref{sec:rd_multipop} for more details).

The replicator dynamics provide useful insights into the micro-dynamical characteristics of games, revealing strategy flows, basins of attraction, and equilibria \cite{BloembergenTHK15} when visualized on a trajectory plot over the strategy simplex (e.g, \cref{fig:rd_examples}).
The accessibility of these insights, however, becomes limited for games involving large strategy spaces and many-player interactions.
For instance, trajectory plots may be visualized only for subsets of three or four strategies in a game, and are complex to analyze for multi-population games due to the inherently-coupled nature of the trajectories.
While methods for scalable empirical game-theoretic analysis of games have been recently introduced, they are still limited to two-population games \cite{TuylsSym,Tuyls18}.

\subsubsection{Macro-model: Discrete-time Dynamics}\label{sec:multipop_model}

\begin{figure}[t!]
    \centering
    \begin{subtable}[b]{\textwidth}
        \begin{tabular}{p{0.26\textwidth}p{0.67\textwidth}}
        \toprule
        Concept & Meaning \\
        \midrule
        $K$-wise meta-game & An NFG with $K$ player slots. \\
        Strategy & The agents under evaluation (e.g., variants of AlphaGo agents) in the meta-game.  \\
        Individual & A population member, playing a strategy and assigned to a slot in the meta-game. \\
        Population & A finite set of individuals.  \\
        Player & An individual that participates in  the meta-game under consideration.   \\
        Monomorphic Population & A finite set of individuals, playing the same strategy.  \\
        Monomorphic Population Profile & A set of monomorphic populations.\\
        Focal Population & A previously-monomorphic population wherein a rare mutation has appeared.  \\
        \bottomrule
        \end{tabular}
        \caption{}
        \label{tab:concepts}
    \end{subtable}\\
    \begin{subfigure}[t]{0.48\textwidth} 
        \includegraphics[width=\linewidth,page=3]{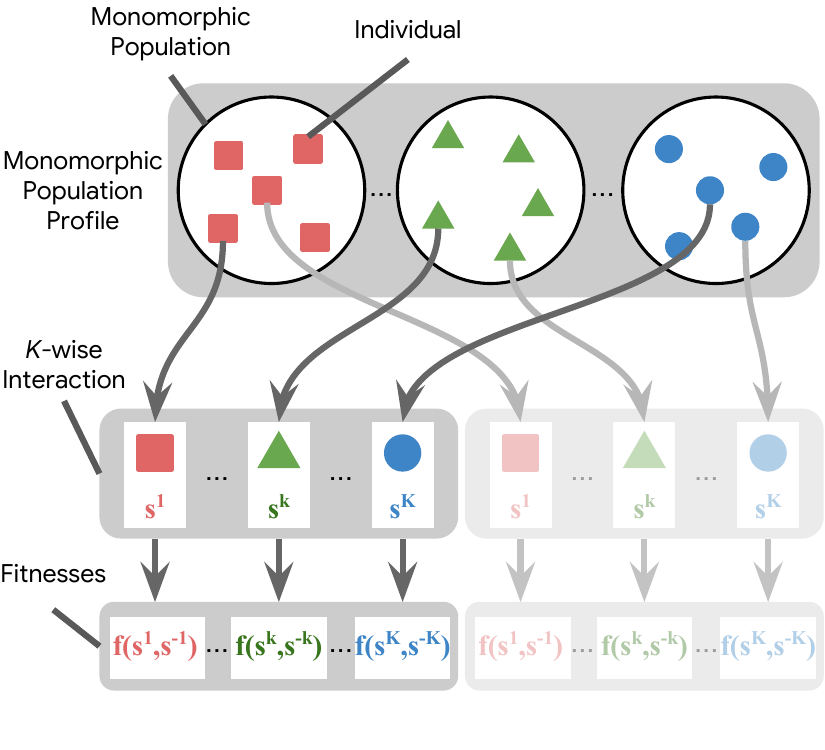}
        \caption{}
        \label{fig:MacroModelOverview_1}
    \end{subfigure}%
    \hfill
    \begin{subfigure}[t]{0.48\textwidth}
        \includegraphics[width=\linewidth,page=4]{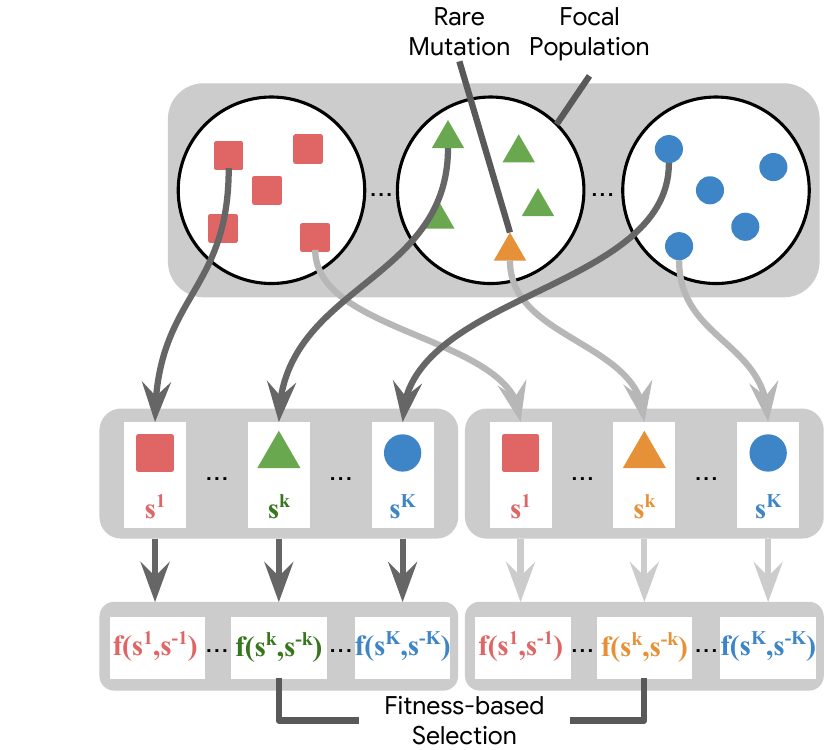}
        \caption{}
        \label{fig:MacroModelOverview_2}
    \end{subfigure}\\
    \begin{subfigure}[t]{0.48\textwidth}
        \includegraphics[width=\linewidth,page=5,trim={0cm 7.5cm 0cm 0.5cm},clip]{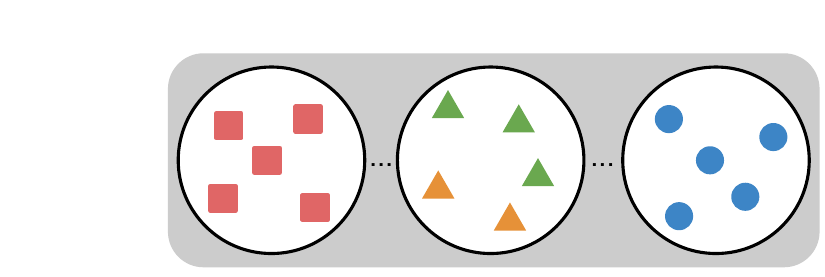}
        \caption{}
        \label{fig:MacroModelOverview_3}
    \end{subfigure}%
    \hfill
    \begin{subfigure}[t]{0.48\textwidth}
        \includegraphics[width=\linewidth,page=6,trim={0cm 7.5cm 0cm 0.5cm},clip]{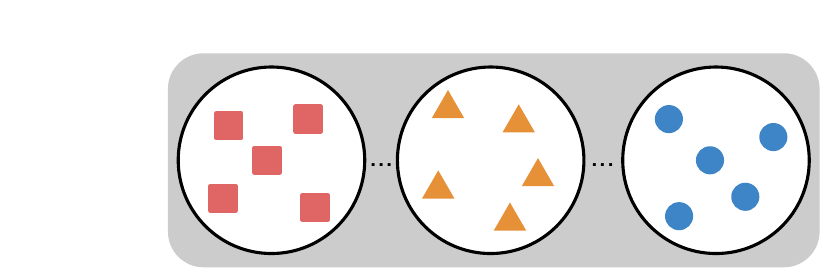}
        \caption{}
        \label{fig:MacroModelOverview_4}
    \end{subfigure}
    \caption{Overview of the discrete-time macro-model.
    (\subref{tab:concepts}) Evolutionary concepts terminology.
    (\subref{fig:MacroModelOverview_1}) We have a set of individuals in each population $k$, each of which is programmed to play a strategy from set $S^k$.
    Under the mutation rate $\mu \rightarrow 0$ assumption, at most one population is not monomorphic at any time. 
    Each individual in a $K$-wise interaction game has a corresponding fitness $f^k(s^k, s^{-k})$ dependent on its identity $k$, its strategy $s^k$, and the strategy profile $s^{-k}$ of the other players.
    (\subref{fig:MacroModelOverview_2}) Let the focal population denote a population $k$ wherein a rare mutant strategy appears. 
    At each timestep, we randomly sample two individuals in population $k$; 
    the strategy of the first individual is updated by either probabilistically copying the strategy of the second individual, mutating with a very small probability to a random strategy, or sticking with its own strategy.
    (\subref{fig:MacroModelOverview_3}) Individual in the focal population copies the mutant strategy.
    (\subref{fig:MacroModelOverview_4}) The mutant propagates in the focal population, yielding a new monomorphic population profile.
    }
    \label{fig:MacroModelOverview}
\end{figure}

This section presents our main evolutionary dynamics model, which extends previous single-population discrete-time models and is later shown to play an integral role in our $\alpha$-Rank method and can also be seen as an instantiation of the framework introduced in \cite{veller2016finite}.

A promising alternative to using the continuous-time replicator dynamics for evaluation is to consider discrete-time finite-population dynamics. 
As later demonstrated, an important advantage of the discrete-time dynamics is that they are not limited to only three or four strategies (i.e., the agents under evaluation) as in the continuous-time case. 
Even though we lose the micro-dynamical details of the strategy simplex, this discrete-time macro-dynamical model, in which we observe the flows over the edges of the high-dimensional simplex, still provides useful insights into the overall system dynamics.

To conduct this discrete-time analysis, we consider a selection-mutation process but with a \emph{very small mutation rate} (following the small mutation rate theorem, see \citep{fudenberg2006imitation}).  
Before elaborating on the details we specify a number of important concepts used in the description below and clarify their respective meanings in \cref{tab:concepts}.
Let a \emph{monomorphic population} denote a population wherein all individuals play identical strategies, and a \emph{monomorphic population profile} is a set of monomorphic populations, where each population may be playing a different strategy (see \cref{fig:MacroModelOverview_1}).
Our general idea is to capture the overall dynamics by defining a \emph{Markov chain} over states that correspond to monomorphic population profiles. We can then calculate the transition probability matrix over these states, which captures the fixation probability of any mutation in any given population (i.e., the probability that the mutant will take over that population). 
By computing the stationary distribution over this matrix we find the evolutionary population dynamics, which can be represented as a graph. 
The nodes of this graph correspond to the states, with the stationary distribution quantifying the average time spent by the populations in each node \cite{Traulsen06a,Nowak793}. 

A large body of prior literature has conducted this discrete-time Markov chain analysis in the context of {pair-wise} interaction games with {symmetric} payoffs \cite{Nowak793,Traulsen06a,Traulsen06b,Segbroek12,Santos11}.
Recent work applies the underlying assumption of small-mutation rates \cite{fudenberg2006imitation} to propose a general framework for discrete-time multi-player interactions \cite{veller2016finite}, which applies to games with asymmetric payoffs.
In our work, we formalize how such an evolutionary model, in the micro/macro dynamics spectrum, should be instantiated to converge to our novel and dynamical solution concept of MCCs. Additionally, we show (in \cref{thm:general_to_symmetric_model}) that in the case of identical per-population payoffs (i.e., $\forall k, M^k = M$) our generalization reduces to the single-population model used by prior works.
For completeness, we also detail the single population model in the Supplementary Material (see \cref{sec:supsingle}).
We now formally define the generalized discrete-time model.

Recall from \cref{sec:nfg} that each individual in a $K$-wise interaction game receives a local payoff $M^k(s^k, s^{-k})$ dependent on its identity $k$, its strategy $s^k$, and the strategy profile $s^{-k}$ of the other $K-1$ individuals involved in the game.
To account for the identity-dependent payoffs of such individuals, we consider the interactions of $K$ finite populations, each corresponding to a specific identity $k \in \{1,\ldots,K\}$.

In each population $k$, we have a set of strategies $S^k$ that we would like to evaluate for their evolutionary strength.
We also have a set of individuals $A$ in each population $k$, each of which is programmed to play a strategy from  the set $S^k$. 
Without loss of generality, we assume all populations have $m$ individuals.

Individuals interact $K$-wise through empirical games.
At each timestep $t$, one individual from each population is sampled uniformly, and the $K$ resulting individuals play a game.
Let $p^k_{s^k}$ denote the number of individuals in population $k$ playing strategy $s^k$ and $p$ denote the joint population state (i.e., vector of states of all populations). 
Under our sampling protocol, the fitness of an individual that plays strategy $s^k$ is,
\begin{align}
    f^k(s^k, p^{-k}) = \sum_{s^{-k} \in S^{-k}}M^{k}(s^k,s^{-k})\prod_{c \in \{1,\ldots,K\} \setminus {k}}\frac{p^c_{s^{c}}}{m}.\label{eq:multipop_fitness}
\end{align}
We consider any two individuals from a population $k$, with respective strategies $\focalAgent,\residentAgent \in S^k$ and respective fitnesses $f^k(\focalAgent, p^{-k})$ and $f^k(\residentAgent,p^{-k})$ (dependent on the values of the meta-game table). 
We introduce here a discrete-time dynamics, where the strategy of the first individual (playing $\focalAgent$) is then updated by either mutating with a very small probability to a random strategy (\cref{fig:MacroModelOverview_2}), probabilistically copying the strategy $\residentAgent$ of the second individual (\cref{fig:MacroModelOverview_3}), or sticking with its own strategy $\focalAgent$.
The idea is that strong individuals will replicate and spread throughout the population (\cref{fig:MacroModelOverview_4}).
While one could choose other variants of discrete-time dynamics \cite{claussen2008discrete}, we show that this particular choice both yields useful closed-form representations of the limiting behaviors of the populations, and also coincides with the MCC solution concept we later introduce under specific conditions.

As individuals from the same population never directly interact, the state of a population $k$ has no bearing on the fitnesses of its individuals.
However, as evident in \cref{eq:multipop_fitness}, each population's fitness may directly be affected by the competing populations' states.
The complexity of analyzing such a system can be significantly reduced by making the assumption of a small mutation rate \cite{fudenberg2006imitation}.
Let the `focal population' denote a population $k$ wherein a mutant strategy appears. 
We denote the probability for a strategy to mutate randomly into another strategy $s^k \in S^k$ by $\mu$ and we will assume it to be infinitesimally small (i.e., we consider a small-mutation limit $\mu \rightarrow 0$). 
If we neglected mutations, the end state of this evolutionary process would be monomorphic. 
If we introduce a very small mutation rate this means that either the mutant fixates and takes over the current population, or the current population is capable of wiping out the mutant strategy \cite{fudenberg2006imitation}.
Therefore, given a small mutation rate, the mutant either fixates or disappears before a new mutant appears in the current population.
This means that any given population $k$ will never contain more than two strategies at any point in time.
We refer the interested reader to \cite{veller2016finite} for a more extensive treatment of these arguments.


Applying the same line of reasoning, in the small-mutation rate regime, the mutant strategy in the focal population will either fixate or go extinct much earlier than the appearance of a mutant in any \emph{other} population  \cite{fudenberg2006imitation}.
Thus, at any given time, there can maximally be only one population with a mutant, and the remaining populations will be monomorphic; i.e., in each competing population $c \in \{1,\ldots,K\} \setminus k$, $\frac{p^c_{s^c}}{m} = 1$ for a single strategy and $0$ for the rest.
As such, given a small enough mutation rate, analysis of any focal population $k$ needs only consider the monomorphic states of all other populations.
Overloading the notation in \cref{eq:multipop_fitness}, the fitness of an individual from population $k$ that plays $s^k$ then considerably simplifies to 
\begin{align}
    f^k(s^k,s^{-k}) = M^{k}(s^k,s^{-k}),\label{eq:multipop_fitness_small_mutation}
\end{align}
where $s^{-k}$ denotes the strategy profile of the other populations.

Let $p^k_{\focalAgent}$ and $p^k_{\residentAgent}$ respectively denote the number of individuals playing $\focalAgent$ and $\residentAgent$ in focal population $k$, where $p^k_{\focalAgent} + p^k_{\residentAgent} = m$.
Per \cref{eq:multipop_fitness_small_mutation}, the fitness of an individual playing $\focalAgent$ in the focal population while the remaining populations play monomorphic strategies $s^{-k}$ is given by $f^k(\focalAgent,s^{-k}) = M^{k}(\focalAgent,s^{-k})$.
Likewise, the fitness of any individual in $k$ playing $\residentAgent$ is, $f^k(\residentAgent,s^{-k}) = M^{k}(\residentAgent,s^{-k})$.

We randomly sample two individuals in population $k$ and consider the probability that the one playing $\focalAgent$ copies the other individual's strategy $\residentAgent$. 
The probability with which the individual playing strategy $\focalAgent$ will copy the individual playing strategy $\residentAgent$ can be described by a \emph{selection} function $\mathbb{P}(\focalAgent \rightarrow \residentAgent, s^{-k})$, which governs the dynamics of the finite-population model.
For the remainder of the paper, we focus on the logistic selection function (aka Fermi distribution),
\begin{equation}\label{eq:fermi_distr_multipop}
    \mathbb{P}(\focalAgent \rightarrow \residentAgent, s^{-k})  = \frac{e^{\alpha f^k(\residentAgent,s^{-k})}}{e^{\alpha f^k(\focalAgent,s^{-k})}+e^{\alpha f^k(\residentAgent,s^{-k})}}=\left(1 + e^{\alpha(f^k(\focalAgent,s^{-k})-f^k(\residentAgent,s^{-k}))}\right)^{-1},
\end{equation}
with $\alpha$ determining the selection strength, which we call the \emph{ranking-intensity} (the correspondence between $\alpha$ and our ranking method will become clear later). 
There are alternative definitions of the selection function that may be used here, we merely focus on the Fermi distribution due to its extensive use in the single-population literature \cite{Santos11,Segbroek12,Traulsen06a}.

Based on this setup, we define a Markov chain over the set of strategy profiles $\prod_k S^{k}$ with $\prod_k |S^k|$ states. 
Each state corresponds to one of the strategy profiles $s \in \prod_k S^k$, representing a multi-population end-state where each population is monomorphic. 
The transitions between these states are defined by the corresponding fixation probabilities (the probability of overtaking the population) when a mutant strategy is introduced in any single monomorphic population $k$.
We now define the Markov chain, which has $(\prod_k |S^k|)^2$ transition probabilities over all pairs of monomorphic multi-population states.
Denote by $\rho^{k}_{\residentAgent,\focalAgent}(s^{-k})$ the probability of mutant strategy $\focalAgent$ fixating in a focal population $k$ of individuals playing $\residentAgent$, while the remaining $K-1$ populations remain in their monomorphic states $s^{-k}$.
For any given monomorphic strategy profile, there are a total of $\sum_{k}(|S^k|-1)$ valid transitions to a subsequent profile where only a single population has changed its strategy. 
Thus, letting $\eta = \frac{1}{\sum_{k}(|S^k|-1)}$, then $\eta\rho^{k}_{\residentAgent,\focalAgent}(s^{-k})$ is the probability that the joint population state transitions from $(\residentAgent,s^{-k})$ to state $(\focalAgent,s^{-k})$ after the occurrence of a single mutation in population $k$. 
The stationary distribution over this Markov chain tells us how much time, on average, the dynamics will spend in each of the monomorphic states.

The fixation probabilities (of a rare mutant playing $\focalAgent$ overtaking the focal population $k$) can be calculated as follows. 
The probability that the number of individuals playing $\focalAgent$ decreases/increases by one in the focal population is given by,
\begin{equation}
  T^{k(\mp 1)}(p^k,\focalAgent,\residentAgent,s^{-k})= \frac{p^{k}_{\focalAgent}p^{k}_{\residentAgent}}{m(m-1)}\left(1+e^{\pm \alpha(f^k(\focalAgent,s^{-k})- f^k(\residentAgent,s^{-k}))}\right)^{-1}.
\end{equation}
Then, the fixation probability $\rho^{k}_{\residentAgent,\focalAgent}(s^{-k})$ of a single mutant with strategy $\focalAgent$ in a population $k$ of $m-1$ individuals playing $\residentAgent$ is,
\begin{align}
    \rho^{k}_{\residentAgent,\focalAgent}(s^{-k}) &=  \left(1+\sum_{l=1}^{m-1}\prod_{p^{k}_{\focalAgent}=1}^l \frac{T^{k(-1)}(p^k,\focalAgent,\residentAgent,s^{-k})}{T^{k(+1)}(p^k,\focalAgent,\residentAgent,s^{-k})}\right)^{-1}\\
     &= \left(1+\sum_{l=1}^{m-1}\prod_{p^{k}_{\focalAgent}=1}^l e^{-\alpha(f^k(\focalAgent,s^{-k})- f^k(\residentAgent,s^{-k}))} \right)^{-1}\\
    &= \left(1+\sum_{l=1}^{m-1} e^{-l\alpha(f^k(\focalAgent,s^{-k})- f^k(\residentAgent,s^{-k}))} \right)^{-1}\\
    &= \begin{cases}
    \frac{1-e^{-\alpha(f^k(\focalAgent,s^{-k})- f^k(\residentAgent,s^{-k}))}}{1-e^{-m\alpha(f^k(\focalAgent,s^{-k})- f^k(\residentAgent,s^{-k}))}} & \text{if } f^k(\focalAgent,s^{-k}) \neq f^k(\residentAgent,s^{-k}) \\
    \frac{1}{m} & \text{if } f^k(\focalAgent,s^{-k})= f^k(\residentAgent,s^{-k})
    \end{cases}
    \label{eq:multipop_fixation_prob}
\end{align}
This corresponds to the computation of an $m$-step transition in the Markov chain corresponding to $\mathbb{P}(\focalAgent \rightarrow \residentAgent, s^{-k})$ \cite{TaylKarl98}. 
The quotient $\frac{T^{k(-1)}(p^k,\focalAgent,\residentAgent,s^{-k})}{T^{k(+1)}(p^k,\focalAgent,\residentAgent,s^{-k})}$ expresses the likelihood (odds) that the mutation process in population $k$ continues in either direction: 
if it is close to zero then it is very likely that the number of mutants (individuals with strategy $\focalAgent$ in population $k$) increases; if it is very large it is very likely that the number of mutants will decrease; and if it close to one then the probabilities of increase and decrease of the number of mutants are equally likely.
This yields the following Markov transition matrix corresponding to the jump from strategy profile $s_i \in \prod_k S^k$ to $s_j \in \prod_k S^k$,
\begin{align}\label{eqn:transition_matrix_C_multipop}
    C_{ij} &= 
    \begin{cases}
         \eta \rho^k_{s^k_i,s^k_j}(s_i^{-k}) & \text{if $\exists k$ such that $s_i^{k}\neq s_j^{k}$ and $s_i^{-k}=s_j^{-k}$}\\ 
         1 - \sum_{j\not= i} C_{ij}  & \text{if $s_i = s_j$}\\ 
         0 & \text{otherwise} 
    \end{cases}
\end{align}
for all $i,j \in \{1,\ldots,|S|\}$, where $|S| = \prod_{k} |S^k|$.



The following theorem formalizes the irreducibility of this finite-population Markov chain, a property that is well-known in the literature (e.g., see \cite[Theorem 2]{fudenberg2006imitation} and \cite[Theorem 1]{veller2016finite}) but stated here for our specialized model for completeness.
\begin{restatable}{theorem}{UniquePi}
\label{property:unique_pi}
    Given finite payoffs, the Markov chain with transition matrix $C$ is irreducible (i.e., it is possible to get to any state starting from any state). Thus a unique stationary distribution $\pi$ (where $\pi^TC= \pi^T$ and $\sum_i \pi_i = 1$) exists.
\end{restatable}
\begin{proof}
    \ProofInSM{{sec:proof_unique_pi}}
\end{proof}
This unique $\pi$ provides the evolutionary ranking, or strength of each strategy profile in the set $\prod_k S^k$, expressed as the average time spent in each state in distribution $\pi$.

This generalized discrete-time evolutionary model, as later shown, will form the basis of our $\alpha$-Rank method.
We would like to clarify the application of this general model to the single population case, which applies only to symmetric 2-player games and is commonly used in the literature (see \cref{sec:related_work}).

\paragraph{Application to Single-Population (Symmetric Two-Player) Games} 
For completeness, we provide a detailed outline of the single population model in Supplementary Material \cref{sec:supsingle}.

\begin{theorem}[Multi-population model generalizes the symmetric single-population model]\label{thm:general_to_symmetric_model}
    The general multi-population model inherently captures the dynamics of the single population symmetric model. 
\end{theorem}

\begin{proof}(Sketch)
    In the pairwise symmetric game setting, we consider only a single population of interacting individuals (i.e., $K=1$), where a maximum of two strategies may exist at any time in the population due to the small mutation rate assumption. 
    At each timestep, two individuals (with respective strategies $\focalAgent,\residentAgent \in S^1$) are sampled from this population and play a game using their respective strategies $\focalAgent$ and $\residentAgent$.
    Their respective fitnesses then correspond directly to their payoffs, i.e., $f_{\focalAgent} = M(\focalAgent,\residentAgent)$ and $f_{\residentAgent} = M(\residentAgent,\focalAgent)$.
    With this change, all other derivations and results follow directly the generalized model.
    For example, the probability of decrease/increase of a strategy of type $s_\focalAgent$ in the single-population case translates to, 
    \begin{equation}
      T^{(\mp 1)}(p,\focalAgent,\residentAgent)= \frac{p_{\focalAgent}p_{\residentAgent}}{m(m-1)}\left(1+e^{\pm \alpha(f_{\focalAgent}- f_{\residentAgent})}\right)^{-1},
    \end{equation}
    and likewise for the remaining equations.
\end{proof}
In other words, the generalized model is general in the sense that one can not only simulate symmetric pairwise interaction dynamics, but also $K$-wise and asymmetric interactions.

\paragraph{Linking the Micro- and Macro-dynamics Models}
We have introduced, so far, a micro- and macro-dynamics model, each with unique advantages in terms of analyzing the evolutionary strengths of agents. 
The formal relationship between these two models remains of interest, and is established in the limit of a large population: 
\begin{restatable}[Discrete-Continuous Edge Dynamics Correspondence]{theorem}{DiscreteContEdgeDyn}
\label{thm:discrete_cont_edge_dyn}
    In the large-population limit, the macro-dynamics model is equivalent to the micro-dynamics model over the edges of the strategy simplex. Specifically, the limiting model is a variant of the replicator dynamics with the caveat that the Fermi revision function takes the place of the usual fitness terms.
\end{restatable}
\begin{proof}
    \ProofInSM{{sec:proof_discrete_cont_edge_dyn}}
\end{proof}
Therefore, a correspondence exists between the two models on the `skeleton' of the simplex, with the macro-dynamics model useful for analyzing the global evolutionary behaviors over this skeleton, and the micro-model useful for `zooming into' the three- or four-faces of the simplex to analyze the interior dynamics.

In the next sections, we first give a few conceptual examples of the generalized discrete-time model, then discuss the need for a new solution concept and the incompatibility between Nash equilibria and dynamical systems.
We then directly link the generalized model to our new game-theoretic solution concept, Markov-Conley chains (in \cref{thm:inf_pop_alpha}).

\subsection{Conceptual Examples}\label{sec:markov_chain_conceptual_example}
\begin{figure}[t]
    \centering
    \begin{subfigure}[t]{0.48\textwidth}
        \centering
        \begin{tabular}{cc|c|c|c|}
        & \multicolumn{1}{c}{} & \multicolumn{3}{c}{Player 2}\\
        & \multicolumn{1}{c}{} & \multicolumn{1}{c}{$R$}  & \multicolumn{1}{c}{$P$} & \multicolumn{1}{c}{$S$} \\\cline{3-5}
        \multirow{3}*{Player 1}  & $R$ & $0$ & $-1$ & $1$ \\\cline{3-5}
                                 & $P$ & $1$ & $0$ & $-1$ \\\cline{3-5}
                                 & $S$ & $-1$ & $1$ & $0$ \\\cline{3-5}
        \end{tabular}
        \includegraphics[width=1\textwidth,trim={3cm 1cm 1cm 1cm}, clip]{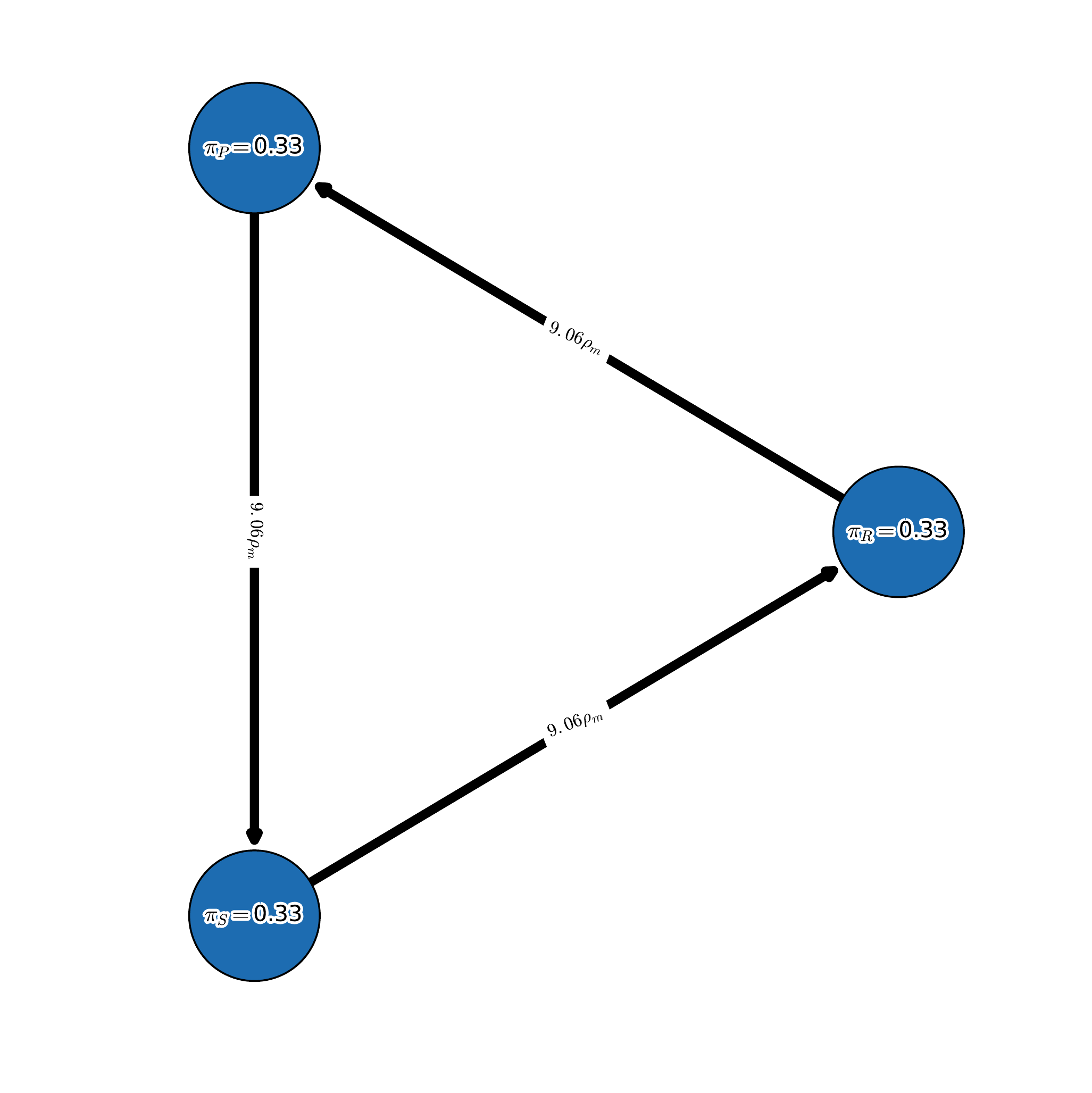}
        \caption{Payoffs (top) and single-population discrete-time dynamics (bottom) for Rock-Paper-Scissors game. Graph nodes correspond to monomorphic populations $R$, $P$, and $S$.
        }
        \label{fig:mcc_rock_paper_scissors_intext}
    \end{subfigure}
    \hfill
    \begin{subfigure}[t]{0.48\textwidth}
        \centering
        \begin{tabular}{cc|c|c|}
            & \multicolumn{1}{c}{} & \multicolumn{2}{c}{Player 2}\\
            & \multicolumn{1}{c}{} & \multicolumn{1}{c}{$O$}  & \multicolumn{1}{c}{$M$} \\\cline{3-4}
            \multirow{2}*{Player 1}  & $O$ & $(3,2)$ & $(0,0)$ \\\cline{3-4}
                                       & $M$ & $(0,0)$ & $(2,3)$ \\\cline{3-4}
        \end{tabular}
        \includegraphics[width=1\textwidth,trim={1cm 1cm 0.3cm 0cm}, clip]{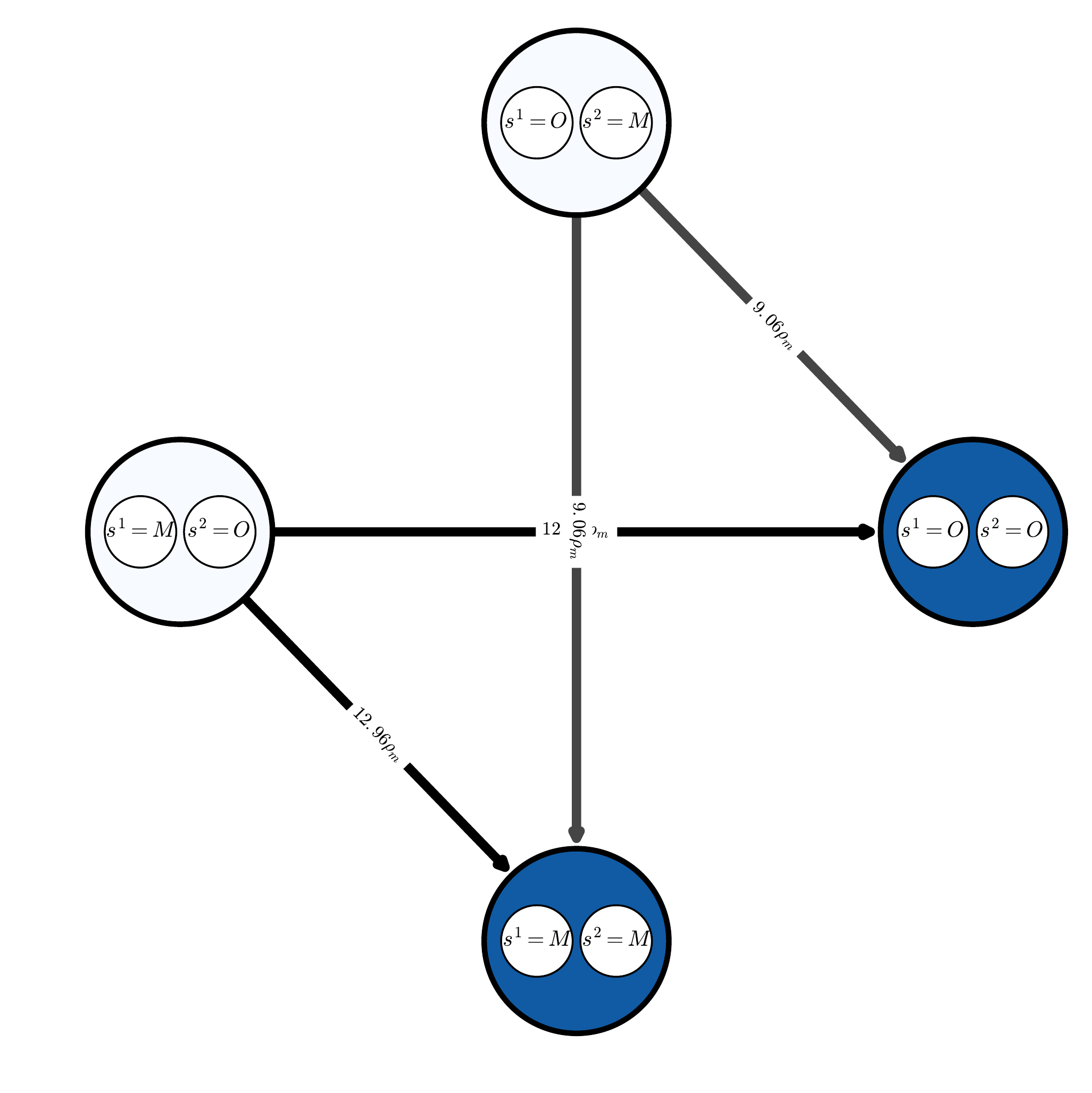}
        \caption{Payoffs (top) and multi-population discrete-time dynamics (bottom) for Battle of the Sexes game. Strategies O and M respectively correspond to going to the Opera and Movies. Graph nodes correspond to monomorphic population profiles $(s^1,s^2)$. The stationary distribution $\pi$ has 0.5 mass on each of profiles $(O,O)$ and $(M,O)$, and 0 mass elsewhere.}
        \label{fig:mcc_battle_of_the_sexes_intext}
    \end{subfigure}
    \caption{Conceptual examples of finite-population models, for population size $m=50$ and ranking-intensity $\alpha=0.1$.}
    \label{fig:finite_pop_overview}
\end{figure}

We present two canonical examples that visualize the discrete-time dynamics and build intuition regarding the macro-level insights gained using this type of analysis.

\subsubsection{Rock-Paper-Scissors}\label{sec:markov_chain_rps}
We first consider the single-population (symmetric) discrete-time model in the Rock-Paper-Scissors (RPS) game, with the payoff matrix shown in \cref{fig:mcc_rock_paper_scissors_intext} (top).
One can visualize the discrete-time dynamics using a graph that corresponds to the Markov transition matrix $C$ defined in \cref{eqn:transition_matrix_C_multipop}, as shown in \cref{fig:mcc_rock_paper_scissors_intext} (bottom).

Nodes in this graph correspond to the monomorphic population states.
In this example, these are the states of the population where all individuals play as agents Rock, Paper, or Scissors.
To quantify the time the population spends as each agent, we indicate the corresponding mass of the stationary distribution $\pi$ within each node.
As can be observed in the graph, the RPS population spends exactly $\frac{1}{3}$ of its time as each agent.

Edges in the graph correspond to the fixation probabilities for pairs of states.
Edge directions corresponds to the flow of individuals from one agent to another, with strong edges indicating rapid flows towards ‘fitter’ agents.
We denote fixation probabilities as a multiple of the neutral fixation probability baseline, $\rho_m=\frac{1}{m}$, which corresponds to using the Fermi selection function with $\alpha=0$.
To improve readability of the graphs,  we also do not visualize edges looping a node back to itself, or edges with fixation probabilities lower than $\rho_m$.
In this example, we observe a cycle (intransitivity) involving all three agents in the graph.
While for small games such cycles may be apparent directly from the structure of the payoff table, we later show that the graph visualization can be used to automatically iterate through cycles even in $K$-player games involving many agents.

\subsubsection{Battle of the Sexes}\label{sec:markov_chain_bots}
Next we illustrate the generalized multi-population (asymmetric) model in the Battle of the Sexes game, with the payoff matrix shown in \cref{fig:mcc_battle_of_the_sexes_intext} (top).
The graph now corresponds to the interaction of two populations, each representing a player type, with each node corresponding to a monomorphic population \emph{profile} $(s^1,s^2)$.
Edges, again, correspond to fixation probabilities, but occur only when a single population changes its strategy to a different one (an artifact of our small mutation assumption). 
In this example, it is evident from the stationary distribution that the populations spend an equal amount of time in profiles $(O,O)$ and $(M,M)$, and a very small amount of time in states $(O,M)$ and $(M,O)$.

\subsection{The Incompatibility of Nash Equilibrium and Dynamical Systems}\label{sec:incompatibility_nash_dynamical}

Continuous- and discrete-time dynamical systems have been used extensively in Game Theory, Economics, and Algorithmic Game Theory.
In the particular case of multi-agent evaluation in meta-games, this type of analysis is relied upon for revealing useful insights into the strengths and weaknesses of interacting agents \cite{Tuyls18}. 
Often, the goal of research in these areas is to establish that, in some sense and manner, the investigated dynamics actually converge to a Nash equilibrium;
there has been limited success in this front, and there  are some negative results \citep{Frong,Hart03,viossat07}.
In fact, all known dynamics in games (the replicator dynamics, the many continuous variants of the dynamics used in the proof of Nash's theorem, etc.) do cycle.
To compound this issue, meta-games are often large, extend beyond pair-wise interactions, and may not be zero-sum.
While solving for a Nash equilibrium can be done in polynomial time for zero-sum games, doing so in general-sum games is known to be PPAD-complete  \cite{Daskalakis06thecomplexity}, which severely limits the feasibility of using such a solution concept for evaluating our agents.

Of course, some dynamics are known to converge to {\em relaxations} of the Nash equilibrium, such as the correlated equilibrium polytope or the coarse correlated equilibria \cite{piliouras2017learning}.  
But unfortunately, this ``convergence'' is typically considered in the sense of {\em time average}; time averages can be useful for establishing performance bounds for games, but tell us little about actual system behavior --- which is a core component of what we study through games. 
For certain games, dynamics may indeed converge to a Nash equilibrium, but they may also {\em cycle}.  
For example, it is encouraging that in all $2 \times 2$ matrix games these equilibria, cycles, and slight generalizations thereof are the only possible limiting behaviors for continuous-time dynamics (i.e., flows). 
But unfortunately this clean behavior (convergence to either a cycle or, as a special case, to a Nash equilibrium) is an artifact of the two-dimensional nature of $2 \times 2$ games, a consequence of the Poincar\' e--Bendixson theorem \citep{Sandholm10}.
There is a wide range of results in different disciplines arguing that learning dynamics in games tend to not equilibrate to any Nash equilibrium but instead exhibit complex, unpredictable behavior (e.g., \cite{Gaunersdorfer,daskalakis10,paperics11,sandholm2010population,wagner2013explanatory,PalaiopanosPP17}).
The dynamics of even simple two-person games with three or more strategies per player can be chaotic \citep{Sato02042002}, that is, inherently difficult to predict and complex. 
Chaos goes against the core of our project; there seems to be little hope for building a predictive theory of player behavior based on dynamics in terms of Nash equilibrium.

\subsection{Markov-Conley chains: A Dynamical Solution Concept}
\begin{figure}[t]
    \centering
    \begin{subfigure}[t]{0.5\textwidth}
        \centering
        \begin{tabular}{cc|c|c|}
        & \multicolumn{1}{c}{} & \multicolumn{2}{c}{Player 2}\\
        & \multicolumn{1}{c}{} & \multicolumn{1}{c}{$H$}  & \multicolumn{1}{c}{$T$} \\\cline{3-4}
        \multirow{2}*{Player 1}  & $H$ & $(1,-1)$ & $(-1,1)$ \\\cline{3-4}
                                   & $T$ & $(-1,1)$ & $(1,-1)$ \\\cline{3-4}
        \end{tabular}
        \includegraphics[width=0.8\textwidth]{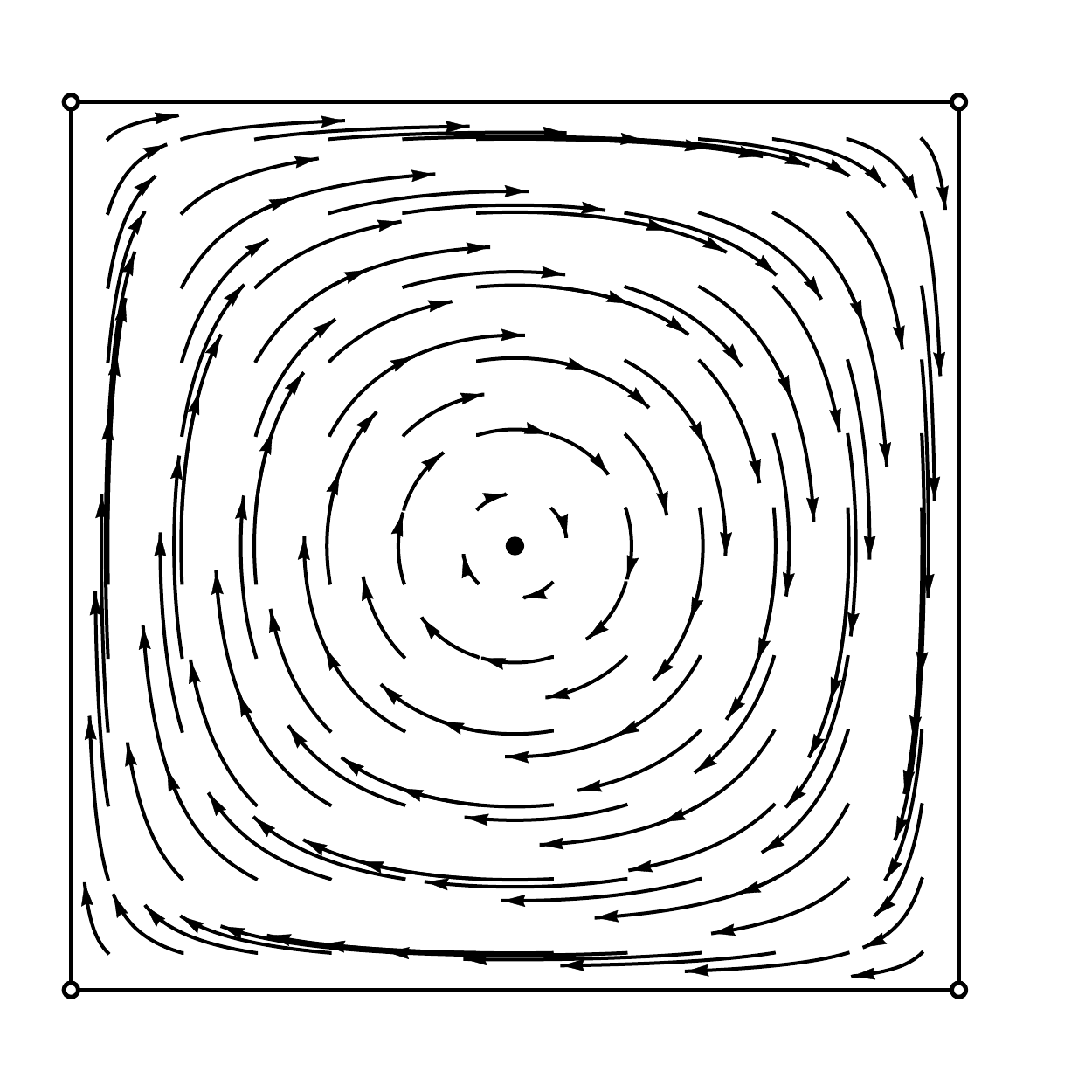}
        \caption{Matching Pennies game.}
        \label{fig:rd_matching_pennies}
    \end{subfigure}%
    \begin{subfigure}[t]{0.5\textwidth}
        \centering
        \begin{tabular}{cc|c|c|}
            & \multicolumn{1}{c}{} & \multicolumn{2}{c}{Player 2}\\
            & \multicolumn{1}{c}{} & \multicolumn{1}{c}{$A$}  & \multicolumn{1}{c}{$B$} \\\cline{3-4}
            \multirow{2}*{Player 1}  & $A$ & $(1,1)$ & $(-1,-1)$ \\\cline{3-4}
                                     & $B$ & $(-1,-1)$ & $(1,1)$ \\\cline{3-4}
        \end{tabular}
        \includegraphics[width=0.8\textwidth]{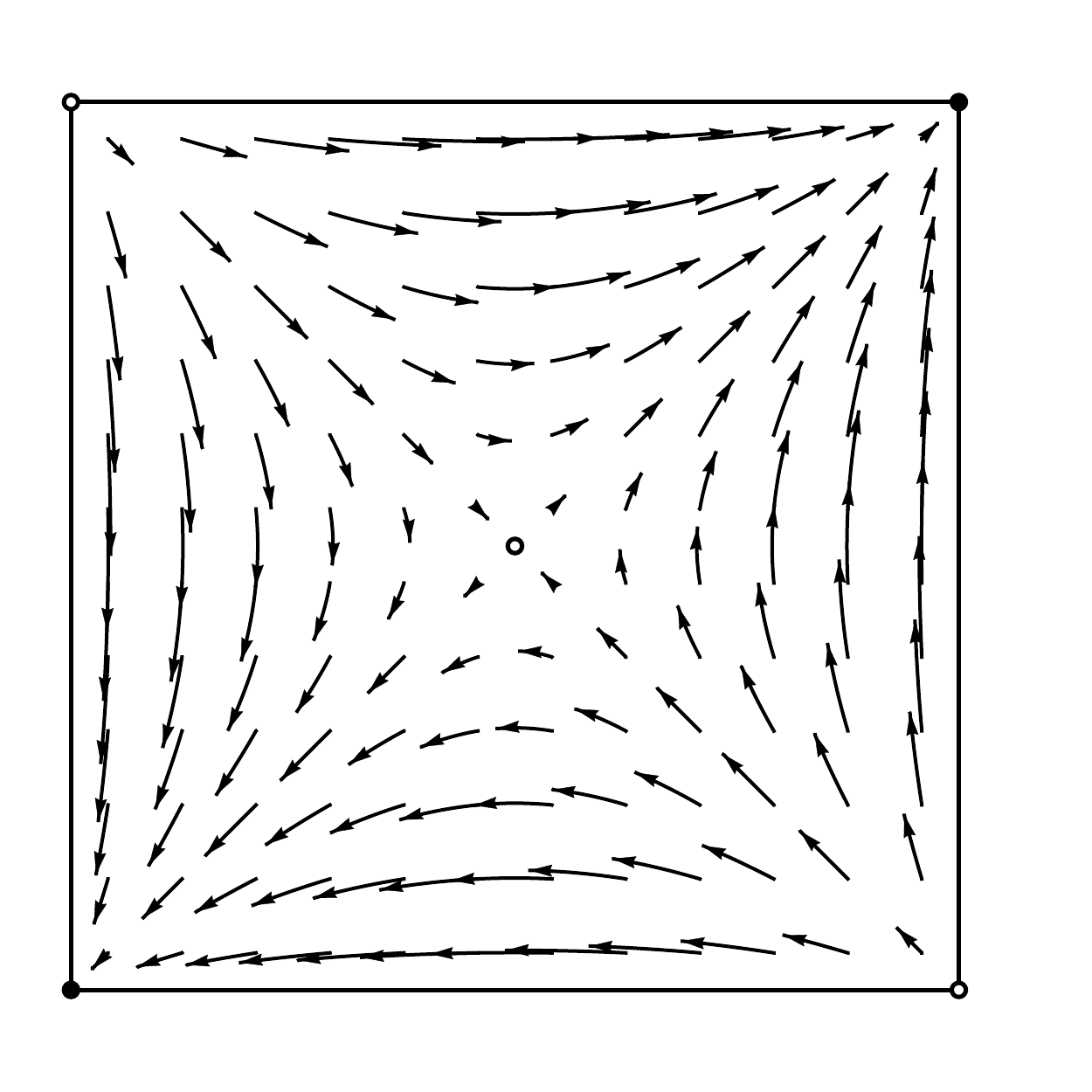}
        \caption{Partnership, coordination game.}
        \label{fig:rd_coordination_game}
    \end{subfigure}
    \caption{Canonical game payoffs and replicator dynamics trajectories. Each point encodes the probability assigned by the players to their first strategy. The matching pennies replicator dynamics have \emph{one} chain component, consisting of the whole domain.
    The coordination game dynamics have five chain components (corresponding to the fixed points, four in the corners and one mixed, which are recurrent by definition), as was formally shown by \cite{Papadimitriou:2016:NEC:2840728.2840757}.}
    \label{fig:rd_examples}
\end{figure}

Recall our overall objective: we would like to understand and evaluate multi-agent interactions using a detailed and realistic model of evolution, such as the replicator dynamics, in combination with a game-theoretic solution concept.
We start by acknowledging the fundamental incompatibility between dynamics and the Nash equilibrium: dynamics are often incapable of reaching the Nash equilibrium. 
However, instead of taking this as a disappointing flaw of dynamics, we see it instead as an opportunity to look for a novel solution concept that does not have the same limitations as Nash in relation to these dynamical systems.
We contemplate whether a plausible algorithmic solution concept can emerge by asking, {\em what do these dynamics converge to?}  
Our goal is to identify the non-trivial, irreducible behaviors of a dynamical system (i.e., behaviors that cannot be partitioned more finely in a way that respects the system dynamics) and thus provide a new solution concept --- an alternative to Nash's --- that will enable evaluation of of multi-agent interactions using the underlying evolutionary dynamics.
We carve a pathway towards this alternate solution concept by first considering the topology of dynamical systems.

\subsubsection{Topology of Dynamical Systems and Conley's Theorem}\label{sec:topology_of_dyn_sys}

Dynamicists and topologists have been working hard throughout the past century to find a way to extend to higher dimensions the benign yet complete limiting dynamical behaviors  described in \cref{sec:incompatibility_nash_dynamical} that one sees in two dimensions:
convergence to cycles (or equilibria as a special case).
That is, they have been trying to find an appropriate \emph{relaxation of the notion of a cycle} such that the two-dimensional picture is restored.  
After many decades of trial and error, new and intuitive conceptions of ``periodicity'' and ``cycles'' were indeed discovered, in the form of \emph{chain recurrent sets} and \emph{chain components}, which we define in this section.
These key ingredients form the foundation of Conley's Fundamental Theorem of Dynamical Systems, which in turn leads to the formulation of our Markov-Conley chain solution concept and associated multi-agent evaluation scheme.

\paragraph{Definitions}
To make our treatment formal, we require definitions of the following set of topological concepts, based primarily on the work of Conley \citep{conley1978isolated}. 
Our chain recurrence approach and the theorems in this section follow from \citep{alongi2007recurrence}.
We also provide the interested reader a general background on dynamical systems in Supplementary Material \ref{sec:background_dynamics} in an effort to make our work self-contained.

\begin{figure}[t]
    \centering
    \includegraphics[width=0.7\textwidth,page=1]{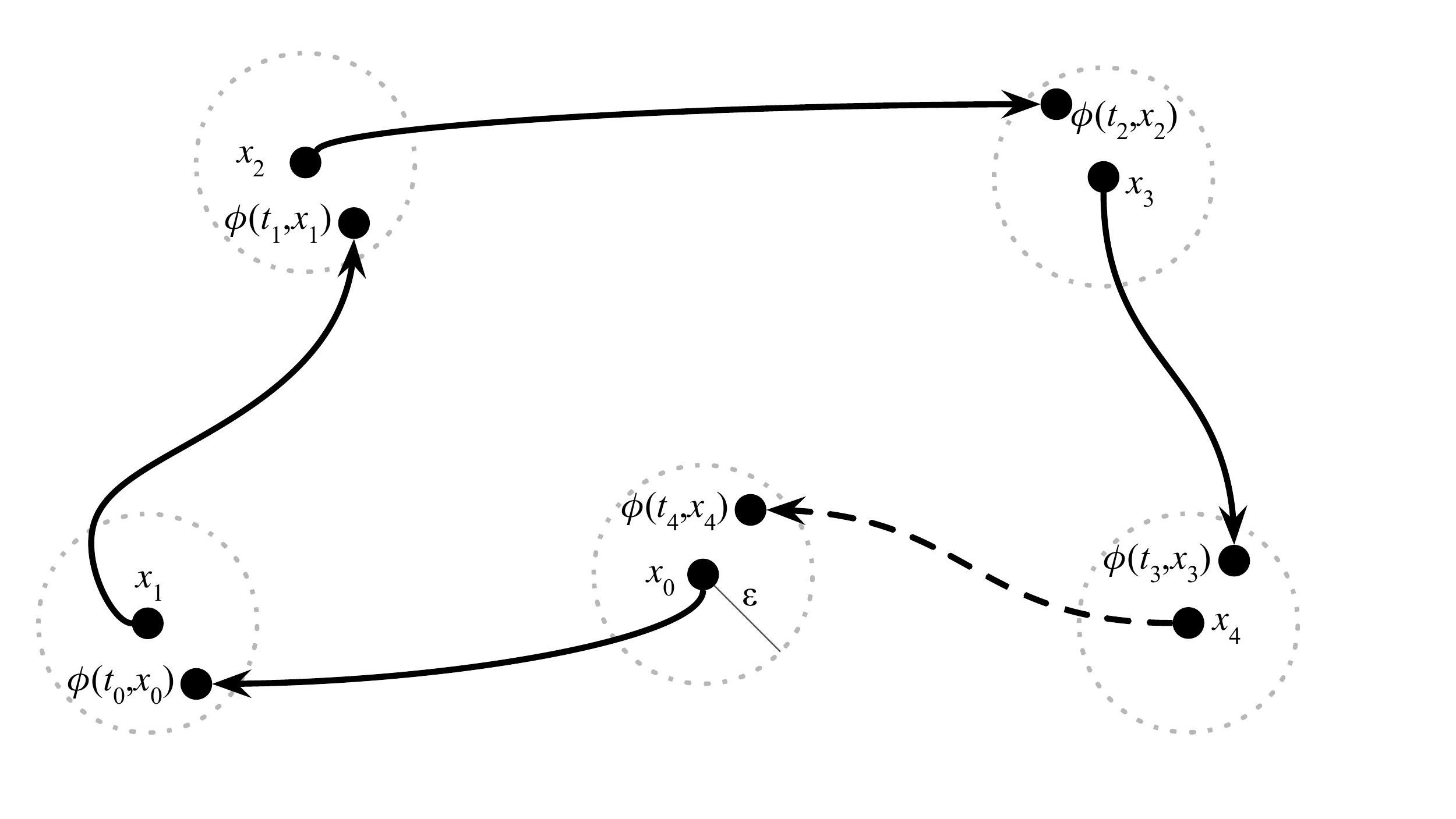}
    \caption{Topology of dynamical systems: an $(\epsilon,T)$-chain from $x_0$ to $x_4$ with respect to flow $\phi$ is exemplified here by the solid arrows and sequence of points $x_0,x_1,x_2,x_3,x_4$. If the recurrent behavior associated with point $x_0$ (indicated by the dashed arrow) holds for all $\epsilon>0$ and $T>0$, then it is a chain recurrent point.}
    \label{fig:epsilon_t_chain}
\end{figure}

\begin{definition}[Flow]
    A flow on a topological space $X$ is a continuous mapping $\phi:\real\times X \rightarrow X$ such that:
    \begin{description}
    \item{(i)} $\phi(t, \cdot): X \rightarrow X$ is a homeomorphism for each $t \in \real$.
    \item{(ii)} $\phi(0, x)=x$ for all $x \in X$.
    \item{(iii)} $\phi(s+t, x)= \phi(s,\phi(t,x))$ for all $s,t \in \real$ and all $x \in X$.
    \end{description}
\end{definition}
Depending on the context, we sometimes write $\phi^t(x)$  for $\phi(t,x)$ and denote a flow $\phi:\real\times X\rightarrow X$  by $\phi^t:X\rightarrow X$, where $t \in \real$.

\begin{definition}[$(\epsilon, T)$-chain]
    Let $\phi$ be a flow on a metric space $(X,d)$.  Given $\epsilon>0$, $T>0$, and $x,y\in X$, an $(\epsilon, T)$-chain from $x$ to $y$ with respect to $\phi$ and $d$ is a pair of finite sequences $x=x_0,x_1,\dots,x_{n-1},x_n=y$ in $X$ and $t_0,\dots,t_{n-1}$ in $[T,\infty)$, denoted together by $(x_0,\dots,x_n;t_0,\dots,t_{n-1})$ such that, \begin{equation}
        d(\phi^{t_i}(x_i),x_{i+1})<\epsilon,
    \end{equation}
 for $i=0,1,2,\dots,n-1$.
\end{definition}
Intuitively, an $(\epsilon,T)$ chain corresponds to the forward dynamics under flow $\phi$ connecting points $x,y \in X$, with slight perturbations allowed at each timestep (see \cref{fig:epsilon_t_chain} for an example). 
Note these deviations are allowed to occur at step-sizes $T$ bounded away from $0$, as otherwise the accumulation of perturbations could yield trajectories completely dissimilar to those induced by the original flow \cite{norton1995fundamental}.

\begin{definition}[Forward chain limit set]
Let $\phi$ be a flow on a metric space $(X,d)$.  The forward chain limit set of $x \in X$ with respect to $\phi$ and $d$ is the set,
\begin{equation}
    \Omega^{+}(\phi,x) = \bigcap_{\epsilon,T>0}\{y \in X~|~ \exists \text{ an } (\epsilon,T)\text{-chain from~} x \text{~to~} y \text{~with respect to~} \phi\}.
\end{equation}
\end{definition}

\begin{definition}[Chain equivalent points]
    Let $\phi$ be a flow on a metric space $(X,d)$. Two points $x, y \in X$ are chain equivalent with respect to $\phi$ and $d$ if $y \in \Omega^+(\phi,x)$ and  $x \in \Omega^+(\phi,y)$.
\end{definition}

\begin{definition}[Chain recurrent point]
    Let $\phi$ be a flow on a metric space $(X,d)$. 
    A point $x \in X$ is chain recurrent with respect to $\phi$ and $d$ if $x$ is chain equivalent to itself; i.e., there exists an $(\epsilon,T)$-chain connecting $x$ to itself for \emph{every} $\epsilon>0$ and $T>0$.
\end{definition}
Chain recurrence can be understood as an orbit with slight perturbations allowed at each time step (see \cref{fig:epsilon_t_chain}), which constitutes a new conception of ``periodicity" with a very intuitive explanation in Computer Science terms:
Imagine Alice is using a computer to simulate the trajectory of a dynamical system that induces a flow $\phi$.
Every time she computes a single iteration of the dynamical process with a minimum step-size $T$, there is a rounding error $\epsilon$.
Consider an adversary, Bob, who can manipulate the result at each timestep within the $\epsilon$-sphere of the actual result.
If, regardless of $\epsilon$ or minimum step-size $T$, Bob can persuade Alice that her dynamical system starting from a point $x$ returns back to this point in a finite number of steps, then this point is chain recurrent.


This new notion of ``periodicity" (i.e., \emph{chain recurrence}) leads to a corresponding notion of a ``cycle" captured in the concept of \emph{chain components}, defined below. 

\begin{definition}[Chain recurrent set]
     The chain recurrent set of flow $\phi$, denoted $\mathcal{R}(\phi)$, is the set of all chain recurrent points of $\phi$ .
\end{definition}




\begin{definition}[Chain equivalence relation $\sim$]
    Let the relation $\sim$ on $\mathcal{R}(\phi)$ be defined by $x \sim y$ if and only if $x$ is chain equivalent to $y$.
    This is an equivalence relation on the chain recurrent set $\mathcal{R}(\phi)$.
\end{definition}

\begin{definition}[Chain component]
    The equivalence classes in $\mathcal{R}(\phi)$ of the chain equivalence relation $\sim$ are called the chain components of $\phi$.
\end{definition}
In the context of the Alice and Bob example, chain components are the maximal sets $A$ such that for any two points $x,y \in A$, Bob can similarly persuade Alice that the flow $\phi$ induced by her dynamical system can get her from $x$ to $y$ in a finite number of steps. 
For example the matching pennies replicator dynamics (shown in \cref{fig:rd_matching_pennies}) have \emph{one} chain component, consisting of the entire domain; in the context of the Alice and Bob example, the cyclical nature of the dynamics throughout the domain means that Bob can convince Alice that any two points may be connected using a series of finite perturbations $\epsilon$, for all $\epsilon>0$ and $T>0$.
On the other hand, the coordination game replicator dynamics (shown in \cref{fig:rd_coordination_game}) has five chain components corresponding to the fixed points (which are recurrent by definition): four in the corners, and one mixed strategy fixed point in the center.
For a formal treatment of these examples, see \cite{Papadimitriou:2016:NEC:2840728.2840757}.

Points in each chain component are transitive by definition.
Naturally, the chain recurrent set $\mathcal{R}(\phi)$ can be partitioned into a (possibly infinitely many) number of chain components.
In other words, chain components cannot be partitioned more finely in a way that respects the system dynamics; they constitute the fundamental topological concept needed to define the irreducible behaviors we seek. 



\paragraph{Conley's Theorem}
We now wish to characterize the role of chain components in the long-term dynamics of systems, such that we can evaluate the limiting behaviors of multi-agent interactions using our evolutionary dynamical models. 
Conley's Fundamental Theorem of Dynamical Systems leverages the above perspective on ``periodicity" (i.e., chain recurrence) and ``cycles" (i.e., chain components) to decompose the domain of any dynamical system into two classes: 1) chain components, and 2) transient points.
To introduce Conley's theorem, we first need to define the notion of a complete Lyapunov function. 
The game-theoretic analogue of this idea is the notion of a potential function in \emph{potential games}.
In a potential game, as long as we are not at an equilibrium, the potential is strictly decreasing and guiding the dynamics towards the standard game-theoretic solution concept, i.e., equilibria \cite{monderer:96}. 
The notion of a complete Lyapunov function switches the target solution concept from equilibria to chain recurrent sets. 
More formally:
\begin{definition}[Complete Lyapunov function]
Let $\phi$ be a flow on a metric space $(X,d)$. A complete Lyapunov function for $\phi$ is a continuous function $\gamma:X\rightarrow \real$ such that,
\end{definition}
    \begin{enumerate}
        \item $\gamma(\phi^t(x))$ is a strictly decreasing function of $t$ for all $x \in X\setminus \mathcal{R}(\phi)$,
        \item for all $x,y \in \mathcal{R}(\phi)$ the points $x$, $y$ are in the same chain component if and only if $\gamma(x)=\gamma(y)$,
        \item $\gamma(\mathcal{R}(\phi))$ is nowhere dense. 
    \end{enumerate}

\noindent Conley's Theorem, the important result in topology that will form the basis of our solution concept and ranking scheme, is as follows:
\begin{theorem}[Conley's Fundamental Theorem of Dynamical Systems \cite{conley1978isolated}, informal statement]
    The domain of any dynamical system can be decomposed into its (possibly infinitely many) chain components; the remaining points are transient, each led to the recurrent part by a Lyapunov function.
\end{theorem}

The powerful implication of Conley's Theorem is that complete Lyapunov functions always exist.
\begin{theorem}[\cite{conley1978isolated}]
    Every flow on a compact metric space has a complete Lyapunov function.
\end{theorem}
In other words, the space $X$ is decomposed into points that are chain recurrent and points that are led to the chain recurrent part in a gradient-like fashion with respect to a Lyapunov function that is guaranteed to exist.
In game-theoretic terms, every game is a ``potential" game, if only we change our solution concept from equilibria to chain recurrent sets.

\subsubsection{Asymptotically Stable Sink Chain Components}\label{sec:sink_chain_components}
Our objective is to investigate the likelihood of an agent being played in a $K$-wise meta-game by using a detailed and realistic model of multi-agent evolution, such as the replicator dynamics.
While chain components capture the limiting behaviors of dynamical systems (in particular, evolutionary dynamics that we seek to use for our multi-agent evaluations), they can be infinite in number (as mentioned in \cref{sec:topology_of_dyn_sys}); 
it may not be feasible to compute or use them in practice within our evaluation scheme.
To resolve this, we narrow our focus onto a particular class of chain components called \emph{asymptotically stable sink chain components}, which we define in this section.
Asymptotically stable sink chain components are a natural target for this investigation as they encode the possible ``final" long term system;
by contrast, we can escape out of other chain components via infinitesimally small perturbations. 
We prove in the subsequent section (\cref{thm:MCC}, specifically) that, in the case of replicator dynamics and related variants, asymptotically stable sink chain components are finite in number;
our desired solution concept is obtained as an artifact of this proof.



We proceed by first showing that the chain components of a dynamical system can be partially ordered by reachability through chains, and we focus on the {\em sinks} of this partial order. We start by defining a partial order on the set of chain components:
\begin{definition}
    Let $\phi$ be a flow on a metric space and $A_1, A_2$ be chain components of the flow.
    Define the relation $A_1 \leq_C A_2$ to hold if and only if there exists $x \in A_2$ and $y \in A_1$ such that $y \in \Omega^{+}(\phi,x)$.
\end{definition}
Intuitively, $A_1 \leq_C A_2$, if we can reach $A_1$ from $A_2$ with $(\epsilon,T)$-chains for arbitrarily small $\epsilon$ and $T$.

\begin{restatable}[Partial order on chain components]{theorem}{PartialOrderChainComponents}
\label{thm:partial_order_chain_components}
    Let $\phi$ be a flow on a metric space and $A_1, A_2$ be chain components of the flow. Then the relation defined by  $A_1 \leq_C A_2$ is a partial order.
\end{restatable}
\begin{proof}
    \ProofInSM{{sec:proof_partial_order_chain_components}}
\end{proof}

We will be focusing on  minimal elements of this partial order, i.e., chain components $A$ such that there does not exist any chain component $B$ such that $B\leq_C A$.
We call such chain components \emph{sink chain components}.
\begin{definition}[Sink chain components]
    A chain component $A$ is called a sink chain component if there does not exist any  chain component $B\neq A$ such that $B\leq_C A$.   
\end{definition}

We can now define the useful notion of asymptotically stable sink chain components, which relies on the notions of Lyapunov stable, asymptotically stable, and attracting sets.
\begin{definition}[Lyapunov stable set]
    Let $\phi$ be a flow on a metric space $(X,d)$. A set $A \subset X$ is Lyapunov stable if for every neighborhood $O$ of $A$ there exists a neighborhood $O'$ of $A$ such that every trajectory that starts in $O'$ is contained in $O$; i.e., if $x\in O'$ then $\phi(t,x)\in O$ for all $t\geq0$. 
\end{definition}

\begin{definition}[Attracting set]
    Set $A$ is attracting if there exists a neighborhood $O$ of $A$ such that every trajectory starting in $O$ converges to $A$.
\end{definition}

\begin{definition}[Asymptotically stable set]
    A set is called asymptotically stable if it is both Lyapunov stable and attracting. 
\end{definition}

\begin{definition}[Asymptotically stable sink chain component]
    Chain component $A$ is called an asymptotically stable sink chain component if it is both  a sink chain component and an asymptotically stable set.
\end{definition}

\subsubsection{Markov-Conley chains}\label{subsec:MCC}
Although we wish to study asymptotically stable sink chain components, it is difficult to do so theoretically as we do not have an exact characterization of their geometry and the behavior of the dynamics inside them.
This is a rather difficult task to accomplish even experimentally. 
Replicator dynamics can be chaotic both in small and large games \citep{Sato02042002,galla2013complex}.
Even when their behavior is convergent for all initial conditions, the resulting equilibrium can be hard to predict and can be highly sensitive to initial conditions \cite{panageas2016average}.   
It is, therefore, not clear how to extract any meaningful information even from many trial runs of the dynamics.
These issues are exacerbated especially when games involve more than three or four strategies, where even visualization of trajectories becomes difficult.
While studies of these dynamics have been conducted for these low-dimensional cases \cite{bomze1983lotka,Bomze95}, very little is known about the geometry and topology of the limit behavior of replicator dynamics for general games, making it hard to even make informed guesses about whether the dynamics have, for practical reasons, converged to an invariant subset (i.e., a sink chain component). 

Instead of studying the actual dynamics, a computationally amenable alternative is to use a discrete-time discrete-space approximation with similar limiting dynamics, but which can be directly and efficiently analyzed.
We will start off by the most crude (but still meaningful) such approximations: a set of Markov chains whose state-space is the set of pure strategy profiles of the game.
We refer to each of these Markov chains as a \emph{Markov-Conley chain}, and prove in \cref{thm:MCC} that a finite number of them exist in any game under the replicator dynamics (or variants thereof).

Let us now formally define the Markov-Conley chains of a game, which relies on the notions of the response graph of a game and its sink strongly connected components.

\begin{definition}[Strictly and weakly better response]
    Let $s_i, s_j \in \prod_k S^k$ be any two pure strategy profiles of the game, which differ in the strategy of a single player $k$.
    Strategy $s_j$ is a strictly (respectively, weakly) better response than $s_i$ for player $k$ if her payoff at $s_j$ is larger than (respectively, at least as large as) her payoff at $s_i$.
\end{definition}

\begin{definition}[Response graph of a game]
    The response graph of a game $G$ is a directed graph whose vertex set coincides with the set of pure strategy profiles of the game, $\prod_k S^k$.
    Let $s_i, s_j \in \prod_k S^k$ be any two pure strategy profiles of the game.
    We include a directed edge from $s_i$ to $s_j$ if $s_j$ is a weakly better response for player $k$ as compared to $s_i$.
\end{definition}

\begin{definition}[Strongly connected components]
    The strongly connected components of a directed graph are the maximal subgraphs wherein there exists a path between each pair of vertices in the subgraph. 
\end{definition}

\begin{definition}[Sink strongly connected components]
    The sink strongly connected components of a directed graph are the strongly connected components with no out-going edges.
\end{definition}

The response graph of a game has a finite number of sink strongly connected components.
If such a component is a singleton, it is a pure Nash equilibrium by definition.

\begin{definition}[Markov-Conley chains (MCCs) of a game]\label{def:mcc}
    A Markov-Conley chain of a game $G$ is an irreducible Markov chain, the state space of which is a sink strongly connected component of the response graph associated with $G$.
    Many MCCs may exist for a given game $G$.
    In terms of the transition probabilities out of a node $s_i$ of each MCC, a canonical way to define them is as follows: with some probability, the node self-transitions. 
    The rest of the probability mass is split between all strictly and weakly improving responses of all players. Namely, the probability of strictly improving responses for all players are set equal to each other, and transitions between strategies of equal payoff happen with a smaller probability also equal to each other for all players. 
    

\end{definition}
When the context is clear, we sometimes overload notation and refer to the set of pure strategy profiles in a sink strongly connected component (as opposed to the Markov chain over them) as an MCC.
The structure of the transition probabilities introduced in \cref{def:mcc} has the advantage that it renders the MCCs invariant under arbitrary positive affine transformations of the payoffs; i.e., the resulting theoretical and empirical insights are insensitive to such transformations, which is a useful desideratum for a game-theoretic solution concept.
There may be alternative definitions of the transition probabilities that may warrant future exploration.

MCCs can be understood as a discrete approximation of the chain components of continuous-time dynamics (hence the connection to Conley's Theorem). 
The following theorem formalizes this relationship, and establishes finiteness of MCCs:
\begin{restatable}{theorem}{TheoremMCC}
\label{thm:MCC}
    Let $\phi$ be the replicator flow when applied to a $K$-person game.
    The number of asymptotically stable sink chain components is finite. Specifically, every asymptotically stable sink chain component contains at least one MCC; each MCC is contained in exactly one chain component. 
\end{restatable}
\begin{proof}
    \ProofInSM{{sec:proof_MCC}}
\end{proof}
The notion of MCCs is thus used as a stepping stone, a computational handle that aims to mimic the long term behavior of replicator dynamics in general games.
Similar results to \cref{thm:MCC} apply for several variants of replicator dynamics~\citep{Weibull97} as long as the dynamics are volume preserving in the interior of the state space, preserve the support of mixed strategies, and the dynamics act myopically in the presence of two strategies/options with fixed payoffs (i.e., if they have different payoffs converge to the best, if they have the same payoffs remain invariant). 

\subsection{From Markov-Conley chains to the Discrete-time Macro-model}\label{sec:mcc_to_discrete}
The key idea behind the ordering of agents we wish to compute is that the evolutionary fitness/performance of a specific strategy should be reflected by how often it is being chosen by the system/evolution.
We have established the solution concept of Markov-Conley chains (MCCs) as a discrete-time sparse-discrete-space analogue of the continuous-time replicator dynamics, which capture these long-term recurrent behaviors for general meta-games (see \cref{thm:MCC}).
MCCs are attractive from a computational standpoint: they can be found efficiently in all games by computing the sink strongly connected components of the response graph, addressing one of the key criticisms of Nash equilibria.
However, similar to Nash equilibria, even simple games may have many MCCs (e.g., five in the coordination game of \cref{fig:rd_coordination_game}). 
The remaining challenge is, thus, to solve the MCC selection problem.

One of the simplest ways to resolve the MCC selection issue is to introduce noise in our system and study a stochastically perturbed version, such that the overall Markov chain is irreducible and therefore has a unique stationary distribution that can be used for our rankings.
Specifically, we consider the following stochastically perturbed model: we choose an agent $k$ at random, and, if it is currently playing strategy $s^k_i$, we choose one of its strategies $s^k_j$ at random and set the new system state to be $\epsilon (s^k,s^{-k}) + (1-\epsilon) (s^k_j,s^{-k})$. 
Remarkably, these perturbed dynamics correspond closely to the macro-model introduced in \cref{sec:multipop_model} for a particularly large choice of ranking-intensity value $\alpha$:
\begin{restatable}{theorem}{InfinitePopAlpha}\label{thm:inf_pop_alpha}
    In the limit of infinite ranking-intensity $\alpha$, the Markov chain associated with the generalized multi-population model introduced in \cref{sec:multipop_model} coincides with the MCC.
\end{restatable}
\begin{proof}
    \ProofInSM{{sec:proof_inf_pop_alpha}}
\end{proof}
A low ranking-intensity ($\alpha \ll 1$) corresponds to the case of weak selection, where a weak mutant strategy can overtake a given population. 
A large ranking-intensity, on the other hand, ensures that the probability that a sub-optimal strategy overtakes a given population is close to zero, which corresponds closely to the MCC solution concept. 
In practice, setting the ranking-intensity to infinity may not be computationally feasible;
in this case, the underlying Markov chain may be reducible and the existence of a unique stationary distribution (which we use for our rankings) may not be guaranteed.
To resolve the MCC selection problem, we require a perturbed model, but one with a large enough ranking-intensity $\alpha$ such that it approximates an MCC, but small enough such that the MCCs remain connected.
By introducing this perturbed version of Markov-Conley chains, the resulting Markov chain is now irreducible (per \cref{property:unique_pi}). 
The long-term behavior is thus captured by the unique stationary distribution under the large-$\alpha$ limit.
Our so-called $\alpha$-Rank evaluation method then corresponds to the ordering of the agents in this particular stationary distribution.
The perturbations introduced here imply the need for a sweep over the ranking-intensity parameter $\alpha$ -- a single hyperparameter -- which we find to be computationally feasible across all of the large-scale games we analyze using $\alpha$-Rank.

The combination of \cref{thm:MCC} and \cref{thm:inf_pop_alpha} yields a unifying perspective involving a chain of models of increasing complexity: the continuous-time replicator dynamics is on one end, our generalized discrete-time concept is on the other, and MCCs are the link in between.

\section{Results}\label{sec:results}
In the following we summarize our generalized ranking model and the main theoretical and empirical results. We start by outlining how the $\alpha$-Rank procedure exactly works. Then we continue with illustrating $\alpha$-Rank in a number of canonical examples. We continue with some deeper understanding of $\alpha$-Rank's evolutionary dynamics model by introducing some further intuitions and theoretical results, and we end with an empirical validation of $\alpha$-Rank in various domains.

\subsection{$\alpha$-Rank: Evolutionary Ranking of Strategies}
We first detail the $\alpha$-Rank algorithm, then provide some insights and intuitions to further facilitate the understanding of our ranking method and solution concept.

\subsubsection{Algorithm}
Based on the dynamical concepts of chain recurrence and MCCs established, we now detail a descriptive method, titled $\alpha$-Rank, for computing strategy rankings in a multi-agent interaction:

\begin{enumerate}
    \item Construct the meta-game payoff table $M^k$ for each population $k$ from data of multi-agent interactions, or from running game simulations.

    \item Compute the transition matrix $C$ as outlined in \cref{sec:multipop_model}. Per the discussions in \cref{sec:mcc_to_discrete}, one must use a sufficiently large ranking-intensity value $\alpha$ in \cref{eq:fermi_distr_multipop}; this ensures that $\alpha$-Rank preserves the ranking of strategies with closest correspondence to the MCC solution concept. 
    As a large enough value is dependent on the domain under study, a useful heuristic is to conduct a sweep over $\alpha$, starting from a small value and increasing it exponentially until convergence of rankings.
    
    \item Compute the unique stationary distribution, $\pi$, of transition matrix $C$. Each element of the stationary distribution corresponds to the time the populations spend in a given strategy profile. 
    
    \item Compute the agent rankings, which correspond to the ordered masses of the stationary distribution $\pi$. The stationary distribution mass for each agent constitutes a `score' for it (as might be shown, e.g., on a leaderboard).
\end{enumerate}

\subsubsection{$\alpha$-Rank and MCCs as a Solution Concept: A Paradigm Shift}

The solution concept of MCCs is foundationally distinct from that of the Nash equilibrium.
The Nash equilibrium is rooted in classical game theory, which not only models the interactions in multi-agent systems, but is also normative in the sense that it prescribes how a player should behave based on the assumption of individual rationality \cite{Shoham:answer,Gintis09,Weibull97}.
Besides classical game theory making strong assumptions regarding the rationality of players involved in the interaction, there exist many fundamental limitations with the concept of a Nash equilibrium: intractability (computing a Nash is PPAD-complete), equilibrium selection, and the incompatibility of this static concept with the dynamic behaviors of agents in interacting systems. 
To compound these issues, even methods that aim to compute an approximate Nash are problematic: 
a typical approach is to use exploitability to measure deviation from Nash and as such use it as a method to closely approximate one; the problem with this is that it is also intractable for large games (typically the ones we are interested in), and there even still remain issues with using exploitability as a measure of strategy strength (e.g., see \cite{DavisBB14}).
Overall, there seems little hope of deploying the Nash equilibrium as a solution concept for the evaluation of agents in general large-scale (empirical) games.

The concept of an MCC, by contrast, embraces the dynamical systems perspective, in a manner similar to evolutionary game theory. 
Rather than trying to capture the strategic behavior of players in an equilibrium, we deploy a dynamical system based on the evolutionary interactions of agents that captures and describes the long-term behavior of the players involved in the interaction. 
As such, our approach is descriptive rather than prescriptive, in the sense that it is not prescribing the strategies that one should play;
rather, our approach provides useful information regarding the strategies that are evolutionarily non-transient (i.e., resistant to mutants), and highlights the remaining strategies that one might play in practice.
To understand MCCs requires a shift away from the classical models described above for games and multi-agent interactions.
Our new paradigm is to allow the dynamics to roll out and enable strong (i.e., non-transient) agents to emerge and weak (i.e, transient) agents to vanish naturally through their long-term interactions.
The resulting solution concept not only permits an automatic ranking of agents' evolutionary strengths, but is powerful both in terms of computability and usability: 
our rankings are guaranteed to exist, can be computed tractably for any game, and involve no equilibrium selection issues as the evolutionary process converges to a unique stationary distribution.
Nash tries to identify \emph{static} single points in the simplex that capture simultaneous best response behaviors of agents, but comes with the range of complications mentioned above.
On the other hand, the support of our stationary distribution captures the strongest non-transient agents, which may be interchangeably played by interacting populations and therefore constitute a \emph{dynamic} output of our approach.

Given that both Nash and MCCs share a common foundation in the notion of a best response (i.e., simultaneous best responses for Nash, and the sink components of a best response graph for MCCs), it is interesting to consider the circumstances under which the two concepts coincide.
There do, indeed, exist such exceptional circumstances: for example, for a potential game,  every better response sequence converges to a (pure) Nash equilibrium, which coincides with an MCC. 
However, even in relatively simple games, differences between the two solution concepts are expected to occur in general due to the inherently dynamic nature of MCCs (as opposed to Nash).
For example, in the Biased Rock-Paper-Scissors game detailed in \cref{sec:biased_rps}, the Nash equilibrium and stationary distribution are not equivalent due to the cyclical nature of the game; each player's symmetric Nash is $(\frac{1}{16},\frac{5}{8},\frac{5}{16})$, whereas the stationary distribution is $(\frac{1}{3},\frac{1}{3},\frac{1}{3})$.
The key difference here is that whereas Nash is prescriptive and tells players which strategy mixture to use, namely $(\frac{1}{16},\frac{5}{8},\frac{5}{16})$, assuming rational opponents, $\alpha$-Rank is descriptive in the sense that it filters out evolutionary transient strategies and yields a ranking of the remaining strategies in terms of their long-term survival.
In the Biased Rock-Paper-Scissors example, $\alpha$-Rank reveals that all three strategies are equally likely to persist in the long-term as they are part of the same sink strongly connected component of the response graph.
In other words, the stationary distribution mass (i.e., the $\alpha$-Rank score) on a particular strategy is indicative of its resistance to being invaded by any other strategy, including those in the distribution support.
In the case of the Biased Rock-Paper-Scissors game, this means that the three strategies are equally likely to be invaded by a mutant, in the sense that their outgoing fixation probabilities are equivalent.
In contrast to our evolutionary ranking, Nash comes without any such stability properties (e.g., consider the interior mixed Nash in \cref{fig:rd_coordination_game}).
Even computing Evolutionary Stable Strategies (ESS) \cite{Weibull97}, a refinement of Nash equilibria, is intractable \cite{Conitzer18,Etessami2008}. 
In larger games (e.g., AlphaZero in \cref{sec:results_alphazero}), the reduction in the number of agents that are resistant to mutations is more dramatic (in the sense of the stationary distribution support size being much smaller than the total number of agents) and less obvious (in the sense that more-resistant agents are not always the ones that have been trained for longer).
In summary, the strategies chosen by our approach are those favored by evolutionary selection, as opposed to the Nash strategies, which are simultaneous best-responses.

\subsection{Conceptual Examples}\label{sec:conceptual_example_multipop}
We revisit the earlier conceptual examples of Rock-Paper-Scissors and Battle of the Sexes from \cref{sec:markov_chain_conceptual_example} to illustrate the rankings provided by the $\alpha$-Rank methodology.
We use a population size of $m=50$ in our evaluations. 

\subsubsection{Rock-Paper-Scissors}\label{sec:results_rps}
\begin{figure}[t]
    \centering
    \begin{subfigure}[b]{0.7\textwidth}
        \centering
        \includegraphics[width=1\textwidth]{figs/mcc_rock_paper_scissors.pdf}
        \caption{Discrete-time dynamics.}
        \label{fig:mcc_rock_paper_scissors}
     \end{subfigure}\\
    \begin{subfigure}[t]{0.6\textwidth}
        \centering
        \includegraphics[width=0.8\textwidth]{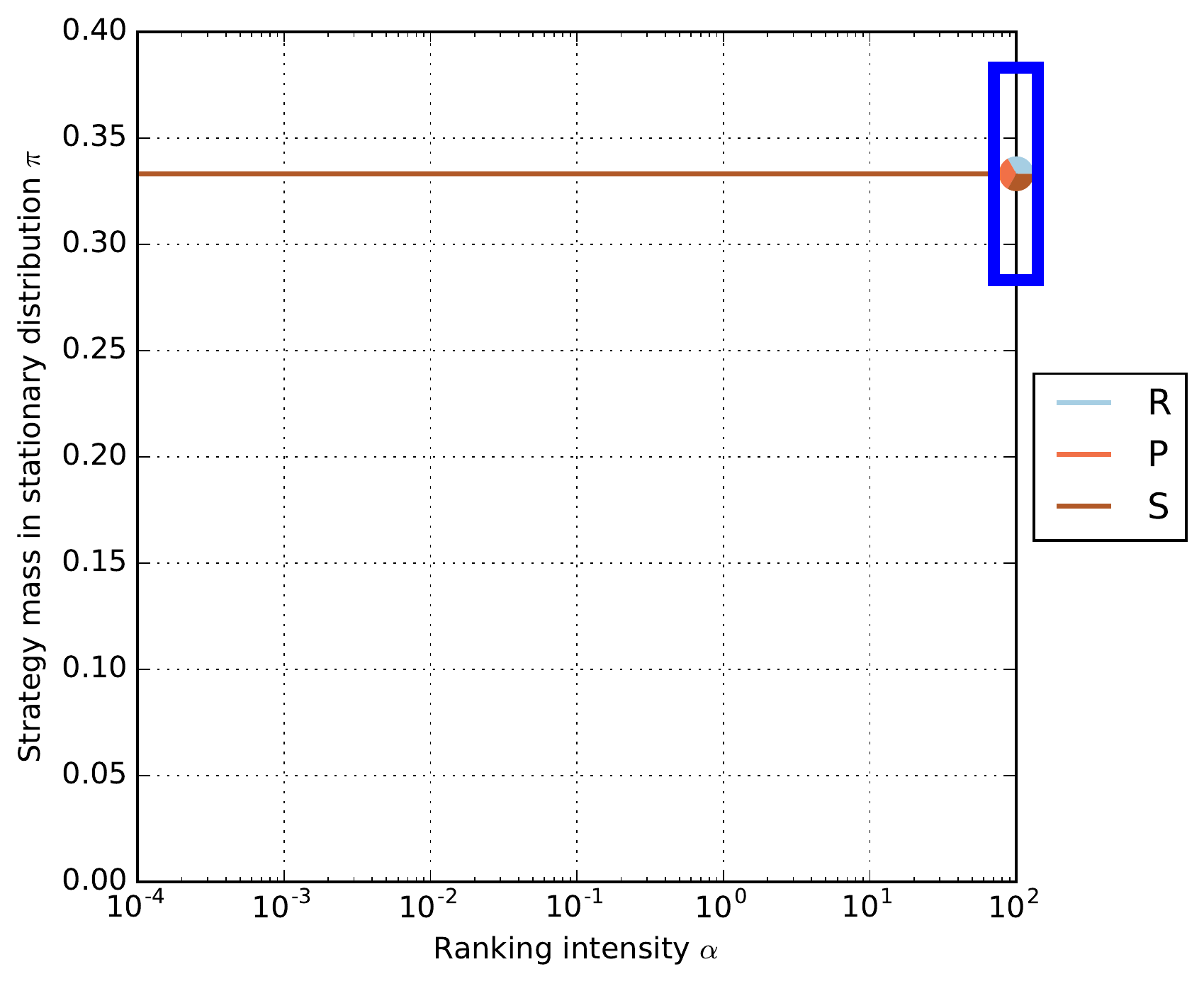}
        \caption{Ranking-intensity sweep}
        \label{fig:pi_vs_alpha_rock_paper_scissors}
    \end{subfigure}
    \hfill
    \begin{subtable}[b]{0.33\textwidth}
        \centering
        \def\arraystretch{1.2}
        \begin{tabular}{ *{3}{c} }
        \toprule
        Agent & Rank & Score\\
        \midrule
        \rowcolor{MyBlue!33.0} \contour{white}{$R$}& \contour{white}{$1$}& \contour{white}{$0.33$}\\
        \rowcolor{MyBlue!33.0} \contour{white}{$P$}& \contour{white}{$1$}& \contour{white}{$0.33$}\\
        \rowcolor{MyBlue!33.0} \contour{white}{$S$}& \contour{white}{$1$}& \contour{white}{$0.33$}\\
        \bottomrule
        \end{tabular}
        \caption{$\alpha$-Rank results.}
        \label{table:alpharank_rock_paper_scissors}
    \end{subtable}
    \caption{Rock-Paper-Scissors game.}
    \label{fig:results_rock_paper_scissors}
\end{figure}

In the Rock-Paper-Scissors game, recall the cyclical nature of the discrete-time Markov chain (shown in \cref{fig:mcc_rock_paper_scissors}) for a fixed value of ranking-intensity parameter, $\alpha$.
We first investigate the impact of the ranking-intensity on overall strategy rankings, by plotting the stationary distribution as a function of $\alpha$ in \cref{fig:pi_vs_alpha_rock_paper_scissors}.
The result is that the population spends $\frac{1}{3}$ of its time playing each strategy regardless of the value of $\alpha$, which is in line with intuition due to the cyclical best-response structure of the game's payoffs.
The Nash equilibrium, for comparison, is also $(\frac{1}{3}, \frac{1}{3}, \frac{1}{3})$.
The $\alpha$-Rank output \cref{table:alpharank_rock_paper_scissors}, which corresponds to a high value of $\alpha$, thus indicates a tied ranking for all three strategies, also in line with intuition.

\subsubsection{Biased Rock-Paper-Scissors}\label{sec:biased_rps}
\begin{figure}[t]
    \centering
    \begin{subfigure}[b]{0.33\textwidth}
        \begin{tabular}{cc|c|c|c|}
            & \multicolumn{1}{c}{} & \multicolumn{3}{c}{Player 2}\\
            & \multicolumn{1}{c}{} & \multicolumn{1}{c}{$R$}  & \multicolumn{1}{c}{$P$} & \multicolumn{1}{c}{$S$} \\\cline{3-5}
            \multirow{3}*{Player 1}  & $R$ & $0$ & $-0.5$ & $1$ \\\cline{3-5}
                                     & $P$ & $0.5$ & $0$ & $-0.1$ \\\cline{3-5}
                                     & $S$ & $-1$ & $0.1$ & $0$ \\\cline{3-5}
        \end{tabular}
        \caption{Payoff matrix.}
        \label{table:biased_rps_payoffs}
    \end{subfigure}\\
    \begin{subfigure}[b]{0.7\textwidth}
        \centering
        \includegraphics[width=1\textwidth]{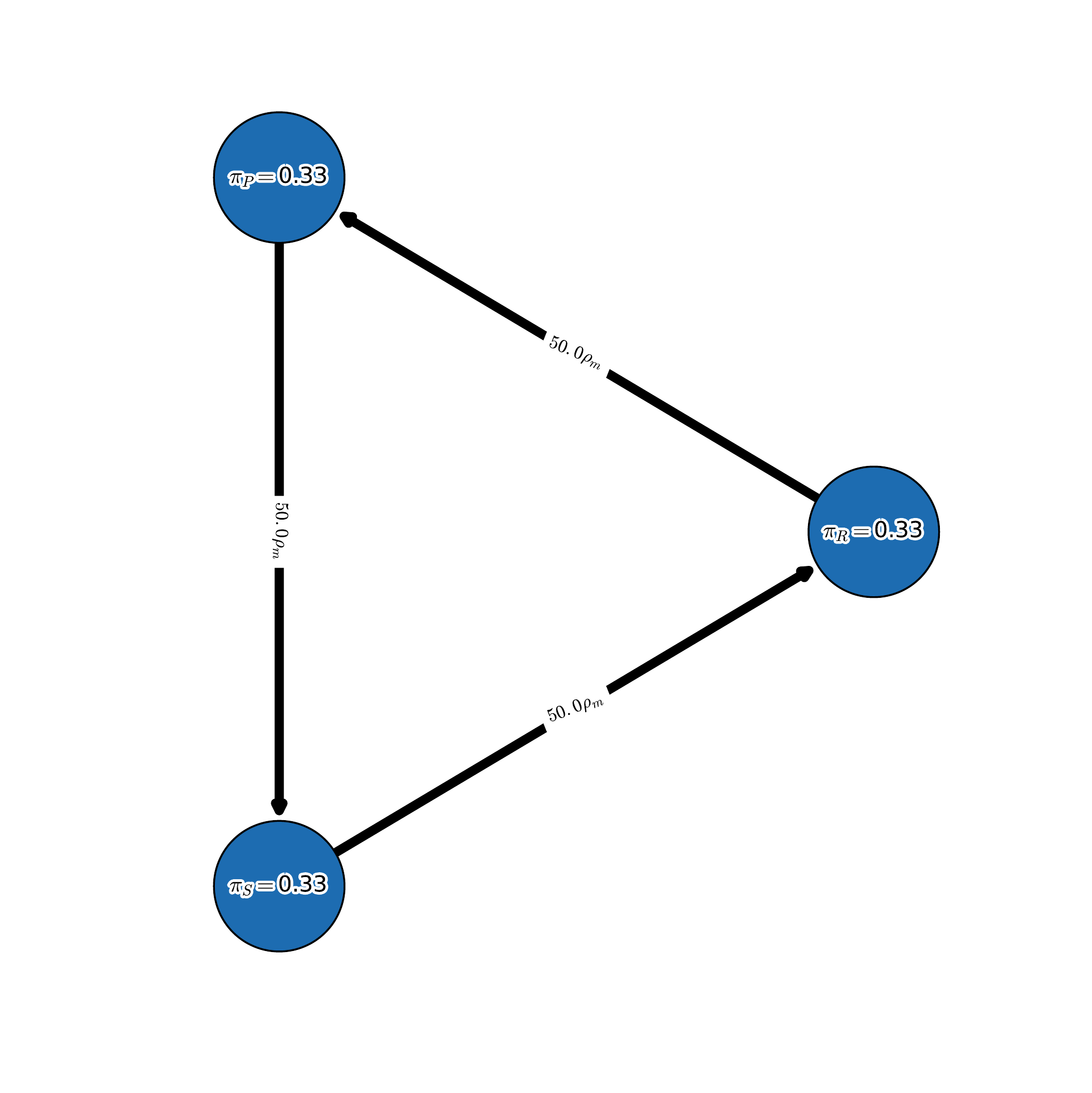}
        \caption{Discrete-time dynamics.}
        \label{fig:mcc_biased_rock_paper_scissors}
     \end{subfigure}\\
    \begin{subfigure}[t]{0.6\textwidth}
        \centering
        \includegraphics[width=0.8\textwidth]{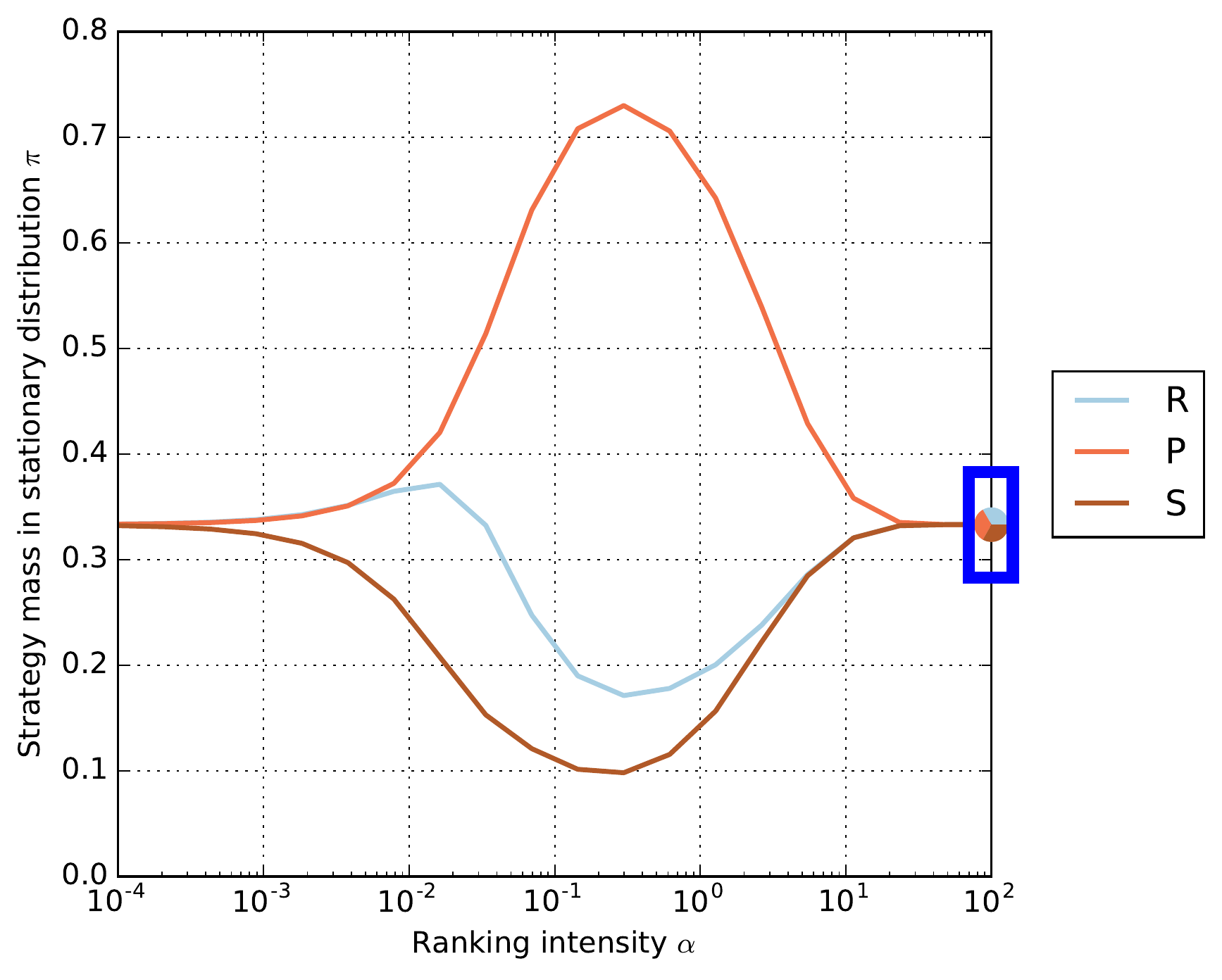}
        \caption{Ranking-intensity sweep.}
        \label{fig:pi_vs_alpha_biased_rock_paper_scissors}
    \end{subfigure}
    \hfill
    \begin{subtable}[b]{0.33\textwidth}
            \centering
            \def\arraystretch{1.2}
            \begin{tabular}{ *{3}{c} }
            \toprule
            Agent & Rank & Score\\
            \midrule
            \rowcolor{MyBlue!33.0} \contour{white}{$R$}& \contour{white}{$1$}& \contour{white}{$0.33$}\\
            \rowcolor{MyBlue!33.0} \contour{white}{$P$}& \contour{white}{$1$}& \contour{white}{$0.33$}\\
            \rowcolor{MyBlue!33.0} \contour{white}{$S$}& \contour{white}{$1$}& \contour{white}{$0.33$}\\
            \bottomrule
            \end{tabular}
        \caption{$\alpha$-Rank results.}
        \label{table:alpharank_biased_rock_paper_scissors}
    \end{subtable}
    \caption{Biased Rock-Paper-Scissors game.}
\end{figure}

Consider now the game of Rock-Paper-Scissors, but with biased payoffs (shown in \cref{table:biased_rps_payoffs}).
The introduction of the bias moves the Nash from the center of the simplex towards one of the corners, specifically $(\frac{1}{16},\frac{5}{8},\frac{5}{16})$ in this case.
It is worthwhile to investigate the corresponding variation of the stationary distribution masses as a function of the ranking-intensity $\alpha$ (\cref{fig:pi_vs_alpha_biased_rock_paper_scissors}) in this case.
As evident from the fixation probabilities \cref{eq:multipop_fixation_prob} of the generalized discrete-time model, very small values of $\alpha$ cause the raw values of payoff to have a very low impact on the dynamics captured by discrete-time Markov chain; in this case, any mutant strategy has the same probability of taking over the population, regardless of the current strategy played by the population.
This corresponds well to \cref{fig:pi_vs_alpha_biased_rock_paper_scissors}, where small $\alpha$ values yield stationary distributions close to $\pi = (\frac{1}{3},\frac{1}{3},\frac{1}{3})$.

As $\alpha$ increases, payoff values play a correspondingly more critical role in dictating the long-term population state; in \cref{fig:pi_vs_alpha_biased_rock_paper_scissors}, the population tends to play Paper most often within this intermediate range of $\alpha$.
Most interesting to us, however, is the case where $\alpha$ increases to the point that our discrete-time model bears a close correspondence to the MCC solution concept (per \cref{thm:inf_pop_alpha}).
In this limit of large $\alpha$, the striking outcome is that the stationary distribution once again converges to $(\frac{1}{3},\frac{1}{3},\frac{1}{3})$.
Thus, $\alpha$-Rank yields the high-level conclusion that in the long term, a monomorphic population playing any of the 3 given strategies can be completely and repeatedly displaced by a rare mutant, and as such assigns the same ranking to all strategies (\cref{table:alpharank_biased_rock_paper_scissors}).
This simple example illustrates perhaps the most important trait of the MCC solution concept and resulting $\alpha$-Rank methodology: they capture the fundamental dynamical structure of games and long-term intransitivities that exist therein, with the rankings produced corresponding to the dynamical \emph{strategy space consumption} or basins of attraction of strategies.

\subsubsection{Battle of the Sexes}\label{results:bots}

\begin{figure}[t]
    \centering
    \begin{subfigure}[t]{1\textwidth}
        \centering
        \includegraphics[width=0.7\textwidth]{figs/mcc_battle_of_the_sexes.pdf}
        \caption{Discrete-time dynamics (see (\subref{table:alpharank_battle_of_the_sexes}) for node-wise scores corresponding to stationary distribution masses).}
        \label{fig:mcc_battle_of_the_sexes}
    \end{subfigure}\\
    \begin{subfigure}[t]{0.6\textwidth}
        \centering
        \includegraphics[width=0.8\textwidth]{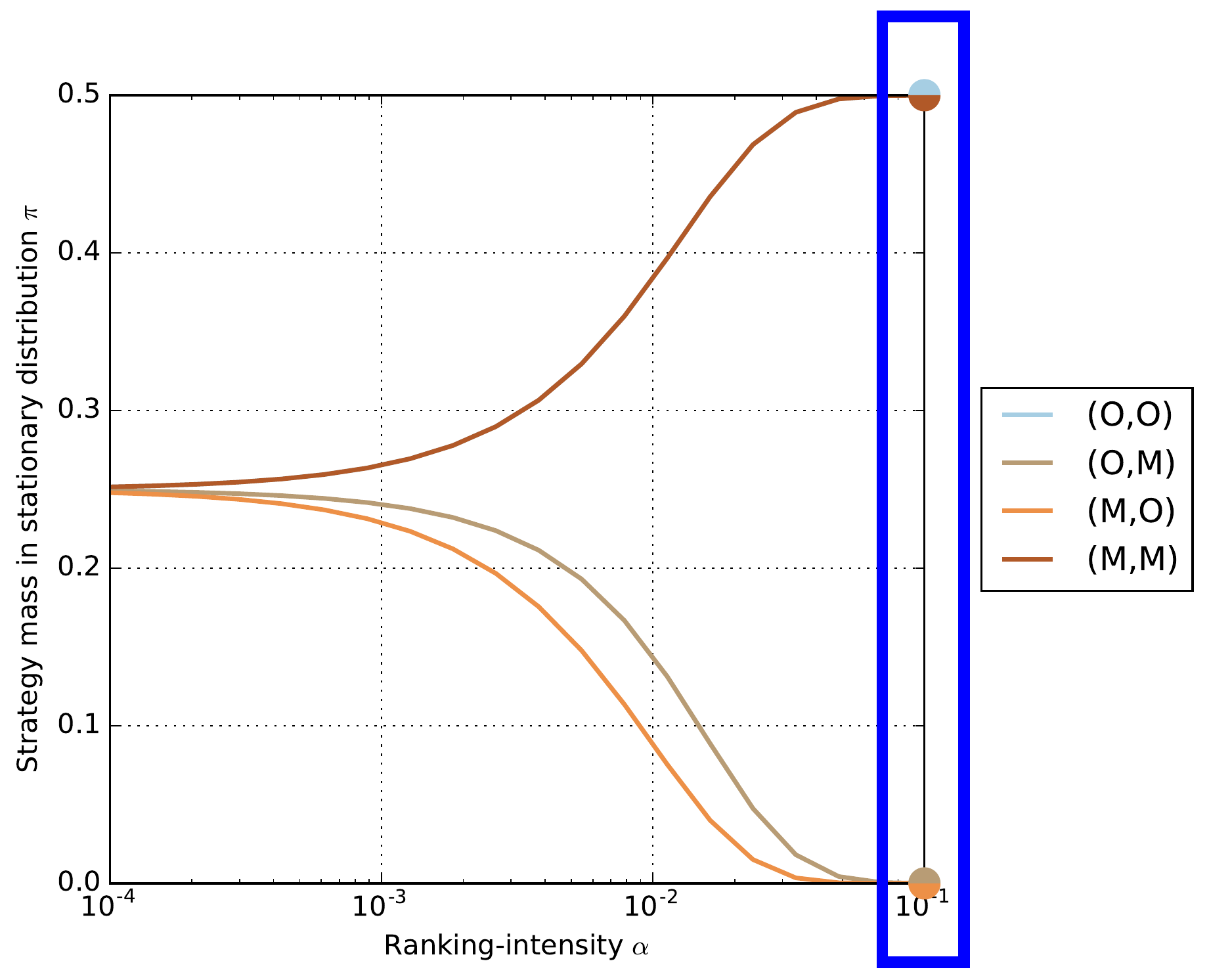}
        \caption{Ranking-intensity sweep.}
        \label{fig:pi_vs_alpha_battle_of_the_sexes}
    \end{subfigure}
    \hfill
    \begin{subtable}[b]{0.33\textwidth}
            \centering
            \def\arraystretch{1.2}
            \begin{tabular}{ *{3}{c} }
            \toprule
            Agent & Rank & Score\\
            \midrule
            \rowcolor{MyBlue!50.0} \contour{white}{$(O,O)$}& \contour{white}{$1$}& \contour{white}{$0.5$}\\
            \rowcolor{MyBlue!50.0} \contour{white}{$(M,M)$}& \contour{white}{$1$}& \contour{white}{$0.5$}\\
            \rowcolor{MyBlue!0.0} \contour{white}{$(O,M)$}& \contour{white}{$2$}& \contour{white}{$0.0$}\\
            \rowcolor{MyBlue!0.0} \contour{white}{$(M,O)$}& \contour{white}{$2$}& \contour{white}{$0.0$}\\
            \bottomrule
            \end{tabular}
        \caption{$\alpha$-Rank results.}
        \label{table:alpharank_battle_of_the_sexes}
    \end{subtable}
    \caption{Battle of the Sexes.}
    \label{fig:battle_of_the_sexes}
\end{figure}

We consider next an example of $\alpha$-Rank applied to an asymmetric game -- the Battle of the Sexes.
\Cref{fig:pi_vs_alpha_battle_of_the_sexes} plots the stationary distribution against ranking-intensity $\alpha$, where we again observe a uniform stationary distribution corresponding to very low values of $\alpha$.
As $\alpha$ increases, we observe the emergence of two sink chain components corresponding to strategy profiles $(O,O)$ and $(M,M)$, which thus attain the top $\alpha$-Rank scores in \cref{table:alpharank_battle_of_the_sexes}.
Note the distinct convergence behaviors of strategy profiles $(O,M)$ and $(M,O)$ in \cref{fig:pi_vs_alpha_battle_of_the_sexes}, where the stationary distribution mass on the $(M,O)$ converges to $0$ faster than that of $(O,M)$ for an increasing value of $\alpha$.
This is directly due to the structure of the underlying payoffs and the resulting differences in fixation probabilities.
Namely, starting from profile $(M,O)$, if either player deviates, that player increases their local payoff from $0$ to $3$. 
Likewise, if either player deviates starting from profile $(O,M)$, that player's payoff increases from $0$ to $2$.
Correspondingly, the fixation probabilities out of $(M,O)$ are higher than those out of $(O,M)$ (\cref{fig:mcc_battle_of_the_sexes}), and thus the stationary distribution mass on $(M,O)$  converges to $0$ faster than that of $(O,M)$ as $\alpha$ increases.
We note that these low-$\alpha$ behaviors, while interesting, have no impact on the final rankings computed in the limit of large $\alpha$ (\cref{table:alpharank_battle_of_the_sexes}).
We refer the interested reader to \cite{veller2017red} for a detailed analysis of the non-coordination components of the stationary distribution in mutualistic interactions, such as the Battle of the Sexes.

We conclude this discussion by noting that despite the asymmetric nature of the payoffs in this example, the computational techniques used by $\alpha$-Rank to conduct the evaluation are essentially identical to the simpler (symmetric) Rock-Paper-Scissors game.
This key advantage is especially evident in contrast to recent evaluation approaches that involve decomposition of a asymmetric game into multiple counterpart symmetric games, which must then be concurrently analyzed \cite{TuylsSym}.

\subsection{Theoretical Properties of $\alpha$-Rank}
This section presents key theoretical findings related to the structure of the underlying discrete-time model used in $\alpha$-Rank, and computational complexity of the ranking analysis. 
Proofs are presented in the Supplementary Material.

\begin{restatable}[Structure of $C$]{property}{MemoryBoundC}\label{thm:memory_bound_c}
 Given strategy profile $s_i$ corresponding to row $i$ of $C$, the number of valid profiles it can transition to is $1+\sum_{k}(|S^k|-1)$ (i.e., either $s_i$ self-transitions, or one of the populations $k$ switches to a different monomorphic strategy).
    The sparsity of $C$ is then,
    \begin{align}
        1 - \frac{|S|(1+\sum_{k}(|S^k|-1))}{|S|^2}. 
    \end{align}
\end{restatable}
Therefore, for games involving many players and strategies, transition matrix $C$ is large (in the sense that there exist $|S|$ states), but extremely sparse (in the sense that there exist only $1+\sum_k (|S^k|-1)$ outgoing edges from each state).
For example, in a $6$-wise interaction game where agents in each population have a choice over $4$ strategies, $C$ is 99.53\% sparse.

\begin{restatable}[Computational complexity of solving for $\pi$]{property}{StationaryComplexity}
\label{thm:stationary_complexity}
    The sparse structure of the Markov transition matrix $C$ (as identified in \cref{thm:memory_bound_c}) can be exploited to solve for the stationary distribution $\pi$ efficiently; specifically, computing the stationary distribution can be formulated as an eigenvalue problem, which can be computed in cubic-time in the number of total pure strategy profiles.
\end{restatable}
The $\alpha$-Rank method is, therefore, tractable, in the sense that it runs in polynomial time with respect to the total number of pure strategies.
This yields a major computational advantage,  in stark contrast to conducting rankings by solving for Nash (which is PPAD-complete for general-sum games \cite{Daskalakis06thecomplexity}, which our meta-games may be).

\subsection{Experimental Validation}
\begin{table}
    \centering
    \begin{tabular}{ccccc}
        \toprule
        Domain & Results & Symmetric? & \# of Populations & \# of Strategies \\
        \midrule
        Rock-Paper-Scissors & \cref{sec:results_rps} & \cmark & $1$ & $[3]$\\
        Biased Rock-Paper-Scissors & \cref{sec:biased_rps} & \cmark & $1$ & $[3]$\\
        Battle of the Sexes & \cref{results:bots} & \xmark & $2$ & $[2,2]$\\
        AlphaGo & \cref{sec:results_alphago} & \cmark & $1$ & $[7]$\\
        AlphaZero Chess & \cref{sec:results_alphazero} & \cmark & $1$ & $[56]$\\
        MuJoCo Soccer & \cref{sec:results_mujoco_soccer} & \cmark & $1$ & $[10]$\\
        Kuhn Poker & \cref{sec:results_kuhn_poker} & \xmark & $3$ & [4,4,4]\\
                   & \cref{sec:results_kuhn_poker} & \xmark & $4$ & [4,4,4,4]\\
        Leduc Poker & \cref{sec:results_psro} & \xmark & $2$ & $[3,3]$\\
        \bottomrule
    \end{tabular}
    \caption{Overview of multi-agent domains evaluated in this paper. These domains are extensive across multiple axes of complexity, and include symmetric and asymmetric games with different numbers of populations and ranges of strategies.}
    \label{table:results_summary}
\end{table}

In this section we provide a series of experimental illustrations of $\alpha$-Rank in a varied set of domains, including AlphaGo, AlphaZero Chess, MuJoCo Soccer, and both Kuhn and Leduc Poker.
As evident in \cref{table:results_summary}, the analysis conducted is extensive across multiple axes of complexity, as the domains considered include symmetric and asymmetric games with different numbers of populations and ranges of strategies.

\subsubsection{AlphaGo}\label{sec:results_alphago}

\begin{figure}[t]
    \centering
    \begin{subfigure}[t]{0.7\textwidth}
        \centering
        \includegraphics[width=1.\textwidth]{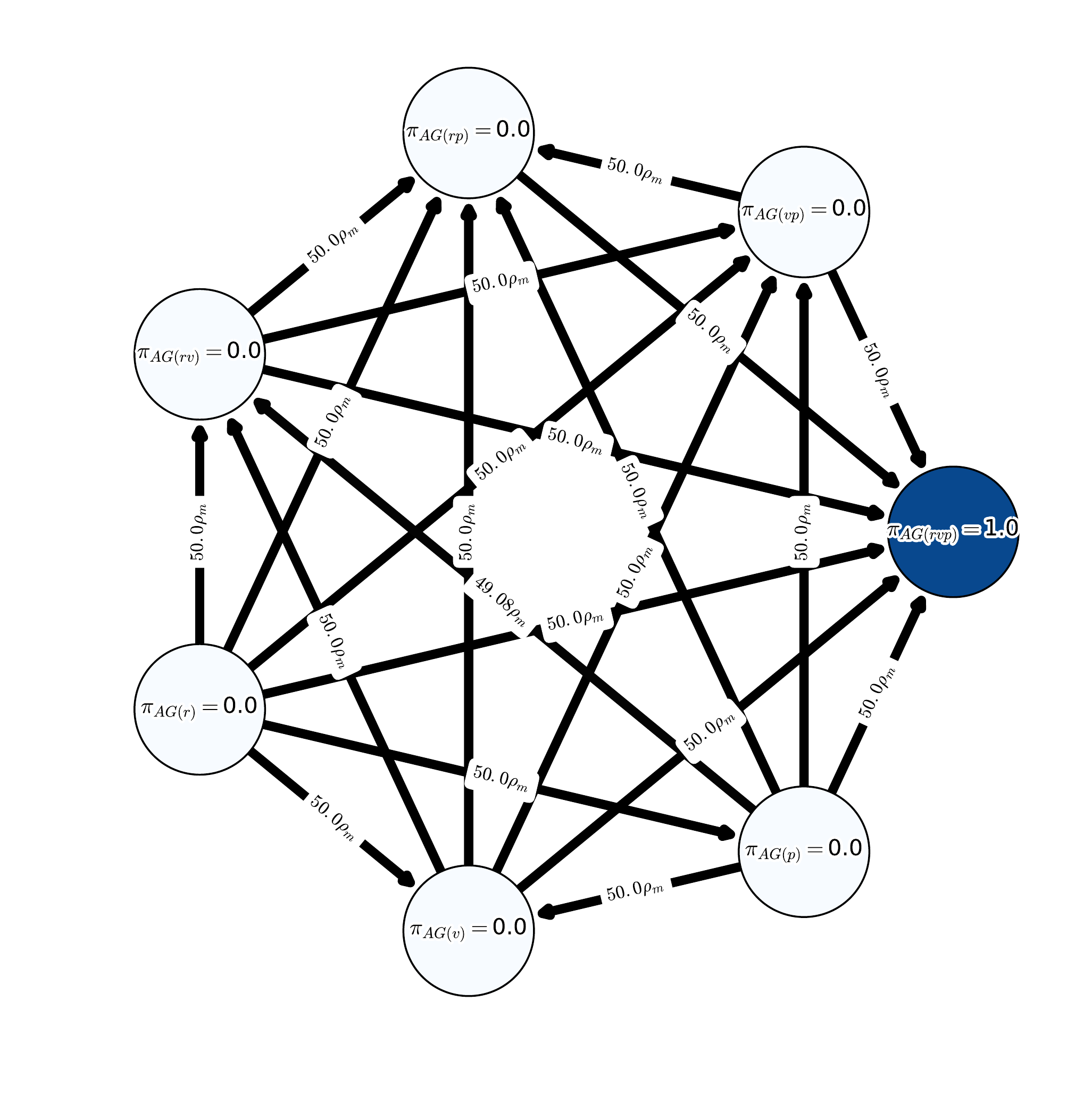}
        \caption{Discrete-time dynamics.}
        \label{fig:mcc_alphago}
    \end{subfigure}\\
    \begin{subfigure}[t]{0.6\textwidth}
        \centering
        \includegraphics[width=0.8\textwidth]{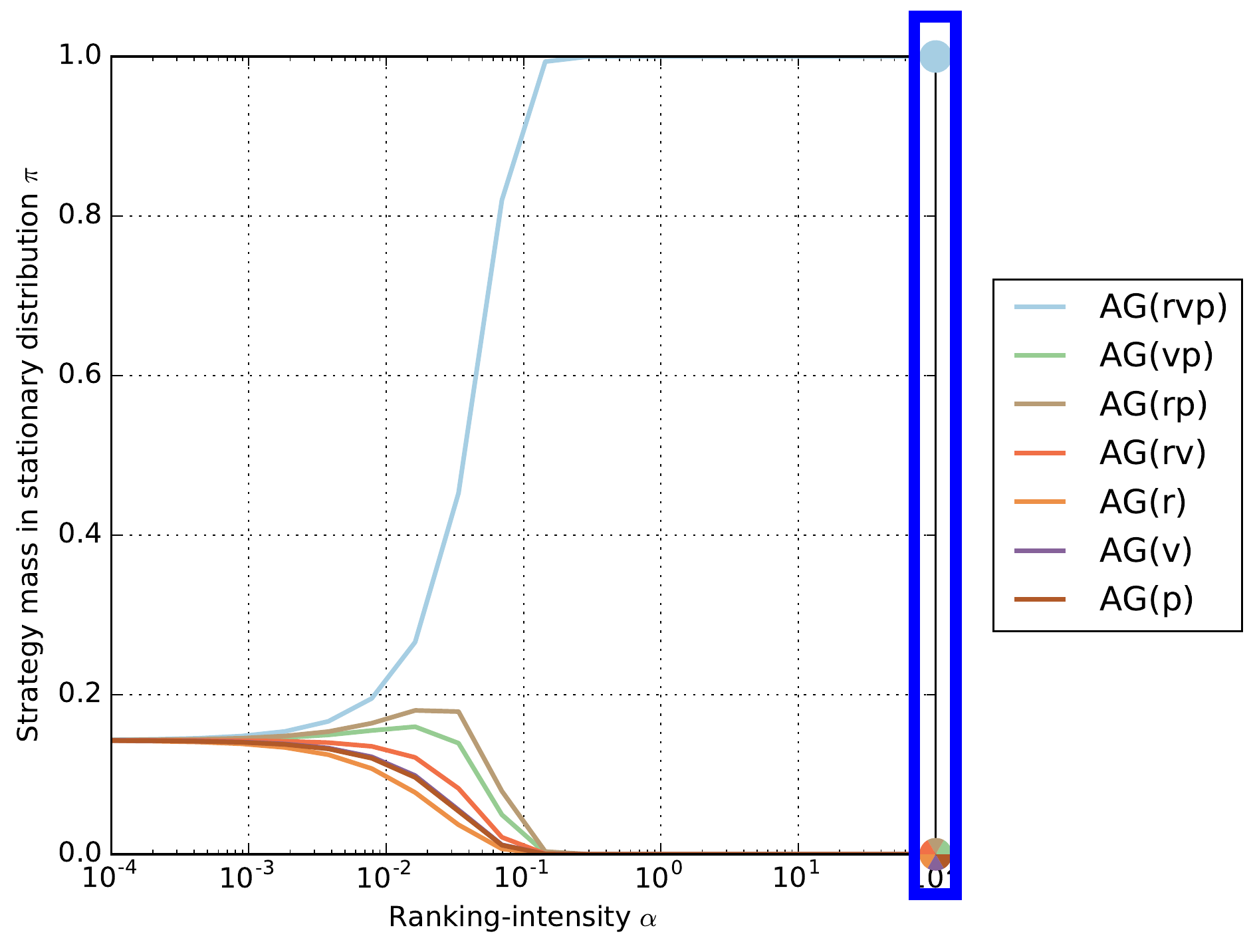}
        \caption{Ranking-intensity sweep.}
        \label{fig:pi_vs_alpha_alphago}
    \end{subfigure}
    \hfill
    \begin{subtable}[b]{0.33\textwidth}
            \centering
            \def\arraystretch{1.2}
            \begin{tabular}{ *{3}{c} }
            \toprule
            Agent & Rank & Score\\
            \midrule
            \rowcolor{MyBlue!100.0} \contour{white}{$AG(rvp)$}& \contour{white}{$1$}& \contour{white}{$1.0$}\\
            \rowcolor{MyBlue!0.0} \contour{white}{$AG(vp)$}& \contour{white}{$2$}& \contour{white}{$0.0$}\\
            \rowcolor{MyBlue!0.0} \contour{white}{$AG(rp)$}& \contour{white}{$2$}& \contour{white}{$0.0$}\\
            \rowcolor{MyBlue!0.0} \contour{white}{$AG(rv)$}& \contour{white}{$2$}& \contour{white}{$0.0$}\\
            \rowcolor{MyBlue!0.0} \contour{white}{$AG(r)$}& \contour{white}{$2$}& \contour{white}{$0.0$}\\
            \rowcolor{MyBlue!0.0} \contour{white}{$AG(v)$}& \contour{white}{$2$}& \contour{white}{$0.0$}\\
            \rowcolor{MyBlue!0.0} \contour{white}{$AG(p)$}& \contour{white}{$2$}& \contour{white}{$0.0$}\\
            \bottomrule
            \end{tabular}
        \caption{$\alpha$-Rank results.}
        \label{table:alpharank_alphago}
    \end{subtable}
    \caption{AlphaGo (Nature dataset).}
    \label{fig:results_alphago}
\end{figure}

In this example we conduct an evolutionary ranking of AlphaGo agents based on the data reported in \cite{DSilverHMGSDSAPL16}.
The meta-game considered here corresponds to a 2-player symmetric NFG with 7 AlphaGo agents: $AG(r)$, $AG(p)$, $AG(v)$, $AG(rv)$, $AG(rp)$, $AG(vp)$, and $AG(rvp)$, where $r$, $v$, and $p$ respectively denote the combination of \emph{rollouts}, \emph{value networks}, and/or \emph{policy networks} used by each variant.
The corresponding payoffs are the win rates for each pair of agent match-ups, as reported in Table 9 of \cite{DSilverHMGSDSAPL16}.

In \cref{table:alpharank_alphago} we summarize the rankings of these agents using the $\alpha$-Rank method.
$\alpha$-Rank is quite conclusive in the sense that the top agent, $AG(rvp)$, attains all of the stationary distribution mass, dominating all other agents.
Further insights into the pairwise agent interactions are revealed by visualizing the underlying Markov chain, shown in \cref{fig:mcc_alphago}.
Here the population flows (corresponding to the graph edges) indicate which agents are more evolutionarily viable than others.
For example, the edge indicating flow from $AG(r)$ to $AG(rv)$ indicates that the latter agent is stronger in the short-term of evolutionary interactions.
Moreover, the stationary distribution (corresponding to high $\alpha$ values in \cref{fig:pi_vs_alpha_alphago}) reveals that all agents but $AG(rvp)$ are transient in terms of the long-term dynamics, as a monomorphic population starting from any other agent node eventually reaches $AG(rvp)$.
In this sense, node $AG(rvp)$ constitutes an evolutionary stable strategy.
We also see in \cref{fig:mcc_alphago} that no cyclic behaviors occur in these interactions.
Finally, we remark that the recent work of \cite{Tuyls18} also conducted a meta-game analysis on these particular AlphaGo agents and drew similar conclusions to ours.
The key limitation of their approach is that it can only directly analyze interactions between triplets of agents, as they rely on visualization of the continuous-time evolutionary dynamics on a 2-simplex.
Thus, to draw conclusive results regarding the interactions of the full set of agents, they must concurrently conduct visual analysis of all possible 2-simplices (35 total in this case).
This highlights a key benefit of $\alpha$-Rank as it can  succinctly summarize agent evaluations with minimal intermediate human-in-the-loop analysis.

\subsubsection{AlphaZero}\label{sec:results_alphazero}

\begin{figure}[t]
    \centering
    \begin{subfigure}[t]{0.7\textwidth}
        \centering
        \includegraphics[width=1.\textwidth]{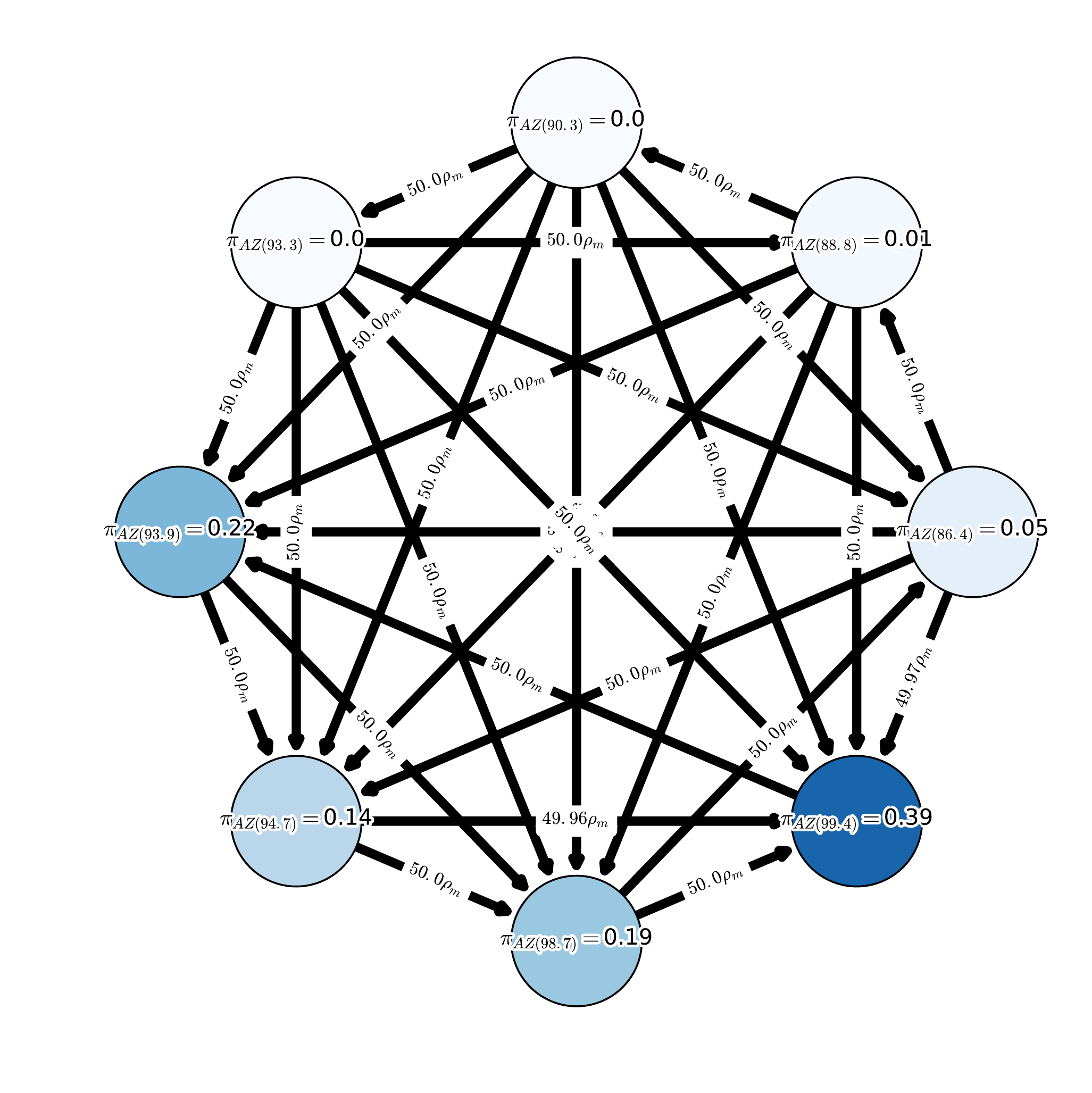}
        \caption{Discrete-time dynamics.}
        \label{fig:mcc_chess}
    \end{subfigure}\\
    \begin{subfigure}[t]{0.6\textwidth}
        \centering
        \includegraphics[width=0.8\textwidth]{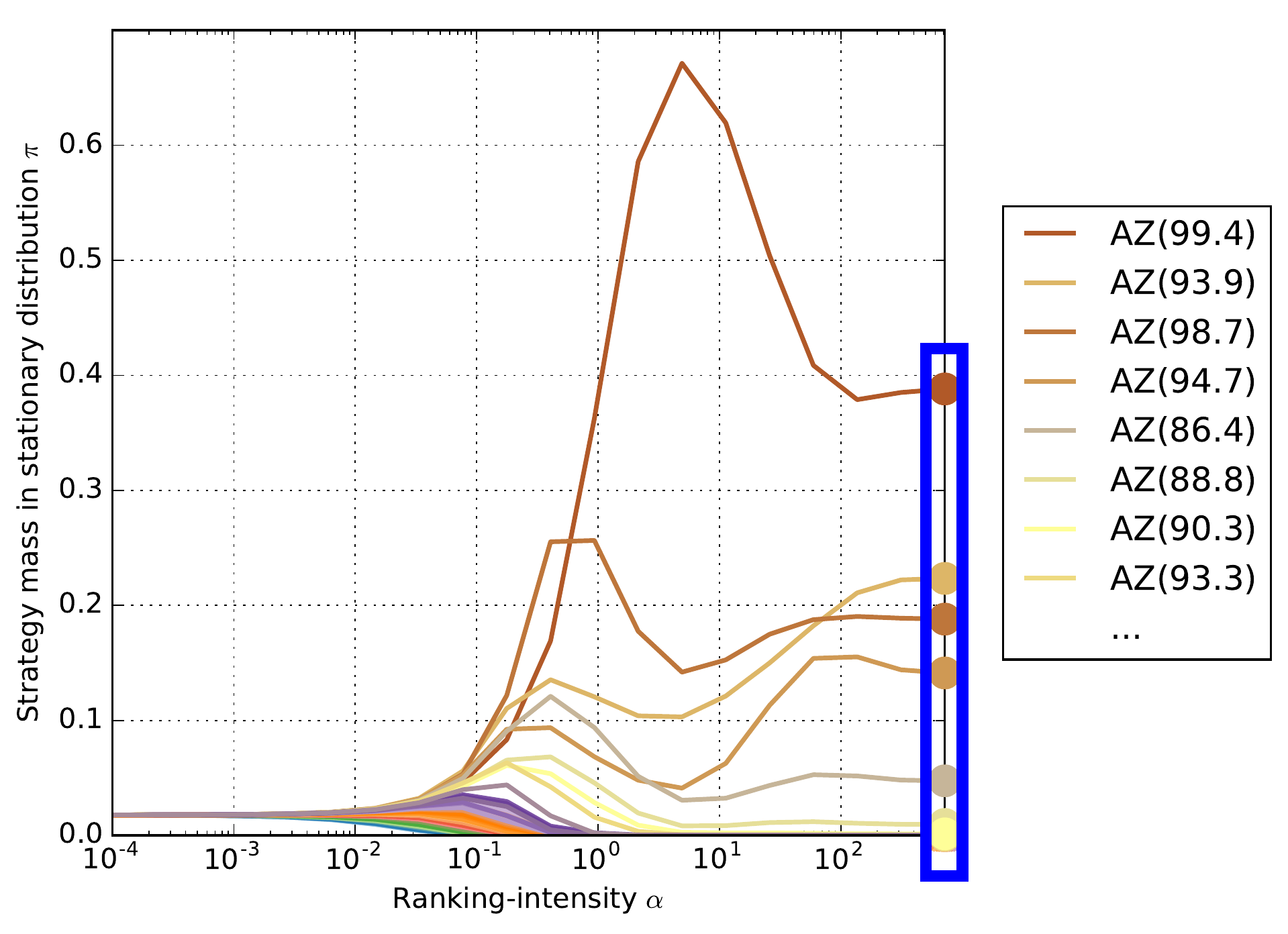}
        \caption{Ranking-intensity sweep.}
        \label{fig:pi_vs_alpha_chess}
    \end{subfigure}
    \hfill
    \begin{subtable}[b]{0.33\textwidth}
            \centering
            \def\arraystretch{1.2}
            \begin{tabular}{*{3}{c}}
            \toprule
            Agent & Rank & Score\\
            \midrule
            \rowcolor{MyBlue!39.0} \contour{white}{$AZ(99.4)$}& \contour{white}{$1$}& \contour{white}{$0.39$}\\
            \rowcolor{MyBlue!22.0} \contour{white}{$AZ(93.9)$}& \contour{white}{$2$}& \contour{white}{$0.22$}\\
            \rowcolor{MyBlue!19.0} \contour{white}{$AZ(98.7)$}& \contour{white}{$3$}& \contour{white}{$0.19$}\\
            \rowcolor{MyBlue!14.0} \contour{white}{$AZ(94.7)$}& \contour{white}{$4$}& \contour{white}{$0.14$}\\
            \rowcolor{MyBlue!5.0} \contour{white}{$AZ(86.4)$}& \contour{white}{$5$}& \contour{white}{$0.05$}\\
            \rowcolor{MyBlue!1.0} \contour{white}{$AZ(88.8)$}& \contour{white}{$6$}& \contour{white}{$0.01$}\\
            \rowcolor{MyBlue!0.0} \contour{white}{$AZ(90.3)$}& \contour{white}{$7$}& \contour{white}{$0.0$}\\
            \rowcolor{MyBlue!0.0} \contour{white}{$AZ(93.3)$}& \contour{white}{$8$}& \contour{white}{$0.0$}\\
            $\cdots$ & $\cdots$ & $\cdots$\\
            \bottomrule
            \end{tabular}
        \caption{$\alpha$-Rank results.}
        \label{table:alpharank_chess}
    \end{subtable}
    \caption{AlphaZero dataset.}
    \label{fig:results_chess}
\end{figure}

\begin{figure}[t]
    \centering
    \begin{subfigure}[t]{0.45\textwidth}
        \centering
        \includegraphics[width=1.\textwidth]{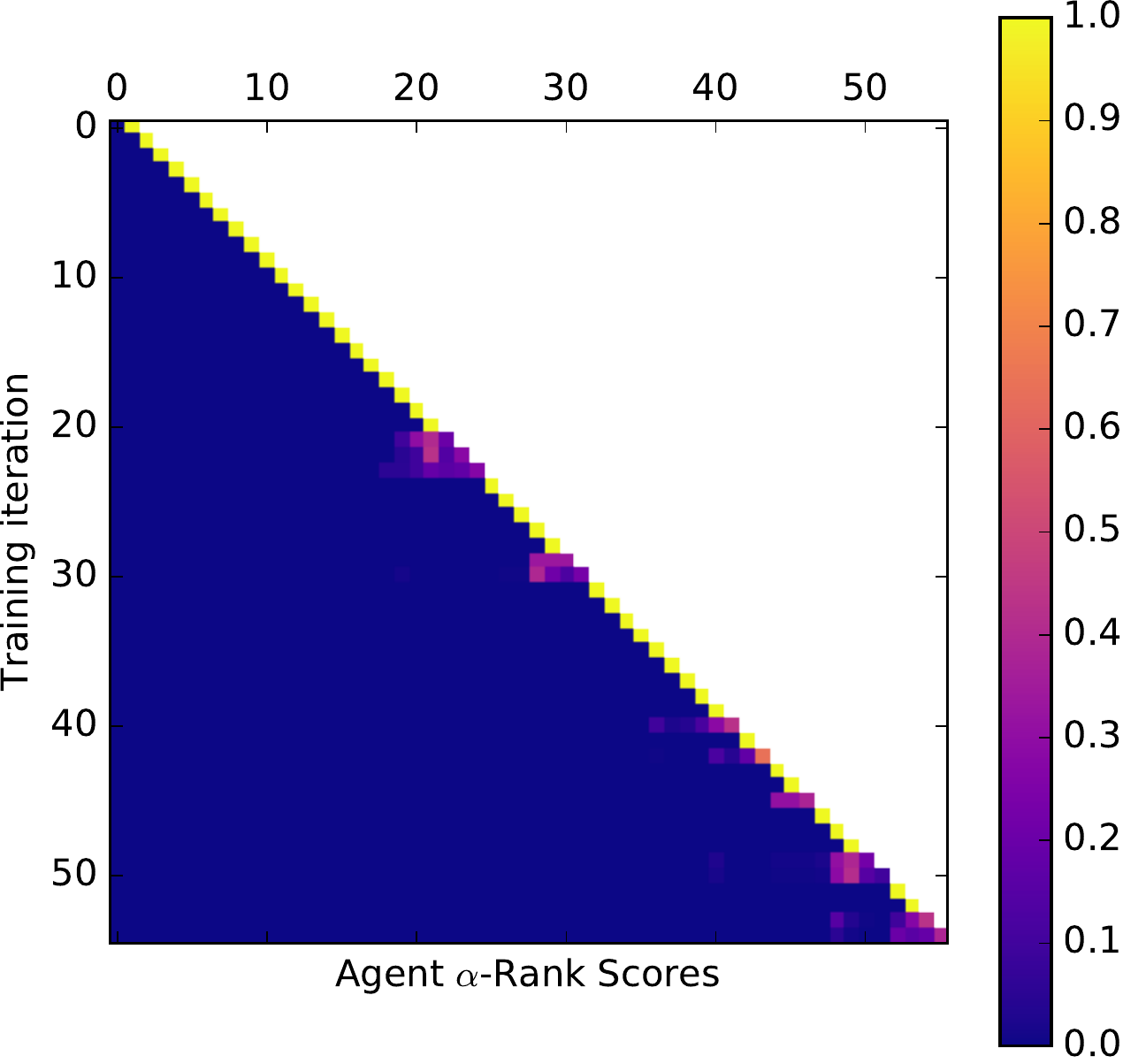}
        \caption{$\alpha$-Score vs. Training Time.}
        \label{fig:incremental_alpha_score_chess}
    \end{subfigure}
    \hfill
    \begin{subfigure}[t]{0.45\textwidth}
        \centering
        \includegraphics[width=1\textwidth]{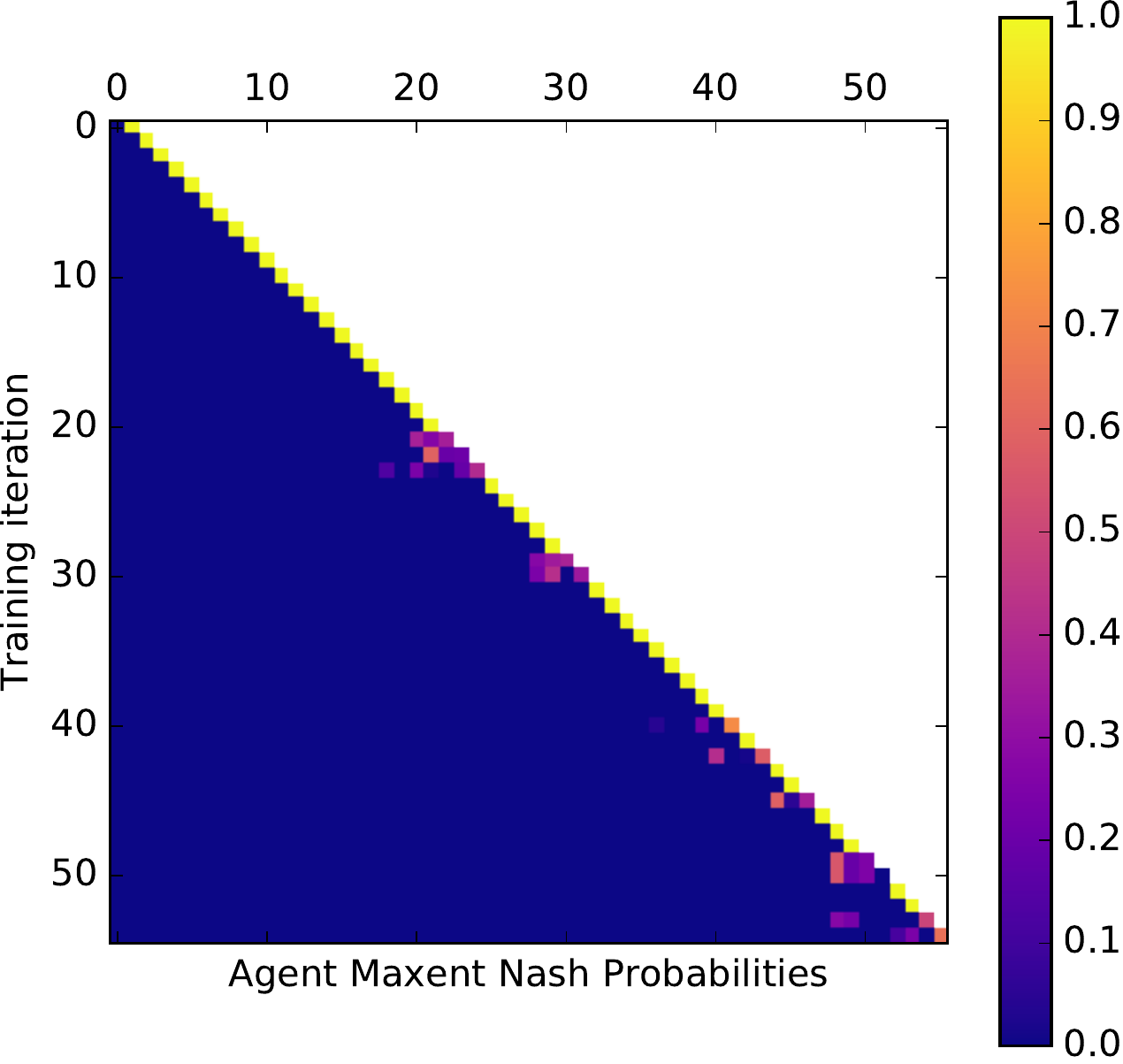}
        \caption{Maximum Entropy Nash vs. Training Time.}
        \label{fig:incremental_maxent_nash_chess}
    \end{subfigure} 
    \\
    \begin{subfigure}[t]{0.5\textwidth}
        \centering
        \includegraphics[width=1\textwidth]{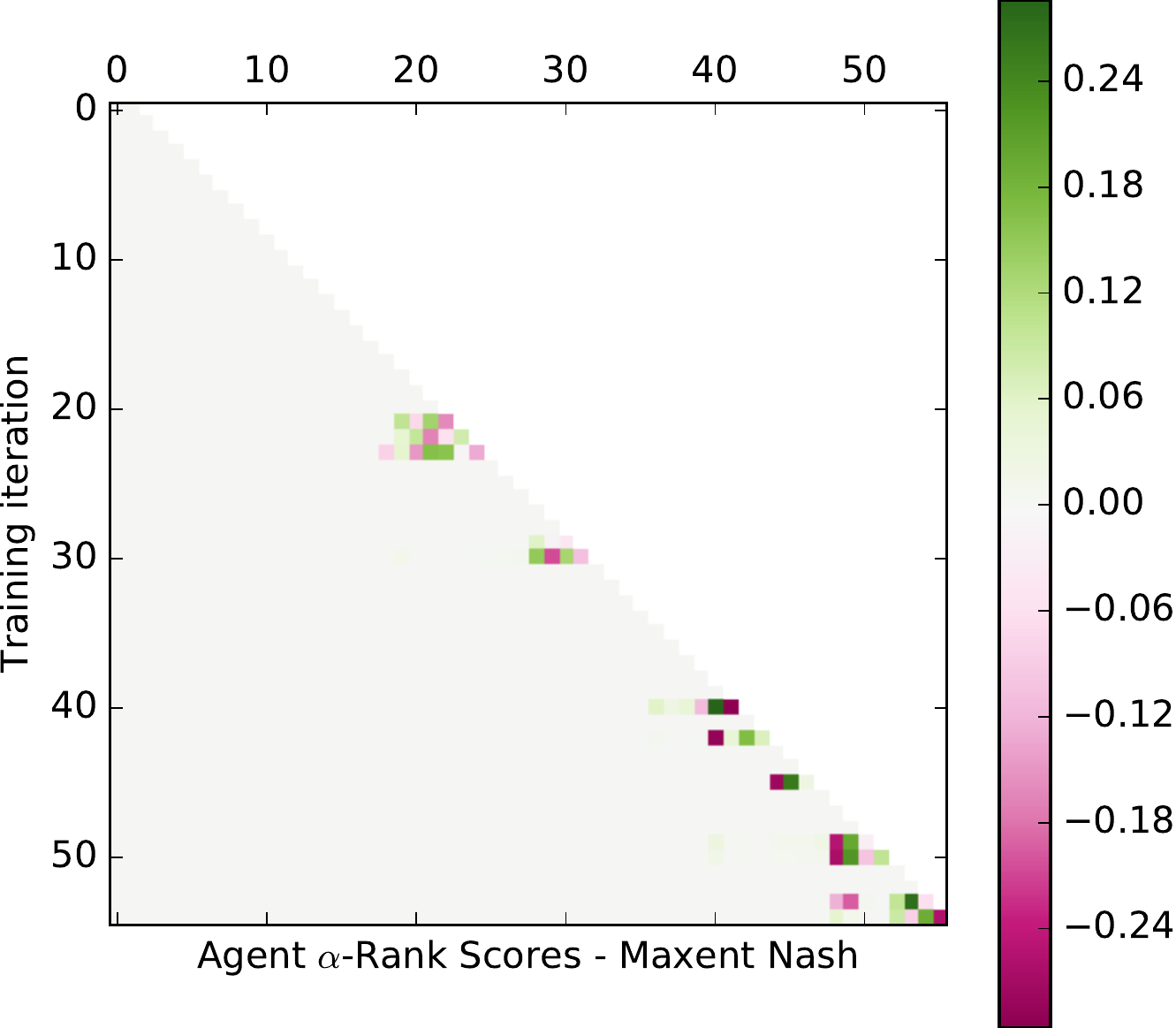}
        \caption{$\alpha$-Score - Maximum Entropy Nash difference.}
        \label{fig:incremental_alpha_score_vs_nash_chess}
    \end{subfigure}
    \caption{AlphaZero (chess) agent evaluations throughout training.}
    \label{fig:incremental_comparisons_chess}
\end{figure}

AlphaZero is a generalized algorithm that has been demonstrated to master the games of Go, Chess, and Shogi without reliance on human data \cite{silver2018general}.
Here we demonstrate the applicability of the $\alpha$-Rank evaluation method to large-scale domains by considering the interactions of a large number of AlphaZero agents playing the game of chess.
In AlphaZero, training commences by randomly initializing the parameters of a neural network used to play the game in conjunction with a general-purpose tree search algorithm.
To synthesize the corresponding meta-game, we take a `snapshot' of the network at various stages of training, each of which becomes an agent in our meta-game.
For example, agent $AZ(27.5)$ corresponds to a snapshot taken at approximately 27.5\% of the total number of training iterations, while $AZ(98.7)$ corresponds to one taken approximately at the conclusion of training.
We take 56 of these snapshots in total.
The meta-game considered here is then a symmetric $2$-player NFG involving 56 agents, with payoffs again corresponding to the win-rates of every pair of agent match-ups.
We note that there exist 27720 total simplex 2-faces in this dataset, substantially larger than those investigated in \cite{Tuyls18}, which quantifiably justifies the computational feasibility of our evaluation scheme.

We first analyze the evolutionary strengths of agents over a sweep of ranking-intensity $\alpha$ (\cref{fig:pi_vs_alpha_chess}).
While the overall rankings are quite invariant to the value of $\alpha$, we note again that a large value of $\alpha$ dictates the final $\alpha$-Rank evaluations attained in \cref{table:alpharank_chess}.
To gain further insight into the inter-agent interactions, we consider the corresponding discrete-time evolutionary dynamics shown in \cref{fig:mcc_chess}.
Note that these interactions are evaluated using the entire 56-agent dataset, though visualized only for the top-ranked agents for readability.
The majority of top-ranked agents indeed correspond to snapshots taken near the end of AlphaZero training (i.e., the strongest agents in terms of training time).
Specifically, $AZ(99.4)$, which is the final snapshot in our dataset and thus the most-trained agent, attains the top rank with a score of 0.39, in contrast to the second-ranked $AZ(93.9)$ agent's score of 0.22.
This analysis does reveal some interesting outcomes, however: agent $AZ(86.4)$ is not only ranked 5-th overall, but also higher than several agents with longer training time, including $AZ(88.8)$, $AZ(90.3)$, and $AZ(93.3)$.

We also investigate here the relationship between the $\alpha$-Rank scores and Nash equilibria.
A key point to recall is the equilibrium selection problem associated with Nash, as multiple equilibria can exist even in the case of two-player zero-sum meta-games.
In the case of zero-sum meta-games, Balduzzi \emph{et al.} show that there exists a unique \emph{maximum entropy} (maxent) Nash equilibrium \cite{Balduzzi18}, which constitutes a natural choice that we also use in the below comparisons.
For general games, unfortunately, this selection issue persists for Nash, whereas it does not for $\alpha$-Rank due to the uniqueness of the associated ranking (see \cref{property:unique_pi}).

We compare the $\alpha$-Rank scores and maxent Nash by plotting each throughout AlphaZero training in \cref{fig:incremental_alpha_score_chess} and \cref{fig:incremental_maxent_nash_chess}, respectively;
we also plot their difference in \cref{fig:incremental_alpha_score_vs_nash_chess}.
At a given training iteration, the corresponding horizontal slice in each plot visualizes the associated evaluation metric (i.e., $\alpha$-Rank, maxent Nash, or difference of the two) computed for all agent snapshots up to that iteration. 
We first note that both evaluation methods reach a consensus that the strengths of AlphaZero agents generally increase with training, in the sense that only the latest agent snapshots (i.e., the ones closest  to the diagonal) appear in the support of both $\alpha$-Rank scores and Nash.
An interesting observation is that less-trained agents sometimes reappear in the support of the distributions as training progresses; 
this behavior may even occur multiple times for a particular agent.

We consider also the quantitative similarity of $\alpha$-Rank and Nash in this domain.
\Cref{fig:incremental_alpha_score_vs_nash_chess} illustrates that differences do exist in the sense that certain agents are ranked higher via one method compared to the other.
More fundamentally, however, we note a relationship exists between $\alpha$-Rank and Nash in the sense that they share a common rooting in the concept of best-response:
by definition, each player's strategy in a Nash equilibrium is a best response to the other players' strategies;
in addition, $\alpha$-Rank corresponds to the MCC solution concept, which itself is derived from the sink strongly-connected components of the game's response graph. 
Despite the similarities, $\alpha$-Rank is a more refined solution concept than Nash in the sense that it is both rooted in dynamical systems and a best-response approach, which not only yields rankings, but also the associated dynamics graph (\cref{fig:mcc_chess}) that gives insights into the long-term evolutionary strengths of agents.
Beyond this, the critical advantage of $\alpha$-Rank is its tractability for general-sum games (per \cref{thm:stationary_complexity}), as well as lack of underlying equilibrium selection issues; in combination, these features yield a powerful empirical methodology with little room for user confusion or interpretability issues.
This analysis reveals fundamental insights not only in terms of the benefits of using $\alpha$-Rank to evaluate agents in a particular domain, but also an avenue of future work in terms of embedding the evaluation methodology into the training pipeline of agents involved in large and general games.

\subsubsection{MuJoCo Soccer}\label{sec:results_mujoco_soccer}
\begin{figure}[t]
    \centering
    \begin{subfigure}[t]{0.7\textwidth}
        \centering
        \includegraphics[width=1.\textwidth]{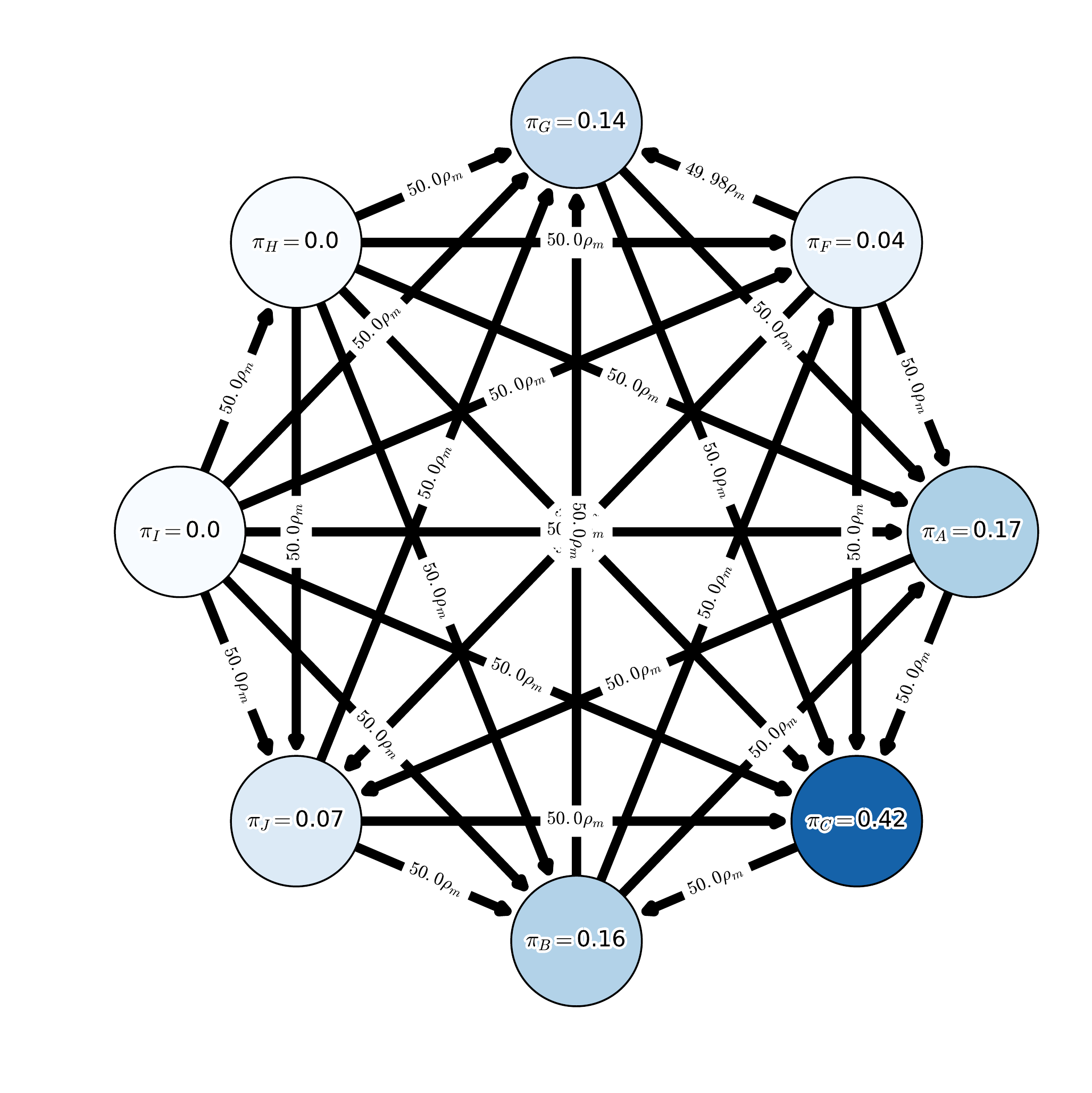}
        \caption{Discrete-time dynamics.}
        \label{fig:mcc_soccer}
    \end{subfigure}\\
    \begin{subfigure}[t]{0.6\textwidth}
        \centering
        \includegraphics[width=0.8\textwidth]{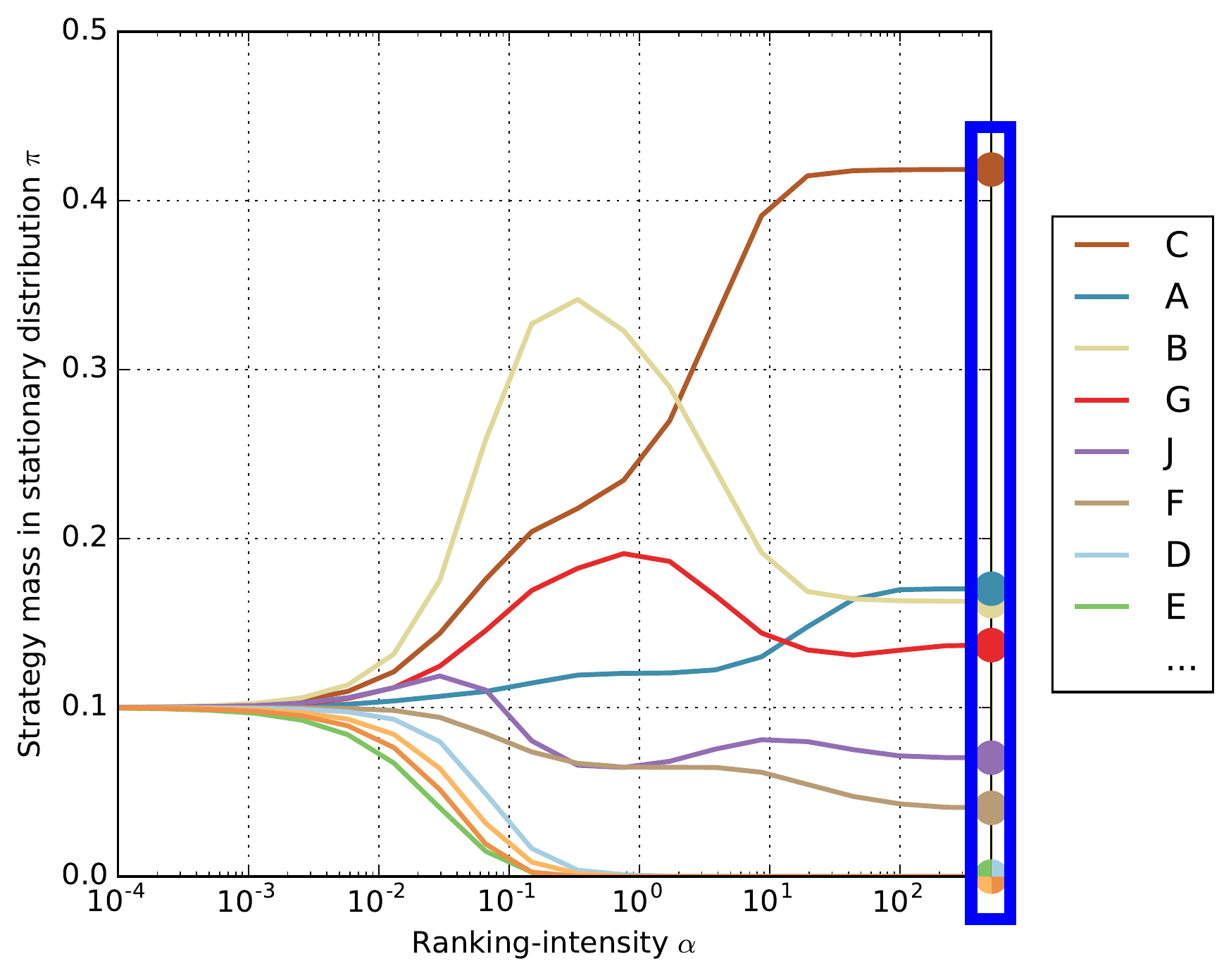}
        \caption{Ranking-intensity sweep.}
        \label{fig:pi_vs_alpha_soccer}
    \end{subfigure}
    \hfill
    \begin{subtable}[b]{0.33\textwidth}
            \centering
            \def\arraystretch{1.2}
            \begin{tabular}{ *{3}{c} }
            \toprule
            Agent & Rank & Score\\
            \midrule
            \rowcolor{MyBlue!42.0} \contour{white}{$C$}& \contour{white}{$1$}& \contour{white}{$0.42$}\\
            \rowcolor{MyBlue!17.0} \contour{white}{$A$}& \contour{white}{$2$}& \contour{white}{$0.17$}\\
            \rowcolor{MyBlue!16.0} \contour{white}{$B$}& \contour{white}{$3$}& \contour{white}{$0.16$}\\
            \rowcolor{MyBlue!14.0} \contour{white}{$G$}& \contour{white}{$4$}& \contour{white}{$0.14$}\\
            \rowcolor{MyBlue!7.0} \contour{white}{$J$}& \contour{white}{$5$}& \contour{white}{$0.07$}\\
            \rowcolor{MyBlue!4.0} \contour{white}{$F$}& \contour{white}{$6$}& \contour{white}{$0.04$}\\
            \rowcolor{MyBlue!0.0} \contour{white}{$D$}& \contour{white}{$7$}& \contour{white}{$0.0$}\\
            \rowcolor{MyBlue!0.0} \contour{white}{$E$}& \contour{white}{$7$}& \contour{white}{$0.0$}\\
            $\cdots$ & $\cdots$ & $\cdots$\\
            \bottomrule
            \end{tabular}
        \caption{$\alpha$-Rank results.}
        \label{table:alpharank_soccer}
    \end{subtable}
    \caption{MuJoCo soccer dataset.}
    \label{fig:results_soccer}
\end{figure}

\begin{figure}[t]
    \centering
    \begin{subfigure}[t]{0.5\textwidth}
        \centering
        \includegraphics[width=1.\textwidth,trim={4cm 3cm 3cm 3cm},clip]{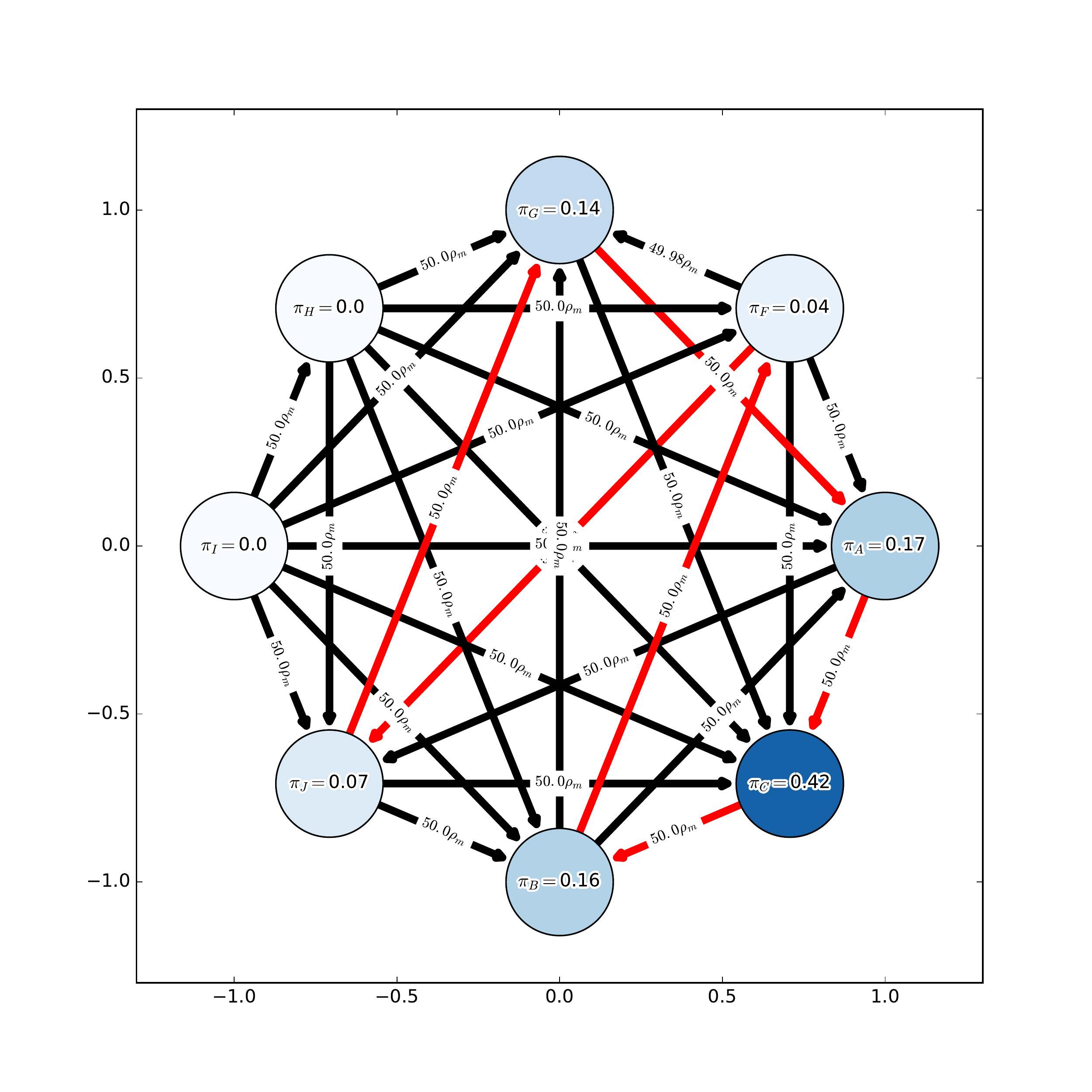}
        \caption{}
        \label{fig:mcc_cycle1}
    \end{subfigure}%
    \begin{subfigure}[t]{0.5\textwidth}
        \centering
        \includegraphics[width=1.\textwidth,trim={4cm 3cm 3cm 3cm},clip]{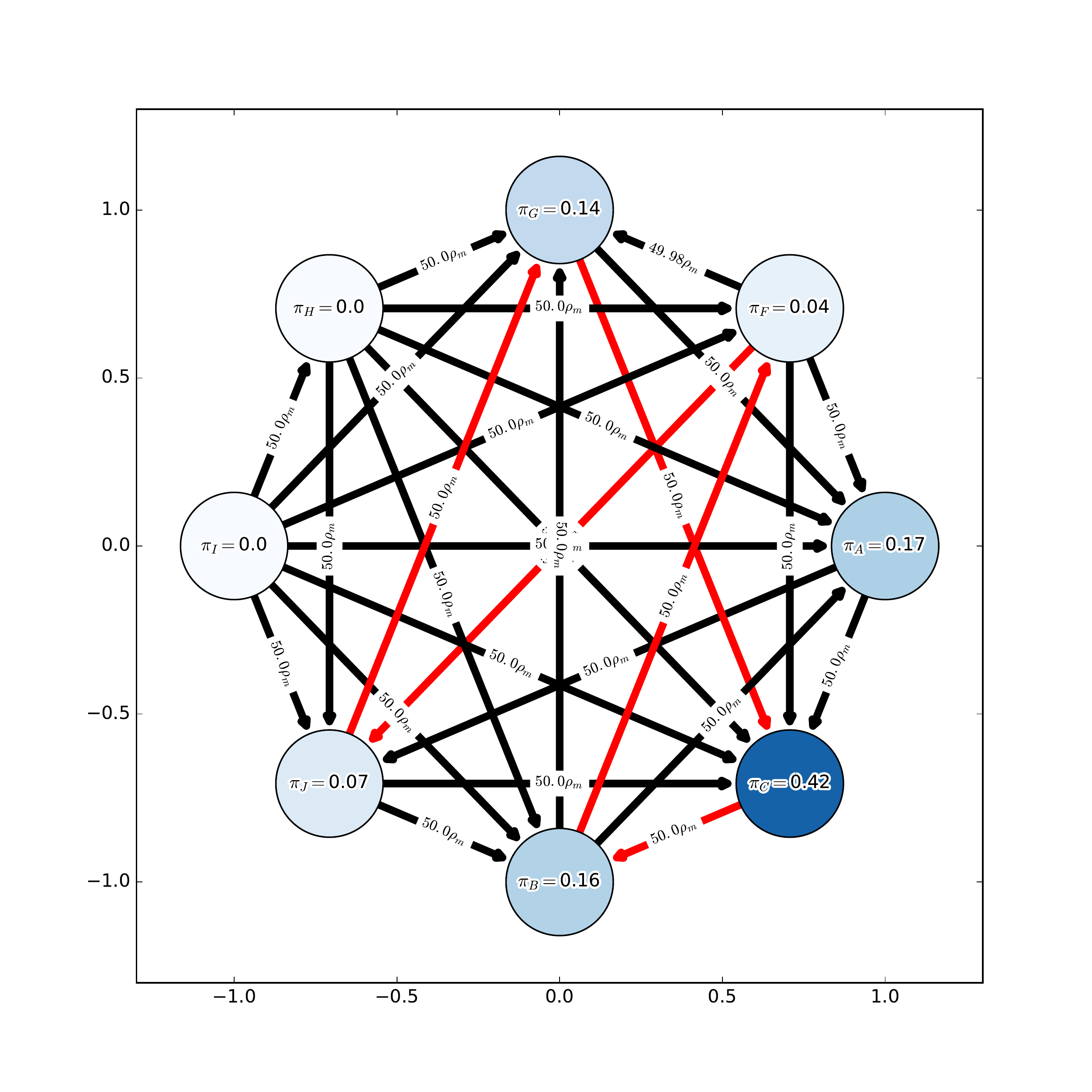}
        \caption{}
        \label{fig:mcc_cycle2}
    \end{subfigure}\\
    \begin{subfigure}[t]{0.5\textwidth}
        \centering
        \includegraphics[width=1.\textwidth,trim={4cm 3cm 3cm 3cm},clip]{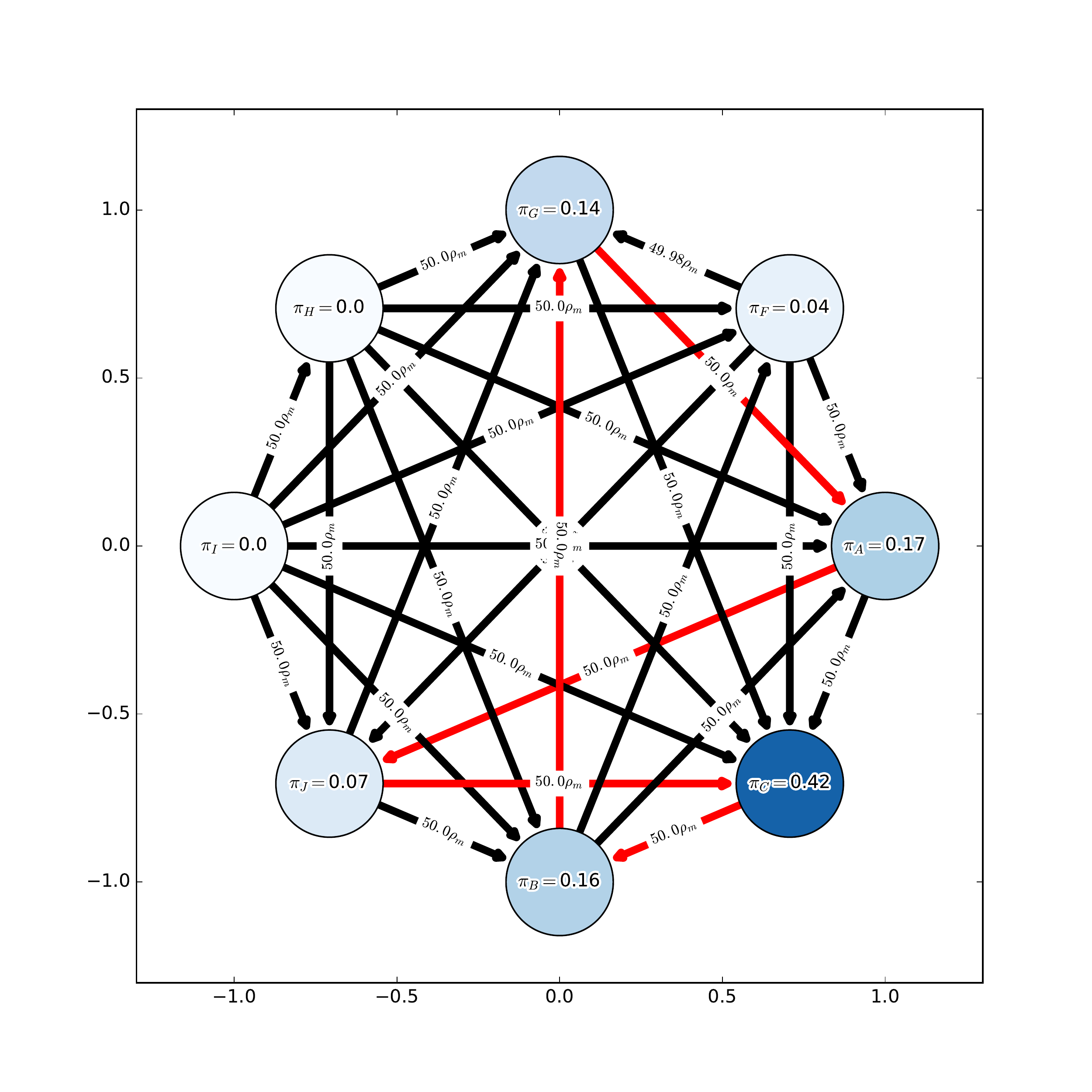}
        \caption{}
        \label{fig:mcc_cycle3}
    \end{subfigure}%
        \begin{subfigure}[t]{0.5\textwidth}
        \centering
        \includegraphics[width=1.\textwidth,trim={4cm 3cm 3cm 3cm},clip]{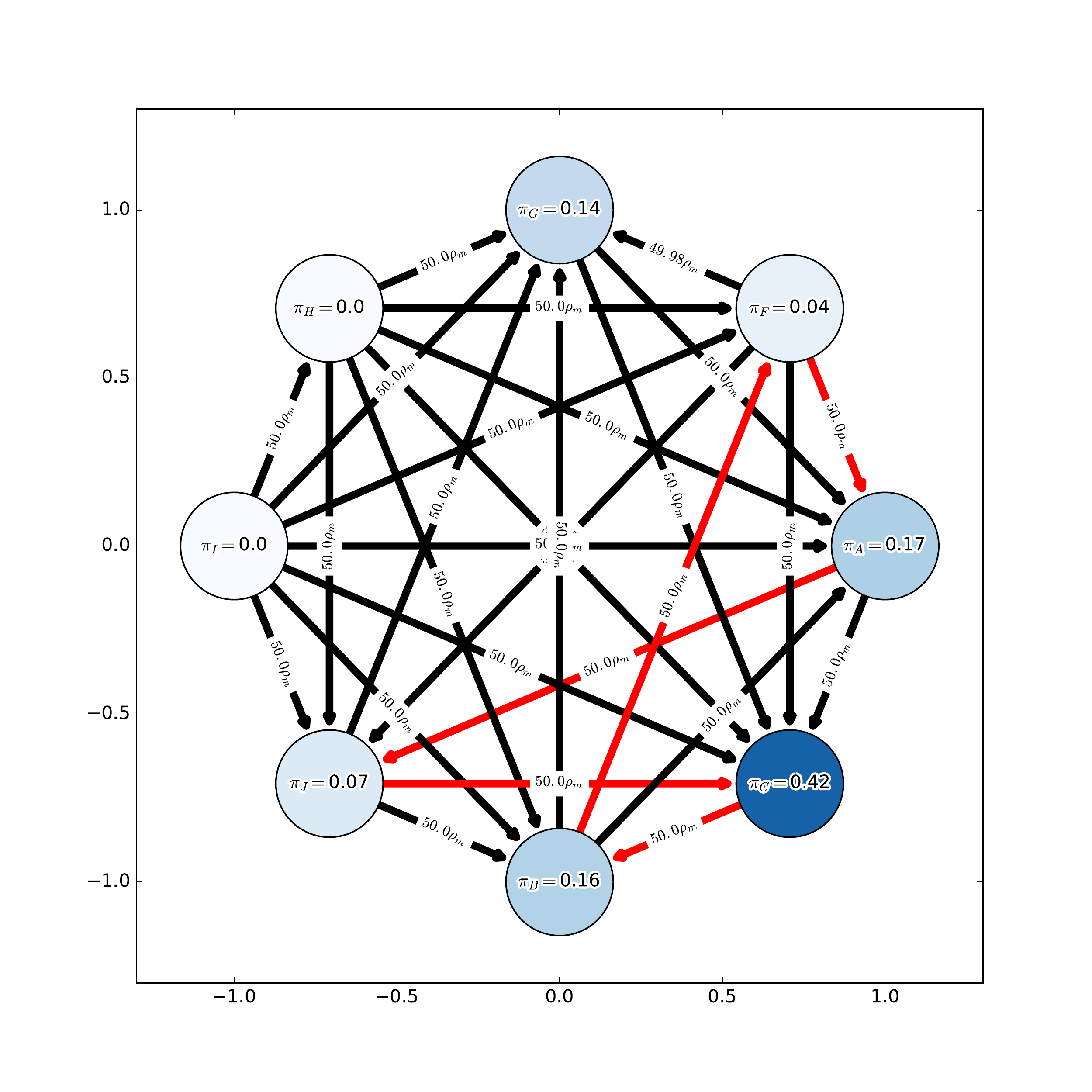}
        \caption{}
        \label{fig:mcc_cycle4}
    \end{subfigure}%
    \caption{Example cycles in the MuJoCo soccer domain.}
    \label{fig:mcc_soccer_cycles}
\end{figure}

We consider here a dataset consisting of complex agent interactions in the continuous-action domain of MuJoCo soccer \cite{liu2018emergent}.
Specifically, this domain involves a multi-agent soccer physics-simulator environment with teams of 2 vs. 2 agents in the MuJoCo physics engine \cite{todorov:12}.
Each agent, specifically, uses a distinct variation of algorithmic and/or policy parameterization component (see \cite{liu2018emergent} for agent specifications).
The underlying meta-game is a $2$-player NFG consisting of 10 agents, with payoffs corresponding to Figure 2 of \cite{liu2018emergent}.

We consider again the variation of the stationary distribution as a function of ranking-intensity $\alpha$ (\cref{fig:pi_vs_alpha_soccer}).
Under the large $\alpha$ limit, only 6 agent survive, with the remaining 4 agents considered transient in the long-term.
Moreoever, the top 3 $\alpha$-Ranked agents ($C$, $A$, and $B$, as shown in \cref{table:alpharank_soccer}) correspond to the strongest agents highlighted in \cite{liu2018emergent}, though $\alpha$-Rank highlights 3 additional agents ($G$, $J$, and $F$) that are not in the top-rank set outlined in their work.
An additional key benefit of our approach is that it can immediately highlight the presence of intransitive behaviors (cycles) in general games.
Worthy of remark in this dataset is the presence of a large number of cycles, several of which are identified in \cref{fig:mcc_soccer_cycles}.
Not only can we identify these cycles visually, these intransitive behaviors are automatically taken into account in our rankings due to the fundamental role that recurrence plays in our underlying solution concept.
This is in contrast to the Elo rating (which is incapable of dealing with intransitivities), the replicator dynamics (which are limited in terms of visualizing such intransitive behaviors for large games), and Nash (which is a static solution concept that does not capture dynamic behavior).

\subsubsection{Kuhn Poker}\label{sec:results_kuhn_poker}
\begin{figure}[t]
    \centering
    \begin{subfigure}[t]{0.7\textwidth}
        \centering
        \includegraphics[width=1.\textwidth]{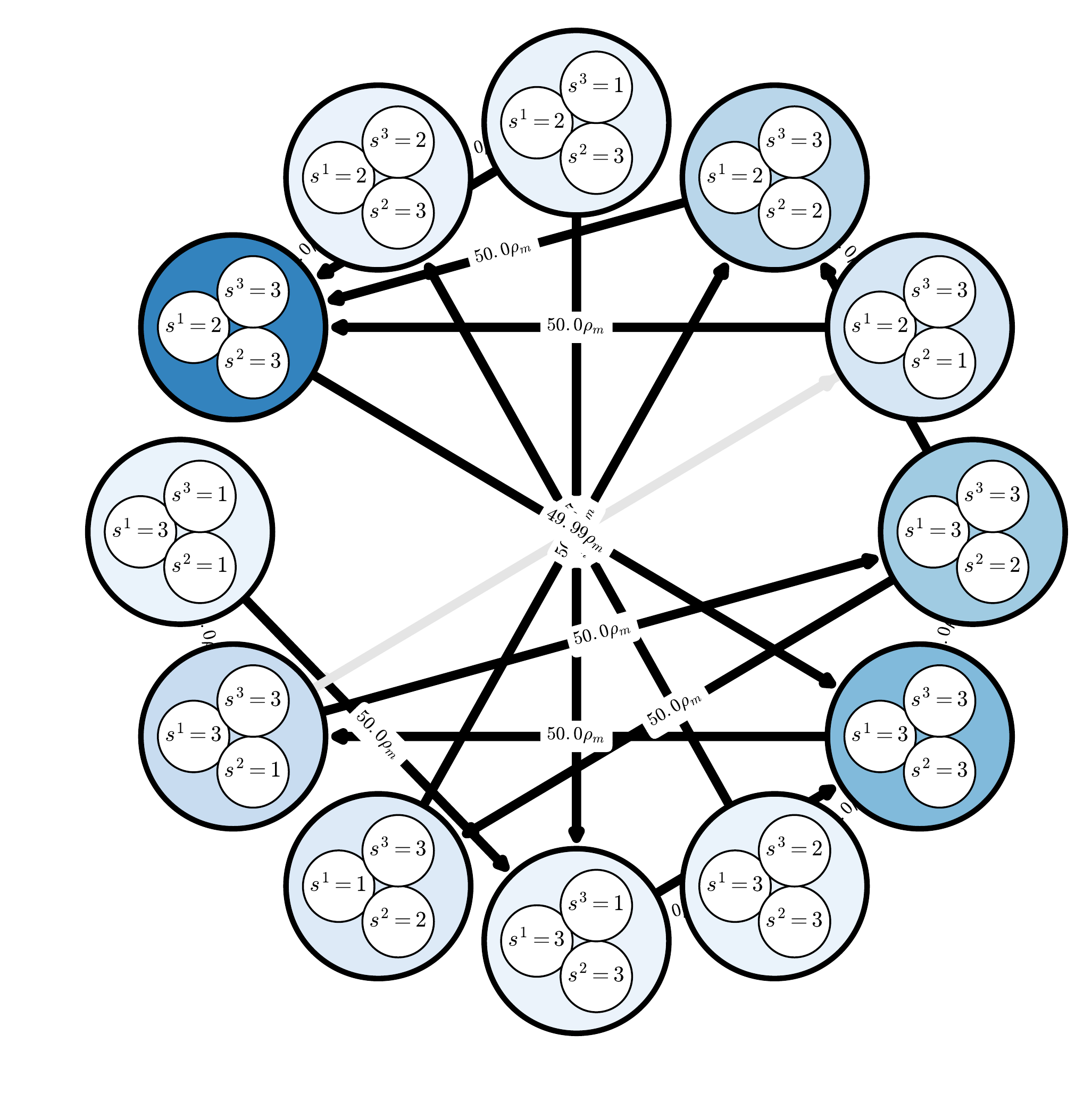}
        \caption{Discrete-time dynamics.}
        \label{fig:mcc_kuhn_poker_3p}
    \end{subfigure}\\
    \begin{subfigure}[t]{0.6\textwidth}
        \centering
        \includegraphics[width=0.8\textwidth]{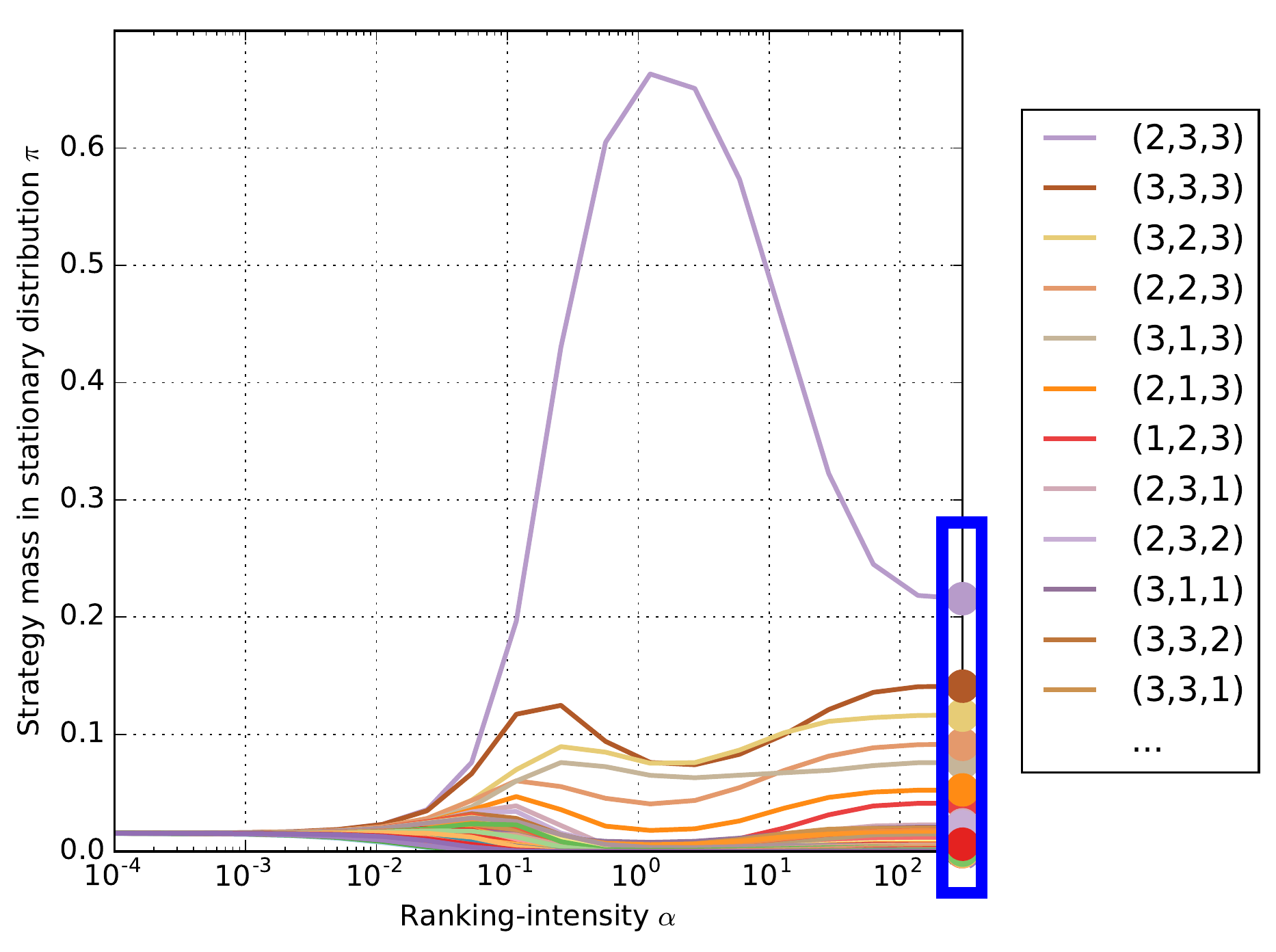}
        \caption{Ranking-intensity sweep.}
        \label{fig:pi_vs_alpha_kuhn_poker_3p}
    \end{subfigure}
    \hfill
    \begin{subtable}[b]{0.33\textwidth}
            \centering
            \def\arraystretch{1.2}
            \begin{tabular}{ *{3}{c} }
            \toprule
            Agent & Rank & Score\\
            \midrule
            \rowcolor{MyBlue!22.0} \contour{white}{$(2,3,3)$}& \contour{white}{$1$}& \contour{white}{$0.22$}\\
            \rowcolor{MyBlue!14.0} \contour{white}{$(3,3,3)$}& \contour{white}{$2$}& \contour{white}{$0.14$}\\
            \rowcolor{MyBlue!12.0} \contour{white}{$(3,2,3)$}& \contour{white}{$3$}& \contour{white}{$0.12$}\\
            \rowcolor{MyBlue!9.0} \contour{white}{$(2,2,3)$}& \contour{white}{$4$}& \contour{white}{$0.09$}\\
            \rowcolor{MyBlue!8.0} \contour{white}{$(3,1,3)$}& \contour{white}{$5$}& \contour{white}{$0.08$}\\
            \rowcolor{MyBlue!5.0} \contour{white}{$(2,1,3)$}& \contour{white}{$6$}& \contour{white}{$0.05$}\\
            \rowcolor{MyBlue!4.0} \contour{white}{$(1,2,3)$}& \contour{white}{$7$}& \contour{white}{$0.04$}\\
            \rowcolor{MyBlue!2.0} \contour{white}{$(2,3,1)$}& \contour{white}{$8$}& \contour{white}{$0.02$}\\
            \rowcolor{MyBlue!2.0} \contour{white}{$(2,3,2)$}& \contour{white}{$9$}& \contour{white}{$0.02$}\\
            \rowcolor{MyBlue!2.0} \contour{white}{$(3,1,1)$}& \contour{white}{$10$}& \contour{white}{$0.02$}\\
            \rowcolor{MyBlue!2.0} \contour{white}{$(3,3,2)$}& \contour{white}{$11$}& \contour{white}{$0.02$}\\
            \rowcolor{MyBlue!2.0} \contour{white}{$(3,3,1)$}& \contour{white}{$12$}& \contour{white}{$0.02$}\\
            $\cdots$ & $\cdots$ & $\cdots$\\
            \bottomrule
            \end{tabular}
        \caption{$\alpha$-Rank results.}
        \label{table:alpharank_kuhn_poker_3p}
    \end{subtable}
    \caption{3-player Kuhn poker (ranking conducted on all 64 pure strategy profiles).}
    \label{fig:results_kuhn_poker_3p}
\end{figure}

\begin{figure}[t]
    \centering
    \begin{subfigure}[t]{0.7\textwidth}
        \centering
        \includegraphics[width=1.\textwidth]{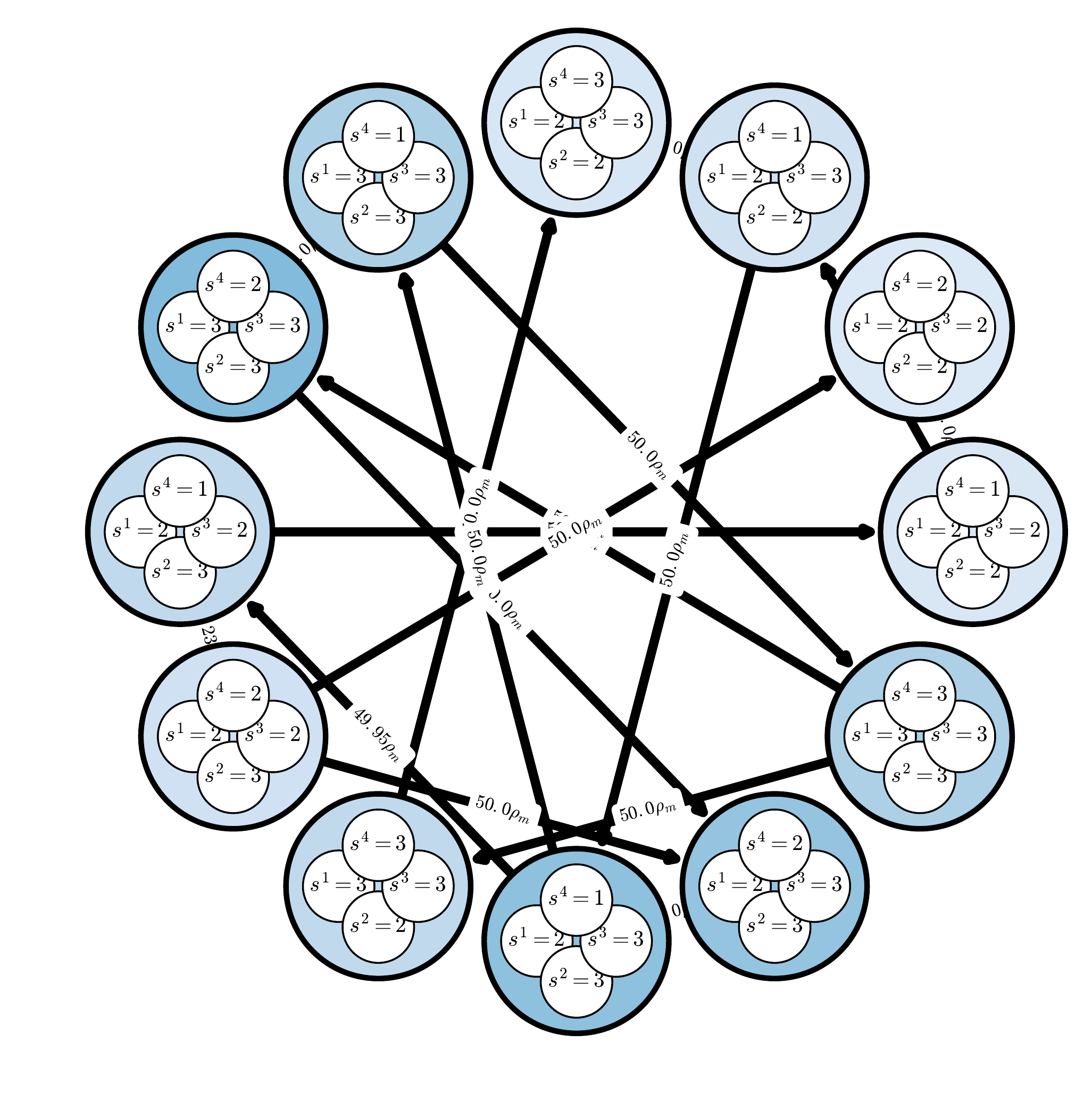}
        \caption{Discrete-time dynamics.}
        \label{fig:mcc_kuhn_poker_4p}
    \end{subfigure}\\
    \begin{subfigure}[t]{0.6\textwidth}
        \centering
        \includegraphics[width=0.8\textwidth]{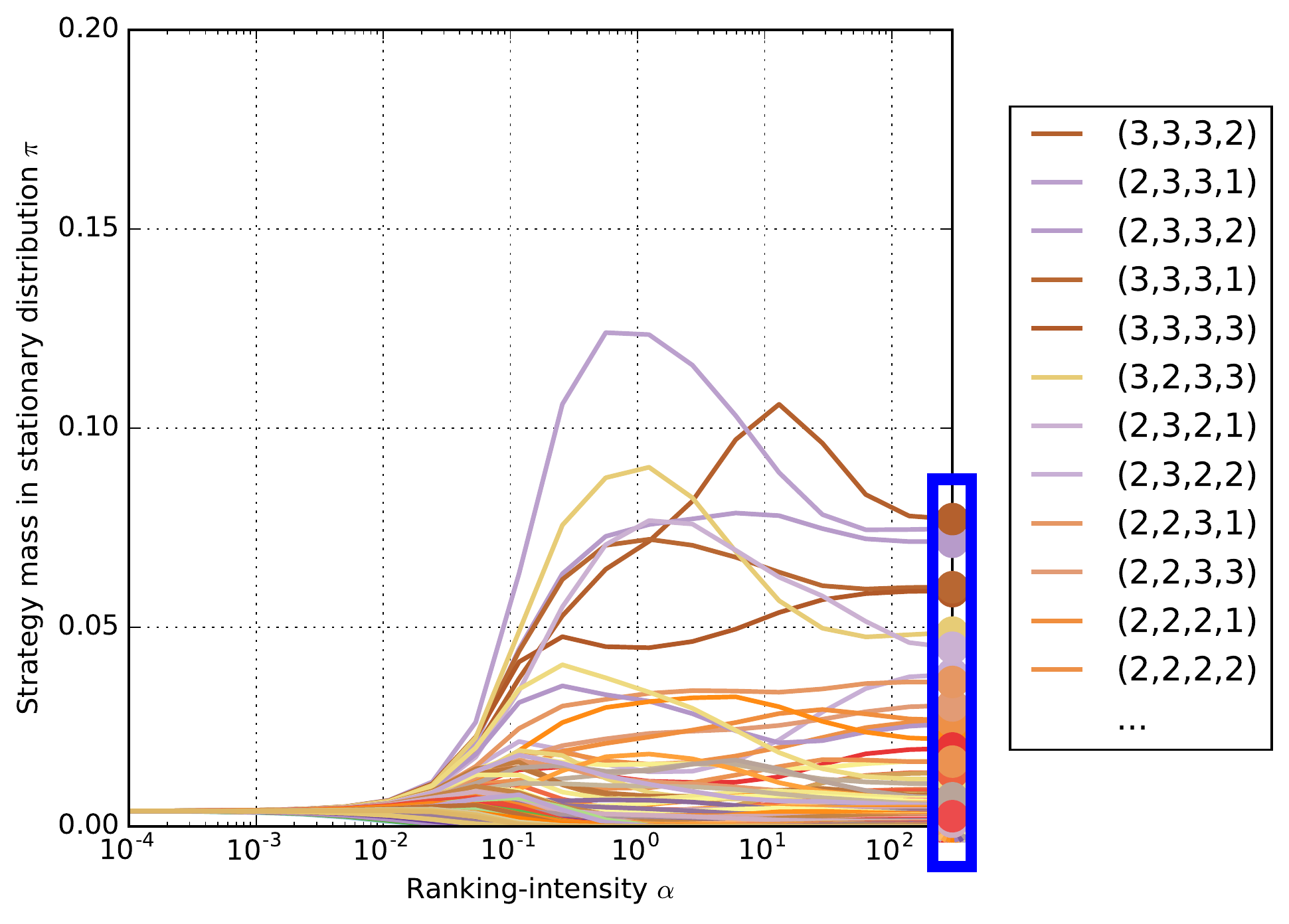}
        \caption{Ranking-intensity sweep.}
        \label{fig:pi_vs_alpha_kuhn_poker_4p}
    \end{subfigure}
    \hfill
    \begin{subtable}[b]{0.33\textwidth}
            \centering
            \def\arraystretch{1.2}
            \begin{tabular}{ *{3}{c} }
            \toprule
            Agent & Rank & Score\\
            \midrule
            \rowcolor{MyBlue!8.0} \contour{white}{$(3,3,3,2)$}& \contour{white}{$1$}& \contour{white}{$0.08$}\\
            \rowcolor{MyBlue!7.0} \contour{white}{$(2,3,3,1)$}& \contour{white}{$2$}& \contour{white}{$0.07$}\\
            \rowcolor{MyBlue!7.0} \contour{white}{$(2,3,3,2)$}& \contour{white}{$3$}& \contour{white}{$0.07$}\\
            \rowcolor{MyBlue!6.0} \contour{white}{$(3,3,3,1)$}& \contour{white}{$4$}& \contour{white}{$0.06$}\\
            \rowcolor{MyBlue!6.0} \contour{white}{$(3,3,3,3)$}& \contour{white}{$5$}& \contour{white}{$0.06$}\\
            \rowcolor{MyBlue!5.0} \contour{white}{$(3,2,3,3)$}& \contour{white}{$6$}& \contour{white}{$0.05$}\\
            \rowcolor{MyBlue!4.0} \contour{white}{$(2,3,2,1)$}& \contour{white}{$7$}& \contour{white}{$0.04$}\\
            \rowcolor{MyBlue!4.0} \contour{white}{$(2,3,2,2)$}& \contour{white}{$8$}& \contour{white}{$0.04$}\\
            \rowcolor{MyBlue!4.0} \contour{white}{$(2,2,3,1)$}& \contour{white}{$9$}& \contour{white}{$0.04$}\\
            \rowcolor{MyBlue!3.0} \contour{white}{$(2,2,3,3)$}& \contour{white}{$10$}& \contour{white}{$0.03$}\\
            \rowcolor{MyBlue!3.0} \contour{white}{$(2,2,2,1)$}& \contour{white}{$11$}& \contour{white}{$0.03$}\\
            \rowcolor{MyBlue!3.0} \contour{white}{$(2,2,2,2)$}& \contour{white}{$12$}& \contour{white}{$0.03$}\\
            $\cdots$ & $\cdots$ & $\cdots$\\
            \bottomrule
            \end{tabular}
        \caption{$\alpha$-Rank results.}
        \label{table:alpharank_kuhn_poker_4p}
    \end{subtable}
    \caption{4-player Kuhn poker (ranking conducted on all 256 pure strategy profiles).}
    \label{fig:results_kuhn_poker_4p}
\end{figure}

We next consider games wherein the inherent complexity is due to the number of players involved.
Specifically, we consider Kuhn poker with $3$ and $4$ players, extending beyond the reach of prior meta-game evaluation approaches that are limited to pairwise asymmetric interactions \cite{Tuyls18}.
Kuhn poker is a small poker game where each player starts with 2 chips, antes 1 chip to play, and receives one card face down from a deck of size $n$ + 1 (one card remains hidden). Players proceed by betting (raise/call) by adding their remaining chip to the pot, or passing (check/fold) until all players are either in (contributed as all other players to the pot) or out (folded, passed after a raise).
The player with the highest-ranked card that has not folded wins the pot.
The two-player game is known to have a continuum of strategies, which could have fairly high support, that depends on a single parameter: the probability that the first player raises with the highest card~\cite{Southey08}.
The three-player game has a significantly more complex landscape~\cite{Szafron13}.
The specific rules used for the three and four player variants can be found in~\cite[Section 4.1]{Lanctot14Further}.

Here, our meta-game dataset consists of several (fixed) rounds of extensive-form fictitious play (specifically, XFP from~\cite{Heinrich15FSP}): 
in round 0, the payoff corresponding to strategy profile $(0,0,0)$ in each meta-game of $3$-player Kuhn corresponds to the estimated payoff of each player using uniform random strategies;
in fictitious play round 1, the payoff entry $(1,1,1)$ corresponds to each player playing an approximate best response to the other players' uniform strategies;
in fictitious play round 2, entry $(2,2,2)$ corresponds to each playing  an approximate best response to the other players' uniform mixtures over their round 0 strategies (uniform random) and round 1 oracle strategy (best response to random); 
and so on. 
Note, especially, that oracles at round 0 are likely to be dominated (as they are uniform random).
In our dataset, we consider two asymmetric meta-games, each involving 3 rounds of fictitious play with 3 and 4 players (\cref{fig:results_kuhn_poker_3p} and \cref{fig:results_kuhn_poker_4p}, respectively). 

Of particular note are the total number of strategy profiles involved in these meta-games, 64 and 256 respectively for the 3 and 4 player games -- the highest considered in any of our datasets.
Conducting the evaluation using the replicator-dynamics based analysis of \cite{Tuyls18} can be quite tedious as all possible 2-face simplices must be considered for \emph{each} player.
Instead, here the $\alpha$-Rankings follow the same methodology used for all other domains, and are summarized succinctly in \cref{table:alpharank_kuhn_poker_3p,table:alpharank_kuhn_poker_4p}.
In both meta-games, the 3-round fictitious play strategies ($(3,3,3)$ and $(3,3,3,3)$, respectively) are ranked amongst the top-5 strategies.

\subsubsection{Leduc Poker}\label{sec:results_psro}
\begin{figure}[t]
    \centering
    \begin{subfigure}[t]{0.7\textwidth}
        \centering
        \includegraphics[width=1.\textwidth]{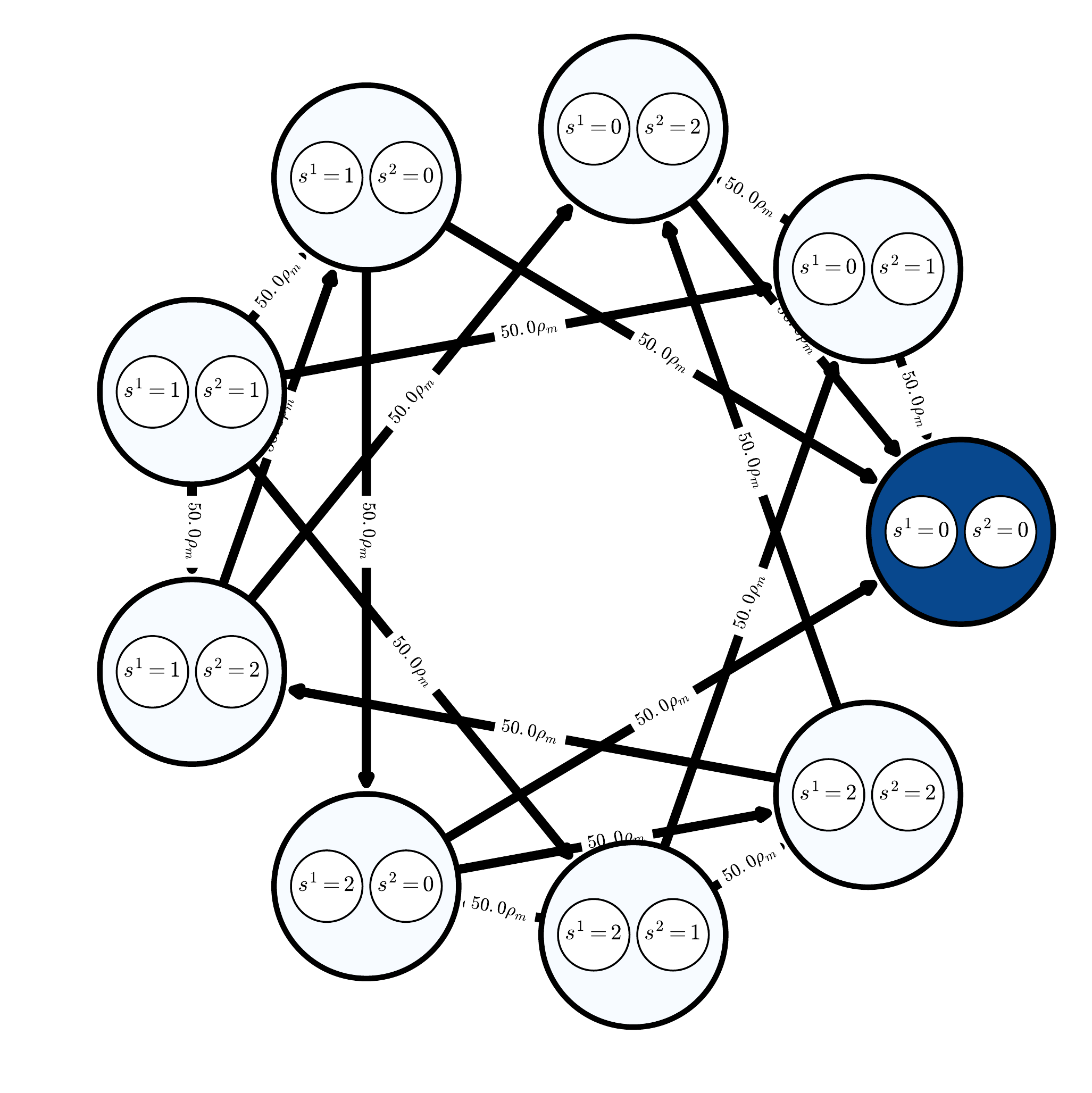}
        \caption{Discrete-time dynamics (top 8 agents shown only).}
        \label{fig:mcc_psro}
    \end{subfigure}\\
    \begin{subfigure}[t]{0.6\textwidth}
        \centering
        \includegraphics[width=0.8\textwidth]{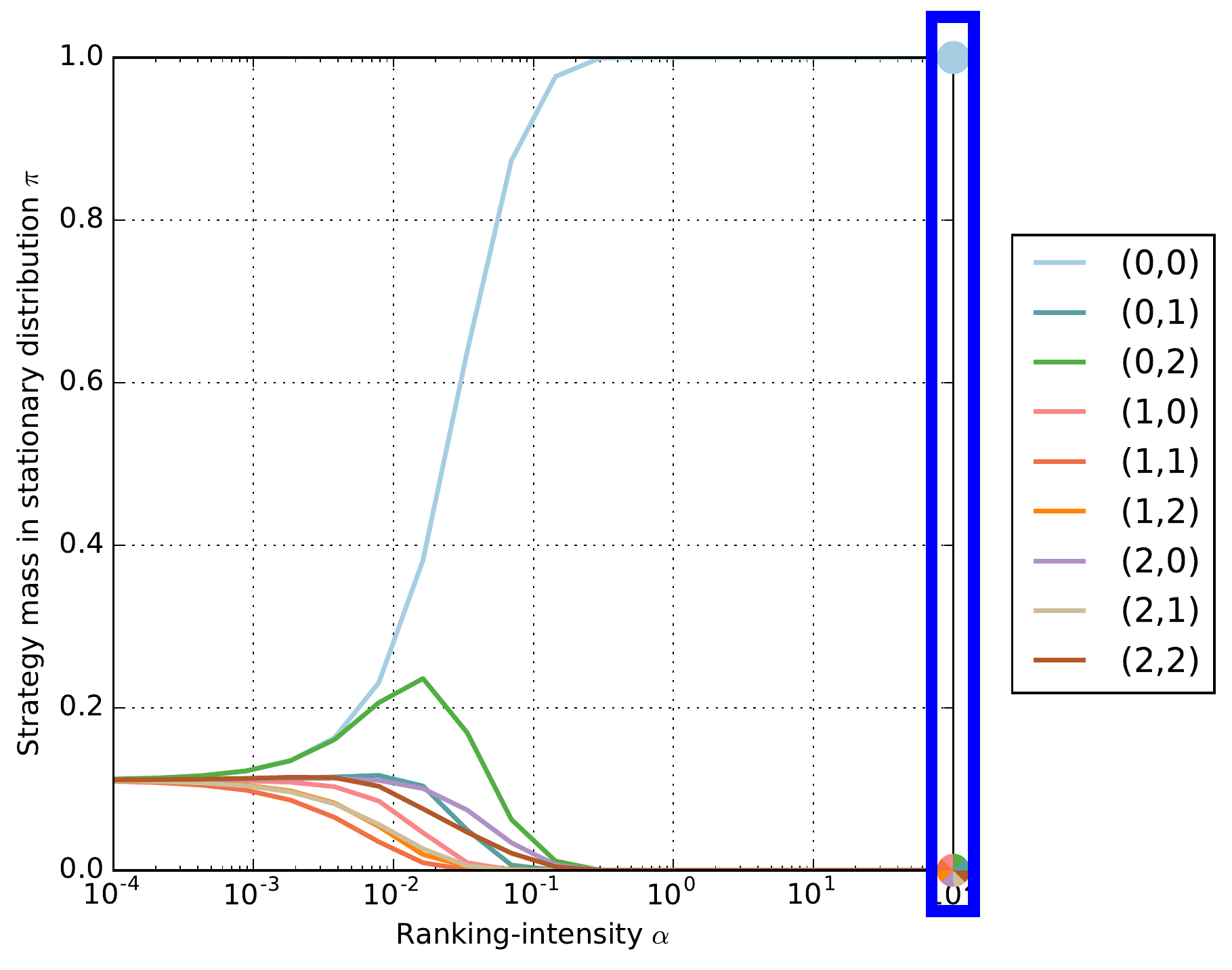}
        \caption{Ranking-intensity sweep.}
        \label{fig:pi_vs_alpha_psro}
    \end{subfigure}
    \hfill
    \begin{subtable}[b]{0.33\textwidth}
            \centering
            \def\arraystretch{1.2}
            \begin{tabular}{ *{3}{c} }
            \toprule
            Agent & Rank & Score\\
            \midrule
            \rowcolor{MyBlue!100.0} \contour{white}{$(0,0)$}& \contour{white}{$1$}& \contour{white}{$1.0$}\\
            \rowcolor{MyBlue!0.0} \contour{white}{$(0,1)$}& \contour{white}{$2$}& \contour{white}{$0.0$}\\
            \rowcolor{MyBlue!0.0} \contour{white}{$(0,2)$}& \contour{white}{$2$}& \contour{white}{$0.0$}\\
            \rowcolor{MyBlue!0.0} \contour{white}{$(1,0)$}& \contour{white}{$2$}& \contour{white}{$0.0$}\\
            \rowcolor{MyBlue!0.0} \contour{white}{$(1,1)$}& \contour{white}{$2$}& \contour{white}{$0.0$}\\
            \rowcolor{MyBlue!0.0} \contour{white}{$(1,2)$}& \contour{white}{$2$}& \contour{white}{$0.0$}\\
            \rowcolor{MyBlue!0.0} \contour{white}{$(2,0)$}& \contour{white}{$2$}& \contour{white}{$0.0$}\\
            \rowcolor{MyBlue!0.0} \contour{white}{$(2,1)$}& \contour{white}{$2$}& \contour{white}{$0.0$}\\
            \rowcolor{MyBlue!0.0} \contour{white}{$(2,2)$}& \contour{white}{$2$}& \contour{white}{$0.0$}\\
            \bottomrule
            \end{tabular}
        \caption{$\alpha$-Rank strategy rankings and scores (top 8 agents shown only).}
        \label{table:alpharank_psro}
    \end{subtable}
    \caption{PSRO poker dataset.}
    \label{fig:results_psro}
\end{figure}

The meta-game we consider next involves agents generated using the Policy Space Response Oracles (PSRO) algorithm \cite{Lanctot17}.
Specifically, PSRO can be viewed as a generalization of fictitious play, which computes approximate responses (``oracles'') using deep reinforcement learning, along with arbitrary meta-strategy solvers; here, PSRO is applied to the game of Leduc poker.
Leduc poker involves a deck of 6 cards (jack, queen, and king in two suits). 
Players have a limitless number of chips.
Each player antes 1 chip to play and receives an initial private card; in the first round players can bet a fixed amount of 2 chips, in the second round can bet 4 chips, with a maximum of two raises in each round.
Before the second round starts, a public card is revealed.
The corresponding meta-game involves 2 players with 3 strategies each, which correspond to the first three epochs of the PSRO algorithm.
Leduc poker is a commonly used benchmark in the computer poker literature~\cite{Southey05Bayes}: our implementation contains 936 information states (approximately 50 times larger then 2-player Kuhn poker), and is non-zero sum due to penalties imposed by selecting of illegal moves, see~\cite[Appendix D.1]{Lanctot17} for details.

We consider in \cref{fig:mcc_psro} the Markov chain corresponding to the PSRO dataset, with the corresponding $\alpha$-Rank yielding profile $(0,0)$ as the top-ranked strategy, which receives 1.0 of the stationary distribution mass and essentially consumes the entire strategy space in the long-term of the evolutionary dynamics. 
This corresponds well to the result of \cite{Tuyls18}, which also concluded that this strategy profile consumes the entire strategy space under the replicator dynamics;
in their approach, however, an equilibrium selection problem had to be dealt with using human-in-the-loop intervention due to the population-wise dynamics decomposition their approach relies on.
Here, we need no such intervention as $\alpha$-Rank directly yields the overall ranking of all strategy profiles.

\section{Discussion}
\begin{figure}[t!]
    \centering
    \includegraphics[width=\linewidth,page=4]{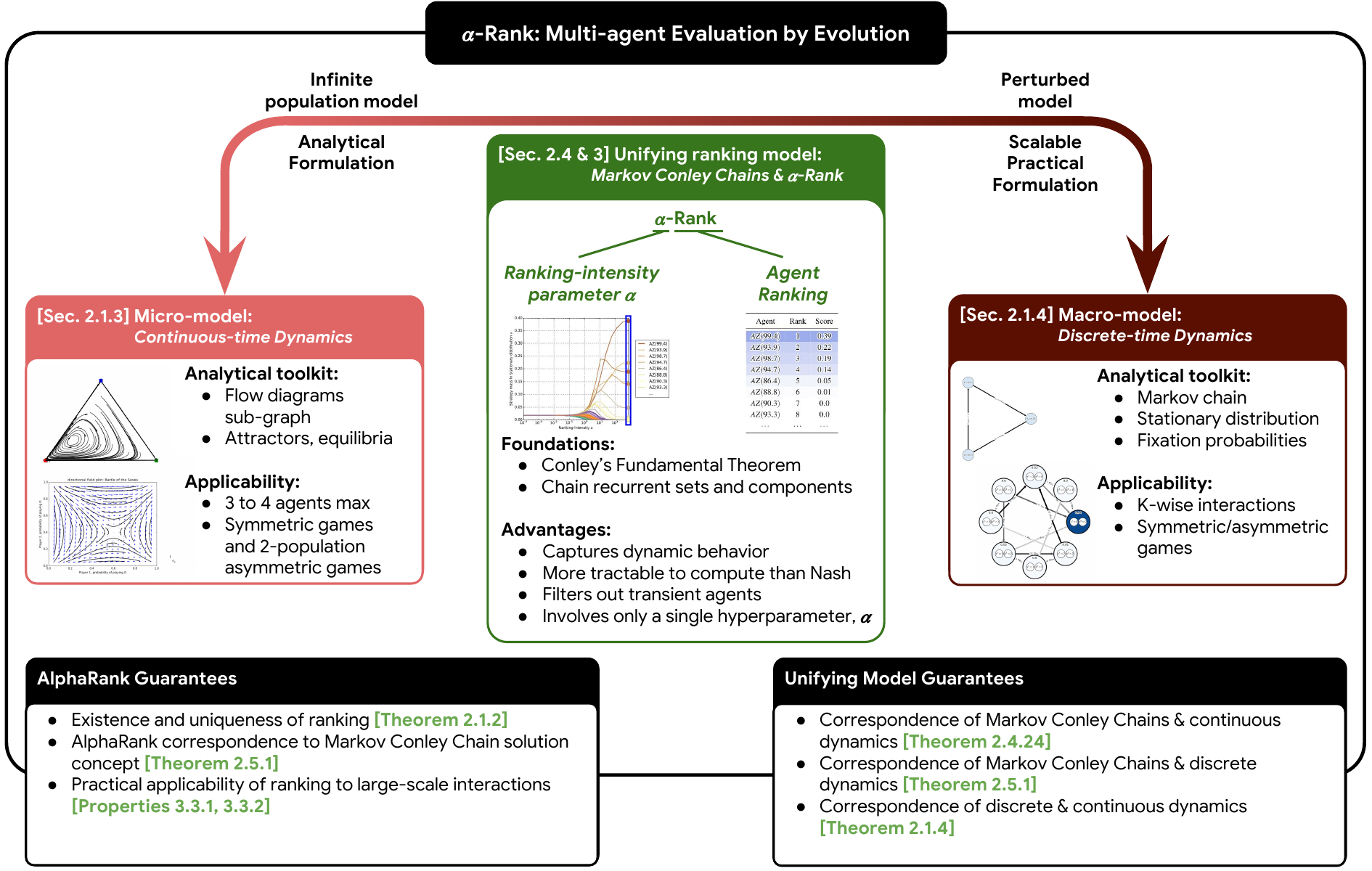}
    \caption{A retrospective look on the paper contributions. We introduced a general descriptive multi-agent evaluation method, called \emph{$\alpha$-Rank}, which is practical in the sense that it is easily applicable in complex game-theoretic settings, and theoretically-grounded in a solution concept called Markov-Conley chains (MCCs). $\alpha$-Rank has a strong theoretical and specifically evolutionary interpretation; the overarching perspective considers a chain of models of increasing complexity, with a discrete-time macro-dynamics model on one end, continuous-time micro-dynamics on the other end, and MCCs as the link in between. We provided both scalability properties and theoretical guarantees for the overall ranking methodology.}
    \label{fig:summary}
\end{figure}


We introduced a general descriptive multi-agent evaluation method, called \emph{$\alpha$-Rank}, which is practical and general in the sense that it is easily applicable in complex game-theoretic settings, including $K$-player asymmetric games that existing meta-game evaluation methods such as \cite{TuylsSym,Tuyls18} cannot feasibly be applied to.
The techniques underlying $\alpha$-Rank were motivated due to the fundamental incompatibility identified between the dynamical processes typically used to model interactions of agents in meta-games, and the Nash solution concept typically used to draw conclusions about these interactions.
Using the Nash equilibrium as a solution concept for meta-game evaluation in these dynamical models is in many ways problematic: computing a Nash equilibrium is not only computationally difficult \cite{VonStengel20021723,Daskalakis06thecomplexity}, and there are also intractable equilibrium selection issues even if Nash equilibria can be computed \cite{Harsan88,Avis10,goldberg2013complexity}.
$\alpha$-Rank, instead, is theoretically-grounded in a novel solution concept called Markov-Conley chains (MCCs), which are inherently dynamical in nature.
A key feature of $\alpha$-Rank is that it relies on only a single hyperparameter, its ranking-intensity value $\alpha$, with sufficiently high values of $\alpha$ (as determined via a parameter sweep) yielding closest correspondence to MCCs.

The combination of MCCs and $\alpha$-Rank yields a principled methodology with a strong evolutionary interpretation of agent rankings, as outlined in \cref{fig:summary}; 
this overarching perspective considers a spectrum of evolutionary models of increasing complexity.
On one end of the spectrum, the continuous-time dynamics micro-model provides detailed insights into the simplex, illustrating flows, attractors, and equilibria of agent interactions.
On the other end, the discrete-time dynamics macro-model provides high-level insights of the time limit behavior of the system as modeled by a Markov chain over interacting agents.
The unifying link between these models is the MCC solution concept, which builds on the dynamical theory foundations of Conley \cite{conley1978isolated} and the topological concept of chain components. 
We provided both scalability properties and theoretical guarantees for our ranking method.
Finally, we evaluated the approach on an extensive range of meta-game domains including AlphaGo \cite{DSilverHMGSDSAPL16}, AlphaZero \cite{silver2018general}, MuJoCo Soccer \cite{liu2018emergent}, 
\ifKeepSCII
\else
and Poker \cite{Lanctot17},
\fi
which exhibit a range of complexities in terms of payoff asymmetries, number of players, and number of agents involved.
We strongly believe that the generality of $\alpha$-Rank will enable it to play an important role in evaluation of AI agents, e.g., on leaderboards. 
More critically, we believe that the computational feasibility of the approach, even when many agents are involved (e.g., AlphaZero), makes its integration into the agent \emph{training} pipeline a natural avenue for future work.

\bibliographystyle{plainnat}


\bibliography{all_bibs}

\begin{thebibliography}{100}
\expandafter\ifx\csname url\endcsname\relax
  \def\url#1{\texttt{#1}}\fi
\expandafter\ifx\csname urlprefix\endcsname\relax\def\urlprefix{URL }\fi
\expandafter\ifx\csname doiprefix\endcsname\relax\def\doiprefix{DOI }\fi
\providecommand{\bibinfo}[2]{#2}
\providecommand{\eprint}[2][]{\url{#2}}

\bibitem{DSilverHMGSDSAPL16}
\bibinfo{author}{Silver, D.} \emph{et~al.}
\newblock \bibinfo{journal}{\bibinfo{title}{{Mastering the game of Go with deep
  neural networks and tree search}}}.
\newblock {\emph{\JournalTitle{Nature}}} \textbf{\bibinfo{volume}{529}},
  \bibinfo{pages}{484--489} (\bibinfo{year}{2016}).

\bibitem{silver:17}
\bibinfo{author}{Silver, D.} \emph{et~al.}
\newblock \bibinfo{journal}{\bibinfo{title}{{Mastering the game of Go without
  human knowledge}}}.
\newblock {\emph{\JournalTitle{Nature}}} \textbf{\bibinfo{volume}{550}},
  \bibinfo{pages}{354--359} (\bibinfo{year}{2017}).

\bibitem{silver2018general}
\bibinfo{author}{Silver, D.} \emph{et~al.}
\newblock \bibinfo{journal}{\bibinfo{title}{A general reinforcement learning
  algorithm that masters chess, shogi, and go through self-play}}.
\newblock {\emph{\JournalTitle{Science}}} \textbf{\bibinfo{volume}{362}},
  \bibinfo{pages}{1140--1144} (\bibinfo{year}{2018}).

\bibitem{Morav}
\bibinfo{author}{Morav{\v c}{\'\i}k, M.} \emph{et~al.}
\newblock \bibinfo{journal}{\bibinfo{title}{{DeepStack}: Expert-level
  artificial intelligence in heads-up no-limit poker}}.
\newblock {\emph{\JournalTitle{Science}}} \textbf{\bibinfo{volume}{356}},
  \bibinfo{pages}{508--513} (\bibinfo{year}{2017}).

\bibitem{liu2018emergent}
\bibinfo{author}{Liu, S.} \emph{et~al.}
\newblock \bibinfo{title}{Emergent coordination through competition}.
\newblock In \emph{\bibinfo{booktitle}{International Conference on Learning
  Representations}} (\bibinfo{year}{2019}).
\newblock \urlprefix\url{https://openreview.net/forum?id=BkG8sjR5Km}.

\bibitem{Walsh02}
\bibinfo{author}{Walsh, W.~E.}, \bibinfo{author}{Das, R.},
  \bibinfo{author}{Tesauro, G.} \& \bibinfo{author}{Kephart, J.}
\newblock \bibinfo{title}{Analyzing complex strategic interactions in
  multi-agent games}.
\newblock In \emph{\bibinfo{booktitle}{AAAI-02 Workshop on Game Theoretic and
  Decision Theoretic Agents, 2002.}} (\bibinfo{year}{2002}).

\bibitem{Wellman06}
\bibinfo{author}{Wellman, M.~P.}
\newblock \bibinfo{title}{Methods for empirical game-theoretic analysis}.
\newblock In \emph{\bibinfo{booktitle}{Proceedings, The Twenty-First National
  Conference on Artificial Intelligence and the Eighteenth Innovative
  Applications of Artificial Intelligence Conference, July 16-20, 2006, Boston,
  Massachusetts, {USA}}}, \bibinfo{pages}{1552--1556} (\bibinfo{year}{2006}).

\bibitem{Tuyls18}
\bibinfo{author}{Tuyls, K.}, \bibinfo{author}{Perolat, J.},
  \bibinfo{author}{Lanctot, M.}, \bibinfo{author}{Leibo, J.~Z.} \&
  \bibinfo{author}{Graepel, T.}
\newblock \bibinfo{title}{{ A Generalised Method for Empirical Game Theoretic
  Analysis}}.
\newblock In \emph{\bibinfo{booktitle}{AAMAS, Stockholm, Sweden}}
  (\bibinfo{year}{2018}).

\bibitem{TuylsSym}
\bibinfo{author}{Tuyls, K.} \emph{et~al.}
\newblock \bibinfo{journal}{\bibinfo{title}{Symmetric decomposition of
  asymmetric games}}.
\newblock {\emph{\JournalTitle{Scientific Reports}}}
  \textbf{\bibinfo{volume}{8}}, \bibinfo{pages}{1015} (\bibinfo{year}{2018}).

\bibitem{TuylsP07}
\bibinfo{author}{Tuyls, K.} \& \bibinfo{author}{Parsons, S.}
\newblock \bibinfo{journal}{\bibinfo{title}{What evolutionary game theory tells
  us about multiagent learning}}.
\newblock {\emph{\JournalTitle{Artif. Intell.}}}
  \textbf{\bibinfo{volume}{171}}, \bibinfo{pages}{406--416}
  (\bibinfo{year}{2007}).

\bibitem{Zeeman80}
\bibinfo{author}{Zeeman, E.}
\newblock \bibinfo{journal}{\bibinfo{title}{Population dynamics from game
  theory}}.
\newblock {\emph{\JournalTitle{Lecture Notes in Mathematics, Global theory of
  dynamical systems}}} \textbf{\bibinfo{volume}{819}} (\bibinfo{year}{1980}).

\bibitem{Zeeman81}
\bibinfo{author}{Zeeman, E.}
\newblock \bibinfo{journal}{\bibinfo{title}{Dynamics of the evolution of animal
  conflicts}}.
\newblock {\emph{\JournalTitle{Theoretical Biology}}}
  \textbf{\bibinfo{volume}{89}}, \bibinfo{pages}{249--270}
  (\bibinfo{year}{1981}).

\bibitem{Weibull97}
\bibinfo{author}{Weibull, J.}
\newblock \bibinfo{journal}{\bibinfo{title}{Evolutionary game theory}}.
\newblock {\emph{\JournalTitle{MIT press}}}  (\bibinfo{year}{1997}).

\bibitem{Hofbauer96}
\bibinfo{author}{Hofbauer, J.}
\newblock \bibinfo{journal}{\bibinfo{title}{Evolutionary dynamics for bimatrix
  games: A {H}amiltonian system?}}
\newblock {\emph{\JournalTitle{J. of Math. Biology}}}
  \textbf{\bibinfo{volume}{34}}, \bibinfo{pages}{675--688}
  (\bibinfo{year}{1996}).

\bibitem{Gintis09}
\bibinfo{author}{Gintis, H.}
\newblock \bibinfo{journal}{\bibinfo{title}{Game theory evolving (2nd
  edition)}}.
\newblock {\emph{\JournalTitle{University Press, Princeton NJ}}}
  (\bibinfo{year}{2009}).

\bibitem{traulsen2005coevolutionary}
\bibinfo{author}{Traulsen, A.}, \bibinfo{author}{Claussen, J.~C.} \&
  \bibinfo{author}{Hauert, C.}
\newblock \bibinfo{journal}{\bibinfo{title}{Coevolutionary dynamics: from
  finite to infinite populations}}.
\newblock {\emph{\JournalTitle{Physical review letters}}}
  \textbf{\bibinfo{volume}{95}}, \bibinfo{pages}{238701}
  (\bibinfo{year}{2005}).

\bibitem{Traulsen06a}
\bibinfo{author}{Traulsen, A.}, \bibinfo{author}{Nowak, M.~A.} \&
  \bibinfo{author}{Pacheco, J.~M.}
\newblock \bibinfo{journal}{\bibinfo{title}{Stochastic dynamics of invasion and
  fixation}}.
\newblock {\emph{\JournalTitle{Phys. Rev. E}}} \textbf{\bibinfo{volume}{74}},
  \bibinfo{pages}{011909} (\bibinfo{year}{2006}).

\bibitem{Santos11}
\bibinfo{author}{Santos, F.~C.}, \bibinfo{author}{Pacheco, J.~M.} \&
  \bibinfo{author}{Skyrms, B.}
\newblock \bibinfo{journal}{\bibinfo{title}{Co-evolution of pre-play signaling
  and cooperation}}.
\newblock {\emph{\JournalTitle{Journal of Theoretical Biology}}}
  \textbf{\bibinfo{volume}{274}}, \bibinfo{pages}{30--35}
  (\bibinfo{year}{2011}).

\bibitem{Segbroek12}
\bibinfo{author}{Segbroeck, S.~V.}, \bibinfo{author}{Pacheco, J.~M.},
  \bibinfo{author}{Lenaerts, T.} \& \bibinfo{author}{Santos, F.~C.}
\newblock \bibinfo{journal}{\bibinfo{title}{Emergence of fairness in repeated
  group interactions}}.
\newblock {\emph{\JournalTitle{Physical Review Letters}}}
  \textbf{\bibinfo{volume}{108}}, \bibinfo{pages}{158104}
  (\bibinfo{year}{2012}).

\bibitem{veller2016finite}
\bibinfo{author}{Veller, C.} \& \bibinfo{author}{Hayward, L.~K.}
\newblock \bibinfo{journal}{\bibinfo{title}{Finite-population evolution with
  rare mutations in asymmetric games}}.
\newblock {\emph{\JournalTitle{Journal of Economic Theory}}}
  \textbf{\bibinfo{volume}{162}}, \bibinfo{pages}{93--113}
  (\bibinfo{year}{2016}).

\bibitem{VonStengel20021723}
\bibinfo{author}{von Stengel, B.}
\newblock \bibinfo{title}{Computing equilibria for two-person games}.
\newblock In \emph{\bibinfo{booktitle}{Handbook of Game Theory with Economic
  Applications}}, vol.~\bibinfo{volume}{3}, \bibinfo{pages}{1723 -- 1759}
  (\bibinfo{publisher}{Elsevier}, \bibinfo{year}{2002}).

\bibitem{Daskalakis06thecomplexity}
\bibinfo{author}{Daskalakis, C.}, \bibinfo{author}{Goldberg, P.~W.} \&
  \bibinfo{author}{Papadimitriou, C.~H.}
\newblock \bibinfo{title}{The complexity of computing a {N}ash equilibrium}.
\newblock In \emph{\bibinfo{booktitle}{Proceedings of the 38th Annual {ACM}
  Symposium on Theory of Computing, Seattle, WA, USA, May 21-23, 2006}},
  \bibinfo{pages}{71--78} (\bibinfo{publisher}{ACM Press},
  \bibinfo{year}{2006}).

\bibitem{Harsan88}
\bibinfo{author}{Harsanyi, J.} \& \bibinfo{author}{Selten, R.}
\newblock \emph{\bibinfo{title}{A General Theory of Equilibrium Selection in
  Games}}, vol.~\bibinfo{volume}{1} (\bibinfo{publisher}{The MIT Press},
  \bibinfo{year}{1988}), \bibinfo{edition}{1} edn.

\bibitem{Avis10}
\bibinfo{author}{Avis, D.}, \bibinfo{author}{Rosenberg, G.},
  \bibinfo{author}{Savani, R.} \& \bibinfo{author}{von Stengel, B.}
\newblock \bibinfo{journal}{\bibinfo{title}{Enumeration of nash equilibria for
  two-player games}}.
\newblock {\emph{\JournalTitle{Economic Theory}}}
  \textbf{\bibinfo{volume}{42}}, \bibinfo{pages}{9--37} (\bibinfo{year}{2010}).

\bibitem{goldberg2013complexity}
\bibinfo{author}{Goldberg, P.~W.}, \bibinfo{author}{Papadimitriou, C.~H.} \&
  \bibinfo{author}{Savani, R.}
\newblock \bibinfo{journal}{\bibinfo{title}{The complexity of the homotopy
  method, equilibrium selection, and {Lemke-Howson} solutions}}.
\newblock {\emph{\JournalTitle{ACM Transactions on Economics and Computation}}}
  \textbf{\bibinfo{volume}{1}}, \bibinfo{pages}{9} (\bibinfo{year}{2013}).

\bibitem{Papadimitriou:2016:NEC:2840728.2840757}
\bibinfo{author}{Papadimitriou, C.} \& \bibinfo{author}{Piliouras, G.}
\newblock \bibinfo{title}{From {Nash} equilibria to chain recurrent sets:
  Solution concepts and topology}.
\newblock In \emph{\bibinfo{booktitle}{Proceedings of the 2016 ACM Conference
  on Innovations in Theoretical Computer Science}}, ITCS '16,
  \bibinfo{pages}{227--235} (\bibinfo{publisher}{ACM}, \bibinfo{address}{New
  York, NY, USA}, \bibinfo{year}{2016}).

\bibitem{Kakutani41}
\bibinfo{author}{Kakutani, S.}
\newblock \bibinfo{journal}{\bibinfo{title}{A generalization of {Brouwer's}
  fixed point theorem}}.
\newblock {\emph{\JournalTitle{Duke Mathematical Journal}}}
  \textbf{\bibinfo{volume}{8}}, \bibinfo{pages}{457--459}
  (\bibinfo{year}{1941}).

\bibitem{conley1978isolated}
\bibinfo{author}{Conley, C.~C.}
\newblock \emph{\bibinfo{title}{Isolated invariant sets and the Morse index}}.
\newblock \bibinfo{number}{38} (\bibinfo{publisher}{American Mathematical
  Soc.}, \bibinfo{year}{1978}).

\bibitem{Lanctot17}
\bibinfo{author}{Lanctot, M.} \emph{et~al.}
\newblock \bibinfo{title}{A unified game-theoretic approach to multiagent
  reinforcement learning}.
\newblock In \emph{\bibinfo{booktitle}{Advances in Neural Information
  Processing Systems 30}}, \bibinfo{pages}{4190--4203} (\bibinfo{year}{2017}).

\bibitem{Hofbauer98}
\bibinfo{author}{J.~Hofbauer, J.} \& \bibinfo{author}{Sigmund, K.}
\newblock \bibinfo{journal}{\bibinfo{title}{Evolutionary games and population
  dynamics}}.
\newblock {\emph{\JournalTitle{Cambridge University Press}}}
  (\bibinfo{year}{1998}).

\bibitem{Cressman03}
\bibinfo{author}{Cressman, R.}
\newblock \emph{\bibinfo{title}{Evolutionary Dynamics and Extensive Form
  Games}} (\bibinfo{publisher}{The MIT Press}, \bibinfo{year}{2003}).

\bibitem{Taylor78}
\bibinfo{author}{Taylor, P.} \& \bibinfo{author}{Jonker, L.}
\newblock \bibinfo{journal}{\bibinfo{title}{Evolutionarily stable strategies
  and game dynamics}}.
\newblock {\emph{\JournalTitle{Mathematical Biosciences}}}
  \textbf{\bibinfo{volume}{40}}, \bibinfo{pages}{145--156}
  (\bibinfo{year}{1978}).

\bibitem{Schuster1983533}
\bibinfo{author}{Schuster, P.} \& \bibinfo{author}{Sigmund, K.}
\newblock \bibinfo{journal}{\bibinfo{title}{Replicator dynamics}}.
\newblock {\emph{\JournalTitle{Journal of Theoretical Biology}}}
  \textbf{\bibinfo{volume}{100}}, \bibinfo{pages}{533 -- 538}
  (\bibinfo{year}{1983}).
\newblock
  \urlprefix\url{http://www.sciencedirect.com/science/article/pii/0022519383904459}.
\newblock \doiprefix http://dx.doi.org/10.1016/0022-5193(83)90445-9.

\bibitem{BloembergenTHK15}
\bibinfo{author}{Bloembergen, D.}, \bibinfo{author}{Tuyls, K.},
  \bibinfo{author}{Hennes, D.} \& \bibinfo{author}{Kaisers, M.}
\newblock \bibinfo{journal}{\bibinfo{title}{Evolutionary dynamics of
  multi-agent learning: {A} survey}}.
\newblock {\emph{\JournalTitle{J. Artif. Intell. Res. {(JAIR)}}}}
  \textbf{\bibinfo{volume}{53}}, \bibinfo{pages}{659--697}
  (\bibinfo{year}{2015}).

\bibitem{fudenberg2006imitation}
\bibinfo{author}{Fudenberg, D.} \& \bibinfo{author}{Imhof, L.~A.}
\newblock \bibinfo{journal}{\bibinfo{title}{Imitation processes with small
  mutations}}.
\newblock {\emph{\JournalTitle{Journal of Economic Theory}}}
  \textbf{\bibinfo{volume}{131}}, \bibinfo{pages}{251--262}
  (\bibinfo{year}{2006}).

\bibitem{Nowak793}
\bibinfo{author}{Nowak, M.~A.} \& \bibinfo{author}{Sigmund, K.}
\newblock \bibinfo{journal}{\bibinfo{title}{Evolutionary dynamics of biological
  games}}.
\newblock {\emph{\JournalTitle{Science}}} \textbf{\bibinfo{volume}{303}},
  \bibinfo{pages}{793--799} (\bibinfo{year}{2004}).

\bibitem{Traulsen06b}
\bibinfo{author}{Traulsen, A.}, \bibinfo{author}{Pacheco, J.~M.} \&
  \bibinfo{author}{Imhof, L.~A.}
\newblock \bibinfo{journal}{\bibinfo{title}{Stochasticity and evolutionary
  stability}}.
\newblock {\emph{\JournalTitle{Phys. Rev. E}}} \textbf{\bibinfo{volume}{74}},
  \bibinfo{pages}{021905} (\bibinfo{year}{2006}).

\bibitem{claussen2008discrete}
\bibinfo{author}{Claussen, J.~C.}
\newblock \bibinfo{journal}{\bibinfo{title}{Discrete stochastic processes,
  replicator and {Fokker-Planck} equations of coevolutionary dynamics in finite
  and infinite populations}}.
\newblock {\emph{\JournalTitle{arXiv preprint arXiv:0803.2443}}}
  (\bibinfo{year}{2008}).

\bibitem{TaylKarl98}
\bibinfo{author}{Taylor, H.~M.} \& \bibinfo{author}{Karlin, S.}
\newblock \emph{\bibinfo{title}{An Introduction To Stochastic Modeling}}
  (\bibinfo{publisher}{Academic Press}, \bibinfo{year}{1998}),
  \bibinfo{edition}{third edition} edn.

\bibitem{Frong}
\bibinfo{author}{Daskalakis, C.}, \bibinfo{author}{Frongillo, R.},
  \bibinfo{author}{Papadimitriou, C.~H.}, \bibinfo{author}{Pierrakos, G.} \&
  \bibinfo{author}{Valiant, G.}
\newblock \bibinfo{title}{On learning algorithms for {N}ash equilibria}.
\newblock In \emph{\bibinfo{booktitle}{International Symposium on Algorithmic
  Game Theory}}, \bibinfo{pages}{114--125} (\bibinfo{organization}{Springer},
  \bibinfo{year}{2010}).

\bibitem{Hart03}
\bibinfo{author}{Hart, S.} \& \bibinfo{author}{Mas-Colell, A.}
\newblock \bibinfo{journal}{\bibinfo{title}{Uncoupled dynamics do not lead to
  nash equilibrium}}.
\newblock {\emph{\JournalTitle{American Economic Review}}}
  \textbf{\bibinfo{volume}{93}}, \bibinfo{pages}{1830--1836}
  (\bibinfo{year}{2003}).

\bibitem{viossat07}
\bibinfo{author}{Viossat, Y.}
\newblock \bibinfo{journal}{\bibinfo{title}{The replicator dynamics does not
  lead to correlated equilibria}}.
\newblock {\emph{\JournalTitle{Games and Economic Behavior}}}
  \textbf{\bibinfo{volume}{59}}, \bibinfo{pages}{397--407}
  (\bibinfo{year}{2007}).

\bibitem{piliouras2017learning}
\bibinfo{author}{Piliouras, G.} \& \bibinfo{author}{Schulman, L.~J.}
\newblock \bibinfo{journal}{\bibinfo{title}{Learning dynamics and the
  co-evolution of competing sexual species}}.
\newblock {\emph{\JournalTitle{arXiv preprint arXiv:1711.06879}}}
  (\bibinfo{year}{2017}).

\bibitem{Sandholm10}
\bibinfo{author}{Sandholm, W.}
\newblock \emph{\bibinfo{title}{Population Games and Evolutionary Dynamics}}.
\newblock Economic Learning and Social Evolution (\bibinfo{publisher}{MIT
  Press}, \bibinfo{year}{2010}).

\bibitem{Gaunersdorfer}
\bibinfo{author}{Gaunersdorfer, A.} \& \bibinfo{author}{Hofbauer, J.}
\newblock \bibinfo{journal}{\bibinfo{title}{Fictitious play, shapley polygons,
  and the replicator equation}}.
\newblock {\emph{\JournalTitle{Games and Economic Behavior}}}
  \textbf{\bibinfo{volume}{11}}, \bibinfo{pages}{279--303}
  (\bibinfo{year}{1995}).

\bibitem{daskalakis10}
\bibinfo{author}{Daskalakis, C.}, \bibinfo{author}{Frongillo, R.},
  \bibinfo{author}{Papadimitriou, C.}, \bibinfo{author}{Pierrakos, G.} \&
  \bibinfo{author}{Valiant, G.}
\newblock \bibinfo{journal}{\bibinfo{title}{On learning algorithms for {N}ash
  equilibria}}.
\newblock {\emph{\JournalTitle{Algorithmic Game Theory}}}
  \bibinfo{pages}{114--125} (\bibinfo{year}{2010}).

\bibitem{paperics11}
\bibinfo{author}{Kleinberg, R.}, \bibinfo{author}{Ligett, K.},
  \bibinfo{author}{Piliouras, G.} \& \bibinfo{author}{Tardos, {\'E}.}
\newblock \bibinfo{title}{Beyond the {Nash} equilibrium barrier}.
\newblock In \emph{\bibinfo{booktitle}{Symposium on Innovations in Computer
  Science (ICS)}} (\bibinfo{year}{2011}).

\bibitem{sandholm2010population}
\bibinfo{author}{Sandholm, W.~H.}
\newblock \emph{\bibinfo{title}{Population games and evolutionary dynamics}}
  (\bibinfo{publisher}{MIT press}, \bibinfo{year}{2010}).

\bibitem{wagner2013explanatory}
\bibinfo{author}{Wagner, E.}
\newblock \bibinfo{journal}{\bibinfo{title}{The explanatory relevance of nash
  equilibrium: One-dimensional chaos in boundedly rational learning}}.
\newblock {\emph{\JournalTitle{Philosophy of Science}}}
  \textbf{\bibinfo{volume}{80}}, \bibinfo{pages}{783--795}
  (\bibinfo{year}{2013}).

\bibitem{PalaiopanosPP17}
\bibinfo{author}{Palaiopanos, G.}, \bibinfo{author}{Panageas, I.} \&
  \bibinfo{author}{Piliouras, G.}
\newblock \bibinfo{title}{Multiplicative weights update with constant step-size
  in congestion games: Convergence, limit cycles and chaos}.
\newblock In \emph{\bibinfo{booktitle}{NIPS}} (\bibinfo{year}{2017}).

\bibitem{Sato02042002}
\bibinfo{author}{Sato, Y.}, \bibinfo{author}{Akiyama, E.} \&
  \bibinfo{author}{Farmer, J.~D.}
\newblock \bibinfo{journal}{\bibinfo{title}{Chaos in learning a simple
  two-person game}}.
\newblock {\emph{\JournalTitle{Proceedings of the National Academy of
  Sciences}}} \textbf{\bibinfo{volume}{99}}, \bibinfo{pages}{4748--4751}
  (\bibinfo{year}{2002}).

\bibitem{alongi2007recurrence}
\bibinfo{author}{Alongi, J.~M.} \& \bibinfo{author}{Nelson, G.~S.}
\newblock \emph{\bibinfo{title}{Recurrence and Topology}},
  vol.~\bibinfo{volume}{85} (\bibinfo{publisher}{American Mathematical Soc.},
  \bibinfo{year}{2007}).

\bibitem{norton1995fundamental}
\bibinfo{author}{Norton, D.~E.}
\newblock \bibinfo{journal}{\bibinfo{title}{The fundamental theorem of
  dynamical systems}}.
\newblock {\emph{\JournalTitle{Commentationes Mathematicae Universitatis
  Carolinae}}} \textbf{\bibinfo{volume}{36}}, \bibinfo{pages}{585--597}
  (\bibinfo{year}{1995}).

\bibitem{monderer:96}
\bibinfo{author}{Monderer, D.} \& \bibinfo{author}{Shapley, L.~S.}
\newblock \bibinfo{journal}{\bibinfo{title}{{Potential Games}}}.
\newblock {\emph{\JournalTitle{Games and Economic Behavior}}}
  \textbf{\bibinfo{volume}{14}}, \bibinfo{pages}{124--143}
  (\bibinfo{year}{1996}).

\bibitem{galla2013complex}
\bibinfo{author}{Galla, T.} \& \bibinfo{author}{Farmer, J.~D.}
\newblock \bibinfo{journal}{\bibinfo{title}{Complex dynamics in learning
  complicated games}}.
\newblock {\emph{\JournalTitle{Proceedings of the National Academy of
  Sciences}}} \textbf{\bibinfo{volume}{110}}, \bibinfo{pages}{1232--1236}
  (\bibinfo{year}{2013}).

\bibitem{panageas2016average}
\bibinfo{author}{Panageas, I.} \& \bibinfo{author}{Piliouras, G.}
\newblock \bibinfo{title}{Average case performance of replicator dynamics in
  potential games via computing regions of attraction}.
\newblock In \emph{\bibinfo{booktitle}{Proceedings of the 2016 ACM Conference
  on Economics and Computation}}, \bibinfo{pages}{703--720}
  (\bibinfo{organization}{ACM}, \bibinfo{year}{2016}).

\bibitem{bomze1983lotka}
\bibinfo{author}{Bomze, I.~M.}
\newblock \bibinfo{journal}{\bibinfo{title}{Lotka-volterra equation and
  replicator dynamics: a two-dimensional classification}}.
\newblock {\emph{\JournalTitle{Biological cybernetics}}}
  \textbf{\bibinfo{volume}{48}}, \bibinfo{pages}{201--211}
  (\bibinfo{year}{1983}).

\bibitem{Bomze95}
\bibinfo{author}{Bomze, I.~M.}
\newblock \bibinfo{journal}{\bibinfo{title}{Lotka-volterra equation and
  replicator dynamics: new issues in classification.}}
\newblock {\emph{\JournalTitle{Biological Cybernetics}}}
  \textbf{\bibinfo{volume}{72}}, \bibinfo{pages}{447--453}
  (\bibinfo{year}{1995}).

\bibitem{Shoham:answer}
\bibinfo{author}{Shoham, Y.}, \bibinfo{author}{Powers, R.} \&
  \bibinfo{author}{Grenager, T.}
\newblock \bibinfo{journal}{\bibinfo{title}{If multi-agent learning is the
  answer, what is the question?}}
\newblock {\emph{\JournalTitle{Artificial Intelligence}}}
  \textbf{\bibinfo{volume}{171}}, \bibinfo{pages}{365--377}
  (\bibinfo{year}{2007}).

\bibitem{DavisBB14}
\bibinfo{author}{Davis, T.}, \bibinfo{author}{Burch, N.} \&
  \bibinfo{author}{Bowling, M.}
\newblock \bibinfo{title}{Using response functions to measure strategy
  strength}.
\newblock In \emph{\bibinfo{booktitle}{Proceedings of the Twenty-Eighth {AAAI}
  Conference on Artificial Intelligence, July 27 -31, 2014, Qu{\'{e}}bec City,
  Qu{\'{e}}bec, Canada.}}, \bibinfo{pages}{630--636} (\bibinfo{year}{2014}).

\bibitem{Conitzer18}
\bibinfo{author}{Conitzer, V.}
\newblock \bibinfo{journal}{\bibinfo{title}{The exact computational complexity
  of evolutionarily stable strategies}}.
\newblock {\emph{\JournalTitle{CoRR}}}
  \textbf{\bibinfo{volume}{abs/1805.02226}} (\bibinfo{year}{2018}).

\bibitem{Etessami2008}
\bibinfo{author}{Etessami, K.} \& \bibinfo{author}{Lochbihler, A.}
\newblock \bibinfo{journal}{\bibinfo{title}{The computational complexity of
  evolutionarily stable strategies}}.
\newblock {\emph{\JournalTitle{International Journal of Game Theory}}}
  (\bibinfo{year}{2008}).

\bibitem{veller2017red}
\bibinfo{author}{Veller, C.}, \bibinfo{author}{Hayward, L.~K.},
  \bibinfo{author}{Hilbe, C.} \& \bibinfo{author}{Nowak, M.~A.}
\newblock \bibinfo{journal}{\bibinfo{title}{The red queen and king in finite
  populations}}.
\newblock {\emph{\JournalTitle{Proceedings of the National Academy of
  Sciences}}} \textbf{\bibinfo{volume}{114}}, \bibinfo{pages}{E5396--E5405}
  (\bibinfo{year}{2017}).

\bibitem{Balduzzi18}
\bibinfo{author}{Balduzzi, D.}, \bibinfo{author}{Tuyls, K.},
  \bibinfo{author}{Perolat, J.} \& \bibinfo{author}{Graepel, T.}
\newblock \bibinfo{journal}{\bibinfo{title}{{Re-evaluating Evaluation}}}.
\newblock {\emph{\JournalTitle{arXiv}}}  (\bibinfo{year}{2018}).

\bibitem{todorov:12}
\bibinfo{author}{Todorov, E.}, \bibinfo{author}{Erez, T.} \&
  \bibinfo{author}{Tassa, Y.}
\newblock \bibinfo{title}{{Mujoco: A physics engine for model-based control}}.
\newblock In \emph{\bibinfo{booktitle}{IROS}} (\bibinfo{year}{2012}).

\bibitem{Southey08}
\bibinfo{author}{Southey, F.}, \bibinfo{author}{Hoehn, B.} \&
  \bibinfo{author}{Holte, R.~C.}
\newblock \bibinfo{journal}{\bibinfo{title}{Effective short-term opponent
  exploitation in simplified poker}}.
\newblock {\emph{\JournalTitle{Machine Learning}}}
  \textbf{\bibinfo{volume}{74}}, \bibinfo{pages}{159--189}
  (\bibinfo{year}{2009}).

\bibitem{Szafron13}
\bibinfo{author}{Szafron, D.}, \bibinfo{author}{Gibson, R.} \&
  \bibinfo{author}{Sturtevant, N.}
\newblock \bibinfo{title}{A parameterized family of equilibrium profiles for
  three-player {K}uhn poker}.
\newblock In \emph{\bibinfo{booktitle}{Proceedings of the Twelfth International
  Conference on Autonomous Agents and Multiagent Systems (AAMAS)}},
  \bibinfo{pages}{247--254} (\bibinfo{year}{2013}).

\bibitem{Lanctot14Further}
\bibinfo{author}{Lanctot, M.}
\newblock \bibinfo{title}{Further developments of extensive-form replicator
  dynamics using the sequence-form representation}.
\newblock In \emph{\bibinfo{booktitle}{Proceedings of the Thirteenth
  International Conference on Autonomous Agents and Multi-Agent Systems
  ({AAMAS})}}, \bibinfo{pages}{1257--1264} (\bibinfo{year}{2014}).

\bibitem{Heinrich15FSP}
\bibinfo{author}{Heinrich, J.}, \bibinfo{author}{Lanctot, M.} \&
  \bibinfo{author}{Silver, D.}
\newblock \bibinfo{title}{Fictitious self-play in extensive-form games}.
\newblock In \emph{\bibinfo{booktitle}{Proceedings of the 32nd International
  Conference on Machine Learning ({ICML} 2015)}} (\bibinfo{year}{2015}).

\bibitem{Southey05Bayes}
\bibinfo{author}{Southey, F.} \emph{et~al.}
\newblock \bibinfo{title}{Bayes' bluff: Opponent modelling in poker}.
\newblock In \emph{\bibinfo{booktitle}{Proceedings of the Twenty-First
  Conference on Uncertainty in Artificial Intelligence (UAI 2005)}}
  (\bibinfo{year}{2005}).

\bibitem{Walsh03}
\bibinfo{author}{Walsh, W.~E.}, \bibinfo{author}{Parkes, D.~C.} \&
  \bibinfo{author}{Das, R.}
\newblock \bibinfo{title}{Choosing samples to compute heuristic-strategy {Nash}
  equilibrium}.
\newblock In \emph{\bibinfo{booktitle}{Proceedings of the Fifth Workshop on
  Agent-Mediated Electronic Commerce}} (\bibinfo{year}{2003}).

\bibitem{Vorobeychik07}
\bibinfo{author}{Vorobeychik, Y.}, \bibinfo{author}{Wellman, M.~P.} \&
  \bibinfo{author}{Singh, S.}
\newblock \bibinfo{journal}{\bibinfo{title}{Learning payoff functions in
  infinite games}}.
\newblock {\emph{\JournalTitle{Machine Learning}}}
  \textbf{\bibinfo{volume}{67}}, \bibinfo{pages}{145--168}
  (\bibinfo{year}{2007}).

\bibitem{Wah15}
\bibinfo{author}{Wah, E.}, \bibinfo{author}{Hurd, D.} \&
  \bibinfo{author}{Wellman, M.}
\newblock \bibinfo{title}{Strategic market choice: Frequent call markets vs.
  continuous double auctions for fast and slow traders}.
\newblock In \emph{\bibinfo{booktitle}{Proceedings of the Third EAI Conference
  on Auctions, Market Mechanisms, and Their Applications}}
  (\bibinfo{year}{2015}).

\bibitem{Brinkman16}
\bibinfo{author}{Brinkman, E.} \& \bibinfo{author}{Wellman, M.}
\newblock \bibinfo{title}{Shading and efficiency in limit-order markets}.
\newblock In \emph{\bibinfo{booktitle}{Proceedings of the IJCAI-16 Workshop on
  Algorithmic Game Theory}} (\bibinfo{year}{2016}).

\bibitem{Wah17}
\bibinfo{author}{Wah, E.}, \bibinfo{author}{Wright, M.} \&
  \bibinfo{author}{Wellman, M.}
\newblock \bibinfo{journal}{\bibinfo{title}{Welfare effects of market making in
  continuous double auctions}}.
\newblock {\emph{\JournalTitle{Journal of Artificial Intelligence Research}}}
  \textbf{\bibinfo{volume}{59}}, \bibinfo{pages}{613--650}
  (\bibinfo{year}{2017}).

\bibitem{Wang18}
\bibinfo{author}{Wang, X.}, \bibinfo{author}{Vorobeychik, Y.} \&
  \bibinfo{author}{Wellman, M.}
\newblock \bibinfo{title}{A cloaking mechanism to mitigate market
  manipulation}.
\newblock In \emph{\bibinfo{booktitle}{Proceedings of the 27th International
  Joint Conference on Artificial Intelligence}}, \bibinfo{pages}{541--547}
  (\bibinfo{year}{2018}).

\bibitem{PonsenTKR09}
\bibinfo{author}{Ponsen, M. J.~V.}, \bibinfo{author}{Tuyls, K.},
  \bibinfo{author}{Kaisers, M.} \& \bibinfo{author}{Ramon, J.}
\newblock \bibinfo{journal}{\bibinfo{title}{An evolutionary game-theoretic
  analysis of poker strategies}}.
\newblock {\emph{\JournalTitle{Entertainment Computing}}}
  \textbf{\bibinfo{volume}{1}}, \bibinfo{pages}{39--45} (\bibinfo{year}{2009}).

\bibitem{Wellman13}
\bibinfo{author}{Wellman, M.}, \bibinfo{author}{Kim, T.} \&
  \bibinfo{author}{Duong, Q.}
\newblock \bibinfo{title}{Analyzing incentives for protocol compliance in
  complex domains: A case study of introduction-based routing}.
\newblock In \emph{\bibinfo{booktitle}{Proceedings of the 12th Workshop on the
  Economics of Information Security}} (\bibinfo{year}{2013}).

\bibitem{Hennes13}
\bibinfo{author}{Hennes, D.}, \bibinfo{author}{Claes, D.} \&
  \bibinfo{author}{Tuyls, K.}
\newblock \bibinfo{title}{Evolutionary advantage of reciprocity in collision
  avoidance}.
\newblock In \emph{\bibinfo{booktitle}{Proceedings of the AAMAS 2013 Workshop
  on Autonomous Robots andMultirobot Systems (ARMS 2013)}}
  (\bibinfo{year}{2013}).

\bibitem{Prakash15}
\bibinfo{author}{Prakash, A.} \& \bibinfo{author}{Wellman, M.}
\newblock \bibinfo{title}{Empirical game-theoretic analysis for moving target
  defense}.
\newblock In \emph{\bibinfo{booktitle}{Proceedings of the Second ACM Workshop
  on Moving Target Defense}} (\bibinfo{year}{2015}).

\bibitem{Wright16}
\bibinfo{author}{Wright, M.}, \bibinfo{author}{Venkatesan, S.},
  \bibinfo{author}{Albenese, M.} \& \bibinfo{author}{Wellman, M.}
\newblock \bibinfo{title}{Moving target defense against {DDoS} attacks: An
  empirical game-theoretic analysis}.
\newblock In \emph{\bibinfo{booktitle}{Proceedings of the Third ACM Workshop on
  Moving Target Defense}} (\bibinfo{year}{2016}).

\bibitem{Nguyen17}
\bibinfo{author}{Nguyen, T.}, \bibinfo{author}{Wright, M.},
  \bibinfo{author}{Wellman, M.} \& \bibinfo{author}{Singh, S.}
\newblock \bibinfo{title}{Multi-stage attack graph security games: Heuristic
  strategies, with empirical game-theoretic analysis}.
\newblock In \emph{\bibinfo{booktitle}{Proceedings of the Fourth ACM Workshop
  on Moving Target Defense}} (\bibinfo{year}{2017}).

\bibitem{Nowakbook06}
\bibinfo{author}{Nowak, M.~A.}
\newblock \emph{\bibinfo{title}{Evolutionary Dynamics: Exploring the Equations
  of Life}} (\bibinfo{publisher}{Harvard University Press},
  \bibinfo{year}{2006}).

\bibitem{Perc18a}
\bibinfo{author}{Liu, L.}, \bibinfo{author}{Wang, S.}, \bibinfo{author}{Chen,
  X.} \& \bibinfo{author}{Perc, M.}
\newblock \bibinfo{journal}{\bibinfo{title}{Evolutionary dynamics in the public
  goods games with switching between punishment and exclusion}}.
\newblock {\emph{\JournalTitle{Chaos}}} \textbf{\bibinfo{volume}{28}},
  \bibinfo{pages}{103105} (\bibinfo{year}{2018}).

\bibitem{Perc18b}
\bibinfo{author}{Szolnoki, A.} \& \bibinfo{author}{Perc, M.}
\newblock \bibinfo{journal}{\bibinfo{title}{Evolutionary dynamics of
  cooperation in neutral populations}}.
\newblock {\emph{\JournalTitle{New Journal of Physics}}}
  \textbf{\bibinfo{volume}{20}}, \bibinfo{pages}{013031}
  (\bibinfo{year}{2018}).

\bibitem{PY1993}
\bibinfo{author}{Young, H.~P.}
\newblock \bibinfo{journal}{\bibinfo{title}{The evolution of conventions}}.
\newblock {\emph{\JournalTitle{Econometrica: Journal of the Econometric
  Society}}} \bibinfo{pages}{57--84} (\bibinfo{year}{1993}).

\bibitem{Basu}
\bibinfo{author}{Basu, K.} \& \bibinfo{author}{Weibull, J.~W.}
\newblock \bibinfo{journal}{\bibinfo{title}{Strategy subsets closed under
  rational behavior}}.
\newblock {\emph{\JournalTitle{Economics Letters}}}
  \textbf{\bibinfo{volume}{36}}, \bibinfo{pages}{141--146}
  (\bibinfo{year}{1991}).

\bibitem{Goemans}
\bibinfo{author}{Goemans, M.}, \bibinfo{author}{Mirrokni, V.} \&
  \bibinfo{author}{Vetta, A.}
\newblock \bibinfo{title}{Sink equilibria and convergence}.
\newblock In \emph{\bibinfo{booktitle}{Foundations of Computer Science, 2005.
  FOCS 2005. 46th Annual IEEE Symposium on}}, \bibinfo{pages}{142--151}
  (\bibinfo{organization}{IEEE}, \bibinfo{year}{2005}).

\bibitem{candogan2011flows}
\bibinfo{author}{Candogan, O.}, \bibinfo{author}{Menache, I.},
  \bibinfo{author}{Ozdaglar, A.} \& \bibinfo{author}{Parrilo, P.~A.}
\newblock \bibinfo{journal}{\bibinfo{title}{Flows and decompositions of games:
  Harmonic and potential games}}.
\newblock {\emph{\JournalTitle{Mathematics of Operations Research}}}
  \textbf{\bibinfo{volume}{36}}, \bibinfo{pages}{474--503}
  (\bibinfo{year}{2011}).

\bibitem{japkowicz2011evaluating}
\bibinfo{author}{Japkowicz, N.} \& \bibinfo{author}{Shah, M.}
\newblock \emph{\bibinfo{title}{Evaluating learning algorithms: a
  classification perspective}} (\bibinfo{publisher}{Cambridge University
  Press}, \bibinfo{year}{2011}).

\bibitem{hernandez2017evaluation}
\bibinfo{author}{Hern{\'a}ndez-Orallo, J.}
\newblock \bibinfo{journal}{\bibinfo{title}{Evaluation in artificial
  intelligence: from task-oriented to ability-oriented measurement}}.
\newblock {\emph{\JournalTitle{Artificial Intelligence Review}}}
  \textbf{\bibinfo{volume}{48}}, \bibinfo{pages}{397--447}
  (\bibinfo{year}{2017}).

\bibitem{hernandez2017measure}
\bibinfo{author}{Hern{\'a}ndez-Orallo, J.}
\newblock \emph{\bibinfo{title}{The measure of all minds: evaluating natural
  and artificial intelligence}} (\bibinfo{publisher}{Cambridge University
  Press}, \bibinfo{year}{2017}).

\bibitem{page1999pagerank}
\bibinfo{author}{Page, L.}, \bibinfo{author}{Brin, S.},
  \bibinfo{author}{Motwani, R.} \& \bibinfo{author}{Winograd, T.}
\newblock \bibinfo{title}{The pagerank citation ranking: Bringing order to the
  web.}
\newblock \bibinfo{type}{Tech. Rep.}, \bibinfo{institution}{Stanford InfoLab}
  (\bibinfo{year}{1999}).

\bibitem{Elo78}
\bibinfo{author}{Elo, A.~E.}
\newblock \emph{\bibinfo{title}{The Rating of Chess players, Past and Present}}
  (\bibinfo{publisher}{Ishi Press International}, \bibinfo{year}{1978}).

\bibitem{Hvattum2010460}
\bibinfo{author}{Hvattum, L.~M.} \& \bibinfo{author}{Arntzen, H.}
\newblock \bibinfo{journal}{\bibinfo{title}{Using {ELO} ratings for match
  result prediction in association football}}.
\newblock {\emph{\JournalTitle{International Journal of Forecasting}}}
  \textbf{\bibinfo{volume}{26}}, \bibinfo{pages}{460 -- 470}
  (\bibinfo{year}{2010}).
\newblock \bibinfo{note}{Sports Forecasting}.

\bibitem{Cattelan13}
\bibinfo{author}{Manuela, C.}, \bibinfo{author}{Cristiano, V.} \&
  \bibinfo{author}{David, F.}
\newblock \bibinfo{journal}{\bibinfo{title}{Dynamic {Bradley–Terry} modelling
  of sports tournaments}}.
\newblock {\emph{\JournalTitle{Journal of the Royal Statistical Society: Series
  C (Applied Statistics)}}} \textbf{\bibinfo{volume}{62}},
  \bibinfo{pages}{135--150} (\bibinfo{year}{2013}).

\bibitem{aldous2017}
\bibinfo{author}{Aldous, D.}
\newblock \bibinfo{journal}{\bibinfo{title}{Elo ratings and the sports model: A
  neglected topic in applied probability?}}
\newblock {\emph{\JournalTitle{Statist. Sci.}}} \textbf{\bibinfo{volume}{32}},
  \bibinfo{pages}{616--629} (\bibinfo{year}{2017}).
\newblock \doiprefix 10.1214/17-STS628.

\bibitem{Sullivan16}
\bibinfo{author}{Sullivan, C.} \& \bibinfo{author}{Cronin, C.}
\newblock \bibinfo{title}{Improving {Elo} rankings for sports experimenting on
  the english premier league}.
\newblock In \emph{\bibinfo{booktitle}{Virginia Tech CSx824/ECEx424 technical
  report}} (\bibinfo{year}{2016}).

\bibitem{Wunderlich18}
\bibinfo{author}{F, W.} \& \bibinfo{author}{D, M.}
\newblock \bibinfo{journal}{\bibinfo{title}{{The Betting Odds Rating System:
  Using soccer forecasts to forecast soccer.}}}
\newblock {\emph{\JournalTitle{PLoS ONE}}} \textbf{\bibinfo{volume}{6}},
  \bibinfo{pages}{e0198668} (\bibinfo{year}{2018}).

\bibitem{mnih:15}
\bibinfo{author}{Mnih, V.} \emph{et~al.}
\newblock \bibinfo{journal}{\bibinfo{title}{Human-level control through deep
  reinforcement learning}}.
\newblock {\emph{\JournalTitle{Nature}}} \textbf{\bibinfo{volume}{518}},
  \bibinfo{pages}{529--533} (\bibinfo{year}{2015}).

\bibitem{Poincare1890}
\bibinfo{author}{Poincar\'{e}, H.}
\newblock \bibinfo{journal}{\bibinfo{title}{Sur le probl\`{e}me des trois corps
  et les \'{e}quations de la dynamique}}.
\newblock {\emph{\JournalTitle{Acta Math}}} \textbf{\bibinfo{volume}{13}}
  (\bibinfo{year}{1890}).

\bibitem{barreira}
\bibinfo{author}{Barreira, L.}
\newblock \bibinfo{title}{Poincare recurrence: old and new}.
\newblock In \emph{\bibinfo{booktitle}{XIVth International Congress on
  Mathematical Physics. World Scientific.}}, \bibinfo{pages}{415--422}
  (\bibinfo{year}{2006}).

\bibitem{bendixson1901courbes}
\bibinfo{author}{Bendixson, I.}
\newblock \bibinfo{journal}{\bibinfo{title}{Sur les courbes d{\'e}finies par
  des {\'e}quations diff{\'e}rentielles}}.
\newblock {\emph{\JournalTitle{Acta Mathematica}}}
  \textbf{\bibinfo{volume}{24}}, \bibinfo{pages}{1--88} (\bibinfo{year}{1901}).

\bibitem{teschl2012ordinary}
\bibinfo{author}{Teschl, G.}
\newblock \emph{\bibinfo{title}{Ordinary differential equations and dynamical
  systems}}, vol. \bibinfo{volume}{140} (\bibinfo{publisher}{American
  Mathematical Soc.}, \bibinfo{year}{2012}).

\bibitem{Meiss2007}
\bibinfo{author}{Meiss, J.}
\newblock \emph{\bibinfo{title}{Differential Dynamical Systems}}
  (\bibinfo{publisher}{SIAM}, \bibinfo{year}{2007}).

\bibitem{Soda14}
\bibinfo{author}{Piliouras, G.} \& \bibinfo{author}{Shamma, J.~S.}
\newblock \bibinfo{title}{Optimization despite chaos: Convex relaxations to
  complex limit sets via {P}oincar\'{e} recurrence}.
\newblock In \emph{\bibinfo{booktitle}{Symposium of Discrete Algorithms
  (SODA)}} (\bibinfo{year}{2014}).

\bibitem{PiliourasAAMAS2014}
\bibinfo{author}{Piliouras, G.}, \bibinfo{author}{Nieto-Granda, C.},
  \bibinfo{author}{Christensen, H.~I.} \& \bibinfo{author}{Shamma, J.~S.}
\newblock \bibinfo{title}{Persistent patterns: Multi-agent learning beyond
  equilibrium and utility}.
\newblock In \emph{\bibinfo{booktitle}{AAMAS}}, \bibinfo{pages}{181--188}
  (\bibinfo{year}{2014}).

\end{thebibliography}

\section*{Acknowledgements}
We are very grateful to G.~Ostrovski, T.~Graepel, E.~Hughes, Y.~Bachrach, K.~Kavukcuoglu, D.~Silver, T.~Hubert, J.~Schrittwieser, S.~Liu, G.~Lever, and D.~Bloembergen, for helpful comments, discussions, and for making available datasets used in this document.

Christos Papadimitriou acknowledges NSF grant 1408635 “Algorithmic Explorations of Networks, Markets, Evolution, and the Brain”, and NSF grant 1763970 to Columbia University.
Georgios Piliouras acknowledges SUTD grant SRG ESD 2015 097, MOE AcRF Tier 2 Grant 2016-T2-1-170,  grant PIE-SGP-AI-2018-01 and NRF 2018 Fellowship NRF-NRFF2018-07.









\clearpage
\section{Supplementary Material}\label{sec:appendix}


\subsection{Most Closely Related Work}\label{sec:related_work}

We describe related work revolving around Empirical Game Theory analysis (EGTA), discrete-time dynamics models and multi-agent interactions in evolution of cooperation research, and precursors to our new solution concept of \mcc.

The purpose of the first applications of EGTA was to reduce the complexity of large economic problems in electronic commerce, such as continuous double auctions, supply chain management, market games, and automated trading~\cite{Walsh02,Walsh03,Vorobeychik07,Wellman06}. 
While these complex economic problems continue to be a primary application area of these methods~\cite{Wah15,Brinkman16,Wah17,Wang18}, 
the general techniques have been applied in many different settings. These include analysis of interactions among heuristic meta-strategies in poker~\cite{PonsenTKR09}, network protocol compliance~\cite{Wellman13}, collision avoidance in robotics \cite{Hennes13}, and security games~\cite{Prakash15,Wright16,Nguyen17}. 

Evolutionary dynamics have often been presented as a practical tool for analyzing interactions among meta-strategies found in EGTA~\cite{Walsh02,Hennes13,BloembergenTHK15}, and for studying the change in policies of multiple learning agents~\cite{BloembergenTHK15}, as the EGTA approach is largely based on the same assumptions as evolutionary game-theory, viz. repeated interactions among sub-groups sampled independently at random from an arbitrarily-large population of agents. 





From the theoretical biology perspective, researchers have additionally deployed discrete-time evolutionary dynamics models \cite{Nowakbook06}. These models typically provide insights in the macro-dynamics of the overall behavior of agents in strategy space, corresponding to flow rates at the edges of a manifold \cite{Nowak793,Traulsen06a,Traulsen06b,Santos11,Segbroek12}.
These studies usually focus on biological games, the evolution of cooperation and fairness in social dilemma's like the iterated prisoner's dilemma or signalling games, deploying, amongst others, imitation dynamics with low mutation rates \cite{fudenberg2006imitation,veller2016finite}. Similar efforts investigating evolutionary dynamics inspired by statistical physics models have been taken as well \cite{Perc18a,Perc18b}.

In the framework of \emph{the evolution of conventions} \citep{PY1993}, a repeated game is played by one-time players who learn from past plays and are subject to noise and mistakes.
Essentially algorithmic, and in the same line of thought as our formalism, it solves the equilibrium selection problem of \emph{weakly acyclic} games (in our terminology explained in \cref{subsec:MCC}: games whose sink strongly connected components happen to be singletons), and in this special case it aligns very well with our proposed solution concept.
Another equilibrium selection concept related to \mcc\ is the concept of {\em closed under rational behavior (CURB)} set of strategies \citep{Basu}.  
The notion of a {\em sink equilibrium}, defined by \citep{Goemans} for the purpose of exploring new variants of the price of anarchy, is also similar to our \mcc\ --- despite differences in mathematical detail, style, and use.
A method for decomposing games was introduced in \citep{candogan2011flows}, based on properties of game dynamics on the graph of pure strategy profiles by exploiting conceptually similarities to the structure of continuous vector fields.
Researchers have also carried out studies of evaluation metrics in the fields of computer science, machine learning, and artificial intelligence \cite{japkowicz2011evaluating,hernandez2017evaluation,hernandez2017measure}. 
PageRank \cite{page1999pagerank}, an algorithm used for ranking webpages, uses a Markov chain where states are webpages and transitions capture links between these pages; though the Markov chain foundations are related to those used in our work, they are not rooted in an evolutionary dynamical system nor in a game-theoretic solution concept, and as such are quite different to the method presented here which also generalizes across several dimensions.
The Elo rating system has ubiquitously been used for ranking and predicting outcomes in board games \cite{Elo78}, sports \cite{Hvattum2010460,Cattelan13,aldous2017,Sullivan16,Wunderlich18}, and artificial intelligence \cite{mnih:15,silver2018general}. 
This rating system, however, comes with two key limitations \cite{Balduzzi18}: first, it has no predictive power in games with intransitive (cyclic) relations in the set of evaluated agents (e.g., in Rock-Paper-Scissors); second, the rating of a given agent can be artificially inflated by including duplicate copies of weaker agents in the set.


\subsection{Background in Dynamical Systems}
\label{sec:background_dynamics}

\begin{definition}[Flow]
    A flow on a topological space $X$ is a continuous mapping $\phi:\real\times X \rightarrow X$ such that
    \begin{description}
    \item{(i)} $\phi$($t, \cdot$)$: X \rightarrow X$ is a homeomorphism for each $t \in \real$.
    \item{(ii)} $\phi(0, x)=x$ for all $x \in X$.
    \item{(iii)} $\phi(s+t, x)= \phi(s,\phi(t,x))$ for all $s,t \in \real$ and all $x \in X$.
    \end{description}
\end{definition}
The second property is known as the group property of the flows. The topological space $X$ is
 called the phase (or state) space of the flow.  

\begin{definition}
Let $X$ be a set. A map (or discrete dynamical system) is a function $f:X\rightarrow X$.
\end{definition}

Typically, we write $\phi^t({x})$  for $\phi(t,{x})$ and denote a flow $\phi:\real\times X\rightarrow X$  by $\phi^t:X\rightarrow X$, where  the group property appears as $\phi^{t+s}({x})=\phi^s(\phi^t({x}))$ for all ${x} \in X$ and $s,t \in \real$. 
Sometimes, depending on context, we use the notation 
 $\phi^t$ to also signify the map $\phi(t,\cdot)$ for a fixed real number $t$. 
 The map $\phi^1$ is useful to relate the behavior of a flow to the behavior of a map.

\begin{definition}
If $\phi(t,\cdot)$ is a flow on a topological space $X$, then the function $\phi^1$ defines the time-one map of $\phi$.
\end{definition}

Since our state space is compact and the replicator vector field is Lipschitz-continuous, 
we can present
the unique solution of  our ordinary differential equation 
by a flow $\phi:\real\times\mathcal{S}\rightarrow \mathcal{S}$. Fixing starting point $x \in \mathcal{S}$ defines a function of time which captures the trajectory (orbit, solution path) of the system with the given starting point. 
This corresponds to the graph of $\phi(\cdot,x):\real\rightarrow \mathcal{S}$, \textit{i.e.}, the set $\{(t,y): y=\phi(t,x) \text{ for some } t \in \real\}$.

If the starting point $x$ does not correspond to an equilibrium then we wish to capture the asymptotic behavior of the system (informally the limit of $\phi(t,x)$ when $t$ goes to infinity). Typically, however, such functions  do not exhibit a unique limit point so instead we study the set of limits of all
possible convergent subsequences. Formally, given a dynamical system $(\real, \mathcal{S}, \phi)$ with flow $\phi: \mathcal{S} \times \real \rightarrow\mathcal{S}$ and a starting point $x\in \mathcal{S}$, we call point $y\in \mathcal{S}$ an $\omega$-limit point of the orbit through $x$ if there exists a sequence $(t_n)_{n \in \mathbb{N}} \in \real$ such that
$\lim_{n \to \infty} t_n = \infty,~ \lim_{n \to \infty} \phi(t_n, x) = y.$
%
%
%
%
%
Alternatively the $\omega$-limit set can be defined as:
$\omega_{\Phi}(x)=\cap_t \overline{\cup_{\tau\geq t} \phi(\tau, x)}.$
%
%


We denote the boundary of a set $S$ as $\text{\text{bd}(S)}$ and the interior of $S$ as $\text{int}(S)$.
In the case of the replicator dynamics where the state space $S$ corresponds to a product of agent (mixed) strategies we will denote by $\phi_i(x,t)$ the projection of the state on the simplex of mixed strategies of agent $i$.
In our replicator system we embed our state space with the standard topology and the Euclidean distance metric.
%
%
Since our state space 
is compact, 
we can present
the solution of  our system 
as a map $\Phi:\mathcal{S}\times\real\rightarrow \mathcal{S}$ called flow of the system. Fixing starting point $x \in \mathcal{S}$ defines function of time which captures the solution path (orbit, trajectory) of the system with the given starting point. On the other hand, by fixing time $t$, we obtain a smooth map of the state space to itself $\Phi^t:\mathcal{S}\rightarrow\mathcal{S}$. The resulting family of mappings exhibits the standard group properties such as identity $(\Phi^0)$ and existence of inverse $(\Phi^{-t})$, and closure under composition $\Phi^{t_1} \circ \Phi^{t_2} = \Phi^{t_1+t_2}$.

Finally, the boundary of a subset $S$ is the set of points in the closure of $S$, not belonging to the interior of $S$. An element of the boundary of $S$ is called a boundary point of $S$. We denote the boundary of a set $S$ as $\text{\text{bd}(S)}$ and the interior of $S$ as $\text{int}(S)$.

\subsection*{Liouville's Formula}

Liouville's formula can be applied to any system of autonomous differential equations with a continuously differentiable
vector field $\xi$ on an open domain of $\mathcal{S} \subset \real^k$.  The divergence of $\xi$ at $x \in \mathcal{S}$ is defined
as the trace of the corresponding Jacobian at $x$, \textit{i.e.},  $\text{div}[\xi(x)]=\sum_{i=1}^k \frac{\partial \xi_i}{\partial x_i}(x)$. 
Since divergence is a continuous function we can compute its integral over measurable sets $A\subset \mathcal{S}$. Given 
any such set $A$, let $A(t)= \{\Phi(x_0,t): x_0 \in A\}$ be the image of $A$ under map $\Phi$ at time $t$. $A(t)$ is measurable
and is volume is $\text{vol}[A(t)]= \int_{A(t)}dx$. Liouville's formula states that the time derivative of the volume $A(t)$ exists and
is equal to the integral of the divergence over $A(t)$: $\frac{d}{dt} [A(t)] = \int_{A(t)} \text{div} [\xi(x)]dx.$

A vector field is called divergence free if its divergence is zero everywhere. Liouville's formula trivially implies that volume is preserved
in such flows.

Volume preservation is a useful property that allows us to argue about recurrent (i.e., cycle-like) behavior of the dynamics.

\subsection*{Poincar\'{e}'s recurrence theorem}
Poincar\'{e} \citep{Poincare1890}  proved that in certain systems 
almost all trajectories return arbitrarily close to their initial position infinitely often.


\begin{theorem}
\citep{Poincare1890,barreira}
%
If a flow preserves volume and has only bounded orbits then for each open
set there exist orbits that intersect the set infinitely often.
\end{theorem}


\subsubsection{ Poincar\'{e}-Bendixson theorem}

A periodic orbit is called a limit cycle if it is the  $\omega$-limit set of some point not on the periodic orbit.
The Poincar\'{e}-Bendixson theorem allows us to prove the existence
of limit cycles in two dimensional systems.
The main idea is to find  a trapping region, \textit{i.e.},  a region from which trajectories cannot escape. 
If a trajectory enters and does not leave such a closed and bounded region of the state space that contains no equilibria then this trajectory must approach a periodic orbit as time goes to infinity.
Formally, we have:

\begin{theorem}
\citep{bendixson1901courbes,teschl2012ordinary}
Given a differentiable real dynamical system defined on an open subset of the plane, then every non-empty compact $\omega$-limit set of an orbit, which contains only finitely many fixed points, is either
a fixed point, a periodic orbit, or a connected set composed of a finite number of fixed points together with homoclinic and heteroclinic orbits connecting these.
\end{theorem}

\subsection*{Homeomorphisms and conjugacy of flows}


A function $f$ between two topological spaces is called a \textit{homeomorphism} if it has the following properties:
$f$ is a bijection, $f$ is continuous, and $f$ has a continuous inverse. 
 A function $f$ between two topological spaces is called a \textit{diffeomorphism} if it has the following properties:
$f$ is a bijection, $f$ is continuously  differentiable, and $f$ has a continuously differentiable inverse. 
 %
  Two flows $\Phi^t : A \rightarrow A$ and $\Psi^t : B \rightarrow B$ are conjugate if there exists a homeomorphism $g : A \rightarrow B$ such that for each $x \in A$ and $t \in \real$:
$g(\Phi^t (x)) = \Psi^t (g(x)).$
%
%
Furthermore, two flows $\Phi^t:A\rightarrow A$ and $\Psi^t :B\rightarrow B$ are \textit{diffeomorhpic} if there exists a diffeomorphism $g : A \rightarrow B$ such that for each $x \in A$ and $t \in \real$
$g(\Phi^t (x)) = \Psi^t (g(x))$. 
 If two flows are diffeomorphic then their vector fields are related by the derivative of the conjugacy. 
 That is, we get precisely the  same result that we would have obtained if we simply transformed the coordinates in their differential equations~\citep{Meiss2007}.

\subsection*{Stability of sets of states}

Let $A\subset X$ be a closed set. We define a set $O\subset X$ a neighborhood of $A$ if it is open relative to $X$ and contains $A$. We say that $A$ is (Lyapunov) stable if for every neighborhood $O$ of $A$ there exists a neighborhood $O'$ of $A$ such that every trajectory that starts in $O'$ is contained in $O$, i.e., of $x(0)\in O'$ then $x(t)\in O$ for all $t\geq0$.  Set $A$ is attracting if there exists a neighborhood $S$ of $A$ such that every trajectory starting in $S$ converges to $A$. A set is called asymptotically stable if it is both Lyapunov stable and attracting. 

\begin{definition}[Sink chain recurrent points]
    Chain recurrent points that belong to a sink chain component are called sink chain recurrent points.
\end{definition}

\begin{definition}[Lyapunov stable set]
    Let $\phi$ be a flow on a metric space $(X,d)$. A set $A \subset X$ is Lyapunov stable if for every neighborhood $O$ of $A$ there exists a neighborhood $O'$ of $A$ such that every trajectory that starts in $O'$ is contained in $O$; i.e., if $x\in O'$ then $\phi(t,x)\in O$ for all $t\geq0$. 
\end{definition}

\begin{definition}[Attracting set]
    Set $A$ is attracting if there exists a neighborhood $S$ of $A$ such that every trajectory starting in $S$ converges to $A$.
\end{definition}

\begin{definition}[Asymptotically stable set]
    A set is called asymptotically stable if it is both Lyapunov stable and attracting. 
\end{definition}

\subsubsection{Multi-population Replicator Dynamics}\label{sec:rd_multipop}
For a $K$-player NFG, one may use a set of $K$ populations with $S^k$ denoting the finite set of pure strategies available to each agent in population $k \in \{1, \ldots, K\}$.
The mass of players in a given population $k$ that use strategy $i \in S^k$ is denoted $x^k_i$, where $\sum_{i \in S^k} x^k_i = 1$. 
Let $S$ denote the set of all populations' pure strategy profiles, and $x$ represent the joint population state.
Let the payoff matrix for a given population $k$ be denoted $M^k: S \rightarrow \mathbb{R}$.
The fitness of an agent in population $k$ playing pure strategy $i$ given state $x$ is then,
\begin{align}
    f^k_i(x) = \sum_{s^{-k} \in S^{-k}}M^{k}(i,s^{-k})\prod_{c \neq k}x^c_{s^{c}}.
\end{align}
Namely, the fitness is the expected payoff the agent playing strategy $i$ receives given every competitor population's state $x^c$.
The $k$-th population's average fitness given state $x$ is then,
\begin{align}
    \bar{f}^k(x) = \sum_i f^k_i(x) x^k_i,
\end{align}
with the corresponding $K$-population replicator dynamics,
\begin{align}
    \dot{x}^k_i = x^k_i(f^k_i(x) - \bar{f}^k(x)) \qquad \forall k \in \{1,\ldots,K\} \quad \forall i \in s^k.
\end{align}

\subsection{Single-population discrete-time model}\label{sec:supsingle}

We have a set of strategies (or agents under evaluation) $S=\{s_1,..,s_n\}$, with $|S|=n$, which we would like to evaluate for their evolutionary strength. We also have a population of individuals $A=\{a_1,..,a_m\}$, with $|A|=m$, that are programmed to play a strategy from  the set $S$. Individuals interact pairwise through empirical games.




We start from a finite well-mixed population of $m$ individuals, in which $p$ individuals are playing $\focalAgent$. 
At each timestep $t$ we randomly choose two individuals $\focalAgent$ and $\residentAgent$, with respective strategies $s_{\focalAgent}$ and $s_{\residentAgent}$.
The strategy of individual $\focalAgent$ is then updated by either probabilistically copying the strategy $s_{\residentAgent}$ of individual $\residentAgent$ it is interacting with, mutating with a very small probability into another strategy, or sticking with its own strategy $s_{\focalAgent}$.
The idea is that strong individuals will replicate and spread throughout the population.
The probability with which individual $\focalAgent$ (playing $s_{\focalAgent}$) will copy strategy $s_{\residentAgent}$ from individual $\residentAgent$ can be described by a \emph{selection} function $\mathbb{P}(\focalAgent \rightarrow \residentAgent)$, which governs the dynamics of the finite-population model.

Individual $\focalAgent$ will thus copy the behavior of individual $\residentAgent$ with probability $p_{\focalAgent \rightarrow \residentAgent}$ and stick to its own strategy with probability $1-\mathbb{P}(\focalAgent \rightarrow \residentAgent)$.
We denote the probability for a strategy to mutate randomly into another strategy $s \in S$ by $\mu$ and we will assume it to be infinitesimally small, (i.e., we consider a small-mutation limit $\mu \rightarrow 0$). 
If we neglected mutations, the end state of this evolutionary process would be monomorphic. 
If we introduce a very small mutation rate this means that either the mutant fixates and takes over the current population, or the current population is capable of wiping out the mutant strategy \cite{fudenberg2006imitation}.
Therefore, given a small mutation rate, the mutant either fixates or disappears before a new mutant appears. 
This means that the population will never contain more than two strategies at any point in time.

We now proceed as follows. At any moment in time when two strategies ($s_{\focalAgent}$ and $s_s$) are present in the population, we can calculate the fitness of an individual $\focalAgent$ playing strategy $s_{\focalAgent}$ in a population of $p$ individuals playing $s_{\focalAgent}$ and $m-p$ individuals playing $s_s$. 
Fitnesses may be calculated using either knowledge of the global population state (i.e., where every individual is aware of the number of other individuals playing each strategy, which may be a strong assumption) or local knowledge (i.e., only the current opponent's strategy) \cite{traulsen2005coevolutionary}. 
The corresponding fitness for the local  case, which we focus on here, is $f(\focalAgent,\residentAgent) = M_{\focalAgent,\residentAgent}$,
where $M_{\focalAgent,\residentAgent}$ is obtained from the meta-game payoff matrix. 
Analogously, the simultaneous payoff of an individual $\residentAgent$ playing $s_{\residentAgent}$ against $s_{\focalAgent}$ is $f(\residentAgent,\focalAgent) = M_{\residentAgent,\focalAgent}$,
For the remainder of the paper, we focus on the logistic selection function (aka Fermi distribution),
\begin{equation}\label{eq:fermi_distr}
    \mathbb{P}(\focalAgent \rightarrow \residentAgent) = \frac{e^{\alpha f(\residentAgent,\focalAgent)}}{e^{\alpha f(\focalAgent,\residentAgent)}+e^{\alpha f(\residentAgent,\focalAgent)}}=(1 + e^{\alpha(f(\focalAgent,\residentAgent)-f(\residentAgent,\focalAgent))})^{-1},
\end{equation}
with $\alpha$ determining the selection intensity. 
While the subsequent empirical methodology extends to general selection functions, the choice of Fermi selection function enables closed-form characterization of certain properties of the discrete-time model.

Based on this setup, we define a Markov chain over the set of strategies $S$ with $n$ states. Each state represents a monomorphic population end-state, corresponding to one of the strategies $s_{\focalAgent}$ with $\focalAgent \in \{1,..,n\}$. 
The transitions between these states are defined by the corresponding fixation probabilities when a mutant strategy is introduced in a monomorphic population. 
The stationary distribution over this Markov chain will tell us how much time on average the dynamics will spend in each of the monomorphic states.


Considering our set $S$ of $n$ strategies, we define the Markov chain with $n^2$ transition probabilities over the monomorphic states.
Let $\eta=\frac{1}{n-1}$ and denote by $\rho_{\residentAgent,\focalAgent}$ the probability of mutant strategy $s_{\focalAgent}$ fixating (taking over) in a resident population of individuals playing $s_{\residentAgent}$. 
So $\eta\rho_{\residentAgent,\focalAgent}$ is the probability that a population which finds itself in state $s_{\residentAgent}$ will end up in state $s_{\focalAgent}$ after the occurrence of a single mutation. 
This yields the following Markov transition matrix,
\begin{equation}
    C = \left(\begin{matrix}
        1-\eta(\rho_{1,2}+\rho_{1,3}+...+\rho_{1,n}) & \eta\rho_{1,2} & ... & \eta\rho_{1,n}\\
        \eta\rho_{2,1} & 1-\eta(\rho_{2,1}+\rho_{2,3}+...+\rho_{2,n}) & ... & \eta\rho_{2,n}\\
        ... & ... & ... & ...\\
        \eta\rho_{n,1} & ... & ... & 1-\eta(\rho_{n,1}+\rho_{n,2}+...+\rho_{n,n-1})
        \end{matrix}\right)
\end{equation}

The fixation probabilities $\rho_{\residentAgent,\focalAgent}$ can be calculated as follows. 
Assume we have a population of $p$ individuals playing $s_{\focalAgent}$ and $m-p$ individuals playing $s_{\residentAgent}$. The probability that the number of type $s_{\focalAgent}$ individuals decreases/increases by one is given by,
\begin{equation}
    T^{(\mp 1)}(p,\focalAgent,\residentAgent)= \frac{p(m-p)}{m(m-1)}\left(1+e^{\pm \alpha(f(\focalAgent,\residentAgent)- f(\residentAgent,\focalAgent))}\right)^{-1}.
\end{equation}
Now we can compute the fixation probability $\rho_{\residentAgent,\focalAgent}$ of a mutant with strategy $s_{\focalAgent}$ in a population of $m-1$ individuals programmed to playing $s_{\residentAgent}$ as follows, 
\begin{align}
\rho_{\residentAgent,\focalAgent} &=  \left(1+\sum_{l=1}^{m-1}\prod_{p=1}^l \frac{T^{(-1)}(p,\focalAgent,\residentAgent)}{T^{(+1)}(p,\focalAgent,\residentAgent)}\right)^{-1}\\   
&=\left(1+\sum_{l=1}^{m-1}\prod_{p=1}^l e^{-\alpha(f(\focalAgent,\residentAgent)- f(\residentAgent,\focalAgent))} \right)^{-1}\label{eq:fixation_prob}
\end{align}
This corresponds to the computation of an $m$-step transition in the Markov chain \cite{TaylKarl98}.
The quotient $\frac{T^{(-1)}(p,\focalAgent,\residentAgent)}{T^{(+1)}(p,\focalAgent,\residentAgent)}$ expresses the likelihood (odds) that the mutation process continues in either direction: 
if it is close to zero then it is very likely that the number of mutants $s_{\focalAgent}$ increases;
if it is very large it is very likely that the number of mutants will decrease;
and if it close to one then the probabilities of increase and decrease of the number of mutants are equally likely.

\begin{property}\label{property:unique_pi_singlepop}
    Given finite payoffs, fixation probabilities $\rho_{\residentAgent,\focalAgent}$ under the Fermi imitative protocol \eqref{eq:fermi_distr} are positive for all $\residentAgent$ and $\focalAgent$; 
    i.e., any single mutation can cause a transition from any state to another. 
    Markov chain $C$ is, therefore, irreducible, and a unique stationary distribution $\pi$ (where $\pi^TC= \pi^T$ and $\sum_i \pi_i = 1$) exists.
\end{property}
This unique $\pi$ provides the evolutionary ranking, or strength of each strategy in the set $S$, expressed as the time the population spends in each state in distribution $\pi$.
This single population model has been widely studied (see, e.g., \cite{Nowak793,Traulsen06a,Traulsen06b,Segbroek12,Santos11}), both theoretically and empirically, but is limited to both pairwise interactions and symmetric games.


\subsection{Proofs}\label{sec:proofs}

\subsubsection{Proof of \cref{property:unique_pi}}\label{sec:proof_unique_pi}
\UniquePi*
\begin{proof}
    Consider any two states (i.e., strategy profiles) $s_i$ and $s_j$ of the $K$-population Markov chain with transition matrix \cref{eqn:transition_matrix_C_multipop}. 
    Under finite payoffs $f^k(\focalAgent,p)$ and $f^k(\residentAgent,p)$, fixation probabilities $\rho^k_{\residentAgent,\focalAgent}$ under the Fermi imitative protocol \eqref{eq:fermi_distr} are positive.
    Let $R(s,k,s_x)$ denote the operation of replacing the $k$-th strategy in a strategy profile $s$ with a different pure strategy $s_x$.
    Thus, state $s_j$ is accessible from any state $s_i$ (namely, consider the chain $\{s_0=s_i,s_1=R(s_0,1,s_j^1),s_2=R(s_1,2,s_j^2)\ldots,s_{K-1}=R(s_{K-2},K-1,s_j^{K-1}),s_{K}=R(s_{K-1},K,s_j^{K})=s_j\}$ connecting strategies $s_i$ to $s_j$ with non-zero probability).
    The Markov chain is, therefore, irreducible and a unique stationary distribution exists.
\end{proof}

\subsubsection{Proof for \cref{thm:discrete_cont_edge_dyn}}\label{sec:proof_discrete_cont_edge_dyn}
\DiscreteContEdgeDyn*
\begin{proof}
    To simplify notation, we prove the theorem for the single-population case without loss of generality.
    Let $x_{i}(t)$ represent the fraction of individuals in the population that are playing strategy $s_i$ at timestep $t$.
    Rather than consider the underlying stochastic evolutionary equations directly, we consider the \emph{mean dynamics}.
    {An alternative proof path for the single population case is presented in \cite{traulsen2005coevolutionary} and may be applied here as well.}
    The mean dynamics constitute a deterministic process governing the expected evolution of state $x_{i}(t)$, and provide a close approximation of the underlying system over finite time spans under a large-population limit \cite[Chapters 4 and 10]{Sandholm10}.
    For a general finite population game, the mean dynamics correspond to the difference of the expected influx and outflux of individuals playing a strategy $i$ against individuals playing any strategy $j \in S$ given the underlying selection function $\mathbb{P}(i \rightarrow j)(x)$,
    \begin{align}
        \dot{x}_{i}(t) &= \sum_{j \in S}x_j x_i \mathbb{P}(j \rightarrow i)(x) - x_i\sum_{j \in S}x_j \mathbb{P}(i \rightarrow j)(x).
    \end{align}
    Under the low-mutation rate assumption, the finite population model considers only the transitions between pairs of monomorphic states $s_{\focalAgent}$ and $s_\residentAgent$, where $x_{\focalAgent} + x_{\residentAgent} = 1$. This yields simplified mean dynamics, 
    \begin{align}
        \dot{x}_{\focalAgent} &= x_{\residentAgent}x_{\focalAgent}\mathbb{P}(\residentAgent \rightarrow \focalAgent) - x_{\focalAgent}x_{\residentAgent}\mathbb{P}(\focalAgent \rightarrow \residentAgent)(x_\focalAgent)\\
        &= (1-x_{\focalAgent})x_{\focalAgent} \left[\mathbb{P}(\residentAgent \rightarrow \focalAgent) - \mathbb{P}(\focalAgent \rightarrow \residentAgent) \right]\label{eq:discrete_infpop_dyn_branch1}\\
        &= x_{\focalAgent}\left[\mathbb{P}(\residentAgent \rightarrow \focalAgent) - \left(x_{\focalAgent}\mathbb{P}(\residentAgent \rightarrow \focalAgent) + (1-x_{\focalAgent})\mathbb{P}(\focalAgent \rightarrow \residentAgent)\right) \right]\\
        &= x_{\focalAgent}\left[\mathbb{P}(\residentAgent \rightarrow \focalAgent) - \bar{p} \right]\label{eq:discrete_infpop_dyn}
    \end{align}
    where, 
    \begin{align}
        \mathbb{P}(\focalAgent \rightarrow \residentAgent) &= (1+e^{\alpha(f(\focalAgent,\residentAgent)-f(\residentAgent,\focalAgent))})^{-1}\\
        \mathbb{P}(\residentAgent \rightarrow \focalAgent) &= (1+e^{-\alpha(f(\focalAgent,\residentAgent)-f(\residentAgent,\focalAgent))})^{-1}\\
        \bar{p} &= x_{\focalAgent}\mathbb{P}(\residentAgent \rightarrow \focalAgent) + (1-x_{\focalAgent})\mathbb{P}(\focalAgent \rightarrow \residentAgent)
    \end{align}
    We, therefore, observe that the discrete large-population mean dynamics \eqref{eq:discrete_infpop_dyn} correspond to the replicator equations (with the caveat that Fermi revision protocol takes the place of the usual fitness terms).
    
    Moreover, one can branch off after \cref{eq:discrete_infpop_dyn_branch1} to yield,
    \begin{align}    
        \dot{x}_{\focalAgent} &= (1-x_{\focalAgent})x_{\focalAgent} \left[\mathbb{P}(\residentAgent \rightarrow \focalAgent) - \mathbb{P}(\focalAgent \rightarrow \residentAgent) \right]\\
        &= x_{\focalAgent}(1-x_{\focalAgent}) \left[\left(1+e^{- \alpha(f(\focalAgent,\residentAgent)- f(\residentAgent,\focalAgent))}\right)^{-1} - \left(1+e^{\alpha(f(\focalAgent,\residentAgent)- f(\residentAgent,\focalAgent))}\right)^{-1} \right]\\
        &= x_{\focalAgent}(1-x_{\focalAgent})\tanh{\frac{\alpha(f(\focalAgent,\residentAgent)- f(\residentAgent,\focalAgent))}{2}}
    \end{align}
    which matches the It\^{o} calculus based derivation of \cite{Traulsen06a} under the large-population limit.
\end{proof}

\subsubsection{Proof for \cref{thm:partial_order_chain_components}}\label{sec:proof_partial_order_chain_components}
We start by introducing the notion of chain transitivity, which will be useful in the proof of the theorem.
\begin{definition}[Chain transitive]
    Let $\phi$ be a flow on a metric space $(X,d)$. A set $A \subset X$ is
    chain transitive with respect to $\phi$ if for any $x,y \in A$ and any $\epsilon>0$ and $T>0$ there exists an $(\epsilon,T)$-chain from $x$ to $y$.
\end{definition}

Next we state the following properties of chain components,
\begin{property}[\cite{alongi2007recurrence}]\label{thm:CCs}
    Every chain component of a flow on a compact metric space is closed, connected, and invariant with respect of the flow. Moreover, 
    \begin{itemize}
        \item Every chain component of a flow on a metric space is chain transitive with respect to the flow.
        \item Every chain transitive set with respect to a flow on a metric space is a subset of a unique chain component of the flow.
        \item If $A$ and $B$ are chain transitive with respect to a flow on a metric space, $A \subset B$ and $C$ is the unique chain component containing $A$, then $B \subset C$.
    \end{itemize}
\end{property}

\PartialOrderChainComponents*

\begin{proof}
We will show that the binary relation $\leq_C$ is reflective, antisymmetric and transitive.
\begin{itemize}
\item $A_1 \leq_C A_1$. Since any chain component is chain transitive then we have that for any $x,y \in A_1$: $x\sim y$.
\item If $A_1 \leq_C A_2$ and $A_2 \leq_C A_1$ then $A_1=A_2$. By chain transitivity of $A_1$, $A_2$ we have that for any $x,x' \in A_1$, $x\sim x'$ and for any $y,y' \in A_2$, $y\sim y'$. Hence if $x\sim y$ then $x'\sim y'$ for any $x'\in A_1$ and any $y' \in A_2$. Hence, $A_1 \cup A_2$ is a chain transitive set and thus by Theorem \ref{thm:CCs} must be a subset of a unique chain component of the flow $C$ such that $A_1 \cup A_2\subset C$. However, we assumed that $A_1, A_2$ are chain components themselves. Thus, $A_1=A_1 \cup A_2=A_2$.
\item  If $A_1 \leq_C A_2$ and $A_2 \leq_C A_3$ then $A_1 \leq_C A_3$. If there exist $x \in A_1$ and $y \in A_2$ such that $x \in \Omega^{+}(\phi,y)$, as well as $y' \in A_2$ and $z \in A_3$ such that $y' \in \Omega^{+}(\phi,z)$ then by chain transitivity of $A_2$ we have that $y \in \Omega^{+}(\phi,y')$ and thus $x \in \Omega^{+}(\phi,z)$, implying $A_1 \leq_C A_3$.
\end{itemize}
\end{proof}

\subsubsection{Proof for \cref{thm:MCC}}\label{sec:proof_MCC}

We first present several results necessary for the proof.
\begin{lemma}
\label{lem:CR1}
A chain recurrent (CR) point $x$ is a sink CR point if and only if for any CR point $y$ if $y\in \Omega^+(\phi,x)$ then $x\in \Omega^+(\phi,y)$, i.e., the two points are chain equivalent.
\end{lemma}

\begin{proof}
We will argue the forward direction by contradiction.
Suppose not. That is, suppose that $x$ is a sink CR point and there exists a CR point $y$ such that $y\in \Omega^+(\phi,x)$ and $x\notin \Omega^+(\phi,y)$, then if $C_x, C_y$ are the equivalence classes/chain components of $x, y$ respectively we have that 
$C_y \leq_C C_x$ and $C_y, C_x$ are clearly distinct chain components since $x\notin \Omega^+(\phi,y)$. Thus, $C_x$ is not a sink chain component and $x$ is not a sink chain recurrent point, contradiction.

For the reverse direction,  once again by contradiction we have that for any CR point $y$ with $y\in \Omega^+(\phi,x)$, $x\in \Omega^+(\phi,y)$ and $x$ is a non-sink CR point. Then there exists another chain component $A$ with $A\leq_C C_x$ where
$C_x$ is the equivalence class/chain component of $x$. Hence, there exists $y \in A$ such that $y\in \Omega^+(\phi,x)$. Since $y$  is a  CR point which does not belong to $C_x$, we have $x\notin \Omega^+(\phi,y)$, contradiction.
\end{proof}

\begin{lemma}
\label{lem:reachable}
If a sink chain component contains a single vertex $s_i$ then it contains any vertex $s_j$ which is reachable from $s_i$ via (weakly)-better response moves. Specifically, it contains an MCC. 
\end{lemma}

\begin{proof}
Any state/vertex $s_j$ is a chain recurrent (CR) point because it is a fixed point of the replicator dynamics. If $s_j$ is reachable by $s_i$ via a weakly-better response path, then $s_j \in \Omega^+(\phi,s_i)$ for the replicator flow. In the case of edges that are strictly improving it suffices to use the $\epsilon$ correction to introduce the improving strategy and replicator will converge to the better outcome. In the case of edges between outcomes of equal payoff all convex combinations of these strategies are fixed points for the replicator and we can traverse this edge with $\lceil1/\epsilon\rceil$ hops of size $\epsilon$.

 But if $s_j \in \Omega^+(\phi,s_i)$ and $s_i$  is a sink CR point (since it belongs to a CR component) then by Lemma \ref{lem:CR1}   $s_i \in \Omega^+(\phi,s_j)$. Therefore, state/vertex $s_j$ also belongs to the same sink chain component. The set of reachable vertices includes a strongly connected component with no outgoing edges and thus a \mcc.

\end{proof}

\TheoremMCC*
\begin{proof}
    Since solutions in the neighborhood of an asymptotically stable set, all approach the set, volume is contracted in this neighborhood, however, replicator dynamics is volume preserving in the interior of the state space \citep{Sandholm10,Soda14,PiliourasAAMAS2014}; the formal argument works by transforming the system induces by the replicator dynamics over the interior of the state space into a conjugate dynamical system that is divergence-free. 
    Any asymptotically stable set cannot lie in the interior of the simplex, i.e., it cannot consist only of fully mixed strategies. Hence, there must exist some product of subsimplices with a non-empty intersection with this set. The intersection of the original asymptotically stable set with this subspace is still asymptotically stable for this invariant subspace and thus we can continue the argument inductively. 
    {The intersection of the attracting neighborhood with this subspace is an attracting neighborhood for the dynamics on this invariant subspace.}
    We deduce that any asymptotically stable chain component must contain at least one vertex of the simplex (pure strategy profile). Let $s_i$ be this vertex. By Lemma \ref{lem:reachable}, this sink chain component must also include all other vertices reachable from $s_i$ via weakly-better replies. Specifically, it must include at least one \mcc. 
    Finally, a \mcc\ is a chain transitive set for the replicator flow via the same argument of the $\epsilon$ hops as in Lemma \ref{lem:reachable}. By Theorem \ref{thm:CCs}, it is a subset of a unique chain component of the flow.
\end{proof}

\subsubsection{Proof for \cref{thm:inf_pop_alpha}}\label{sec:proof_inf_pop_alpha}
\InfinitePopAlpha*
\begin{proof}
    Recall from the MCC definition that the probability of strictly improving responses for all players are set equal to each other, and transitions between strategies of equal payoff happen with a smaller probability also equal to each other for all players.
    Let the ratio of the two probabilities be denoted $\epsilon$ for all players.
    The transition probabilities of the Markov chain of the macro-model when taking taking the limit of $\alpha\rightarrow \infty$ are equal to the transitions probabilities of the Markov chain of the Markov-Conley chains when setting $\epsilon$ equal to $\frac{1}{m}$, where $m$ is the size of the population in the macro-model. Let $A_{s_i}(k)$ be the number of strictly improving moves for player $k$ in state/vertex $s_i$. Similarly, let $B_{s_i}(k)$ be the number of deviating moves for player $k$ in state/vertex $s_i$ that do not affect her payoff.
    It suffices to set the probability of a node $s_i$ self-transitioning equal to $1-\frac{\sum_k A_{s_i}(k) + \epsilon \sum_k B_{s_i}(k) }{\sum_k (|S^k|-1)}$.
\end{proof}


\end{document}